\documentclass[aps,prl,twocolumn,superscriptaddress,floatfix,nofootinbib,showpacs,longbibliography]{revtex4-2}
\usepackage[T1]{fontenc}
\usepackage{times}
\usepackage{epsfig}
\usepackage{amsfonts}
\usepackage{amsmath}
\usepackage{slashbox}
\usepackage{amssymb,amsthm}
\usepackage{xcolor,colortbl}
\usepackage{multirow}
\usepackage{braket}
\newcommand{\ketbra}[1]{\ket{#1}\bra{#1}}
\usepackage{latexsym}
\usepackage{amsfonts}
\usepackage{mathrsfs}
\usepackage{natbib}
\usepackage{verbatim}
\usepackage{gensymb}
\usepackage{caption}
\usepackage{subcaption}
\usepackage{ragged2e}
\usepackage{oplotsymbl}
\usepackage{array}
\usepackage{appendix}
\usepackage{caption}

\usepackage{tikz}

\usepackage{soul}
\DeclareCaptionJustification{justified}{\justifying}
\usepackage{blkarray}
 \usepackage{graphicx}
 \usepackage{braket}
\newtheorem{theorem}{Theorem}
\newtheorem{corollary}{Corollary}

\newtheorem{observation}{Observation}
\usepackage[colorlinks=true,linkcolor=blue,citecolor=magenta,urlcolor=blue]{hyperref}
\allowdisplaybreaks

\newtheorem{Lemma}{Lemma}

\newcommand{\Tr}{\operatorname{Tr}}

\begin{document}

\title{Repeater-Based Quantum Communication Protocol: Maximizing Teleportation Fidelity with Minimal Entanglement}

\author{Arkaprabha Ghosal}
\thanks{These authors contributed equally to the manuscript.}
\affiliation{Optics \& Quantum Information Group, The Institute of Mathematical Sciences, CIT Campus, Taramani, Chennai 600113, India.}
\affiliation{Homi Bhabha National Institute, Training School Complex, Anushakti Nagar, Mumbai 400085, India.}
\affiliation{Mathematics Department, University of New Orleans, Louisiana, USA.}

\author{Jatin Ghai}
\thanks{These authors contributed equally to the manuscript.}
\affiliation{Optics \& Quantum Information Group, The Institute of Mathematical Sciences, CIT Campus, Taramani, Chennai 600113, India.}
\affiliation{Homi Bhabha National Institute, Training School Complex, Anushakti Nagar, Mumbai 400085, India.}

\author{Tanmay Saha}
\affiliation{Optics \& Quantum Information Group, The Institute of Mathematical Sciences, CIT Campus, Taramani, Chennai 600113, India.}
\affiliation{Homi Bhabha National Institute, Training School Complex, Anushakti Nagar, Mumbai 400085, India.}

\author{Sibasish Ghosh}
\affiliation{Optics \& Quantum Information Group, The Institute of Mathematical Sciences, CIT Campus, Taramani, Chennai 600113, India.}
\affiliation{Homi Bhabha National Institute, Training School Complex, Anushakti Nagar, Mumbai 400085, India.}

\author{Mir Alimuddin}   
\thanks{These authors contributed equally to the manuscript.}
\affiliation{ICFO-Institut de Ciencies Fotoniques, The Barcelona Institute of Science and Technology, Av. Carl Friedrich Gauss 3, 08860 Castelldefels (Barcelona), Spain.}

\begin{abstract}
Transmitting unknown quantum states to distant locations is crucial for distributed quantum information protocols. The seminal quantum teleportation scheme achieves this feat while requiring prior maximal entanglement between the sender and receiver. In scenarios with noisy entangled states, optimal teleportation fidelity characterizes the efficacy of transmitting the state, demanding the proper selection of local operations at the sender's and receiver's ends. The complexity escalates further in long-range communication setups, prompting the consideration of a repeater-based approach, which incorporates arrays of nodes with multiple segments to facilitate the efficient transmission of quantum information. The fidelity of the communication line gets degraded even if a single segment is affected by noise. In such cases, the general wisdom employs the standard entanglement swapping protocol involving maximally entangled states across the noiseless segments and applying maximally entangled basis measurement at the corresponding nodes to achieve optimal fidelity. In this Letter, we propose a more efficient protocol for a certain class of noisy states, achieving the same fidelity as the standard protocol while consuming less amount of entanglement. Our approach ensures optimal teleportation fidelity even when the end-to-end state gets noisier, and thus promises efficient utility of quantum resources in repeater-based distributed quantum protocols.    
\end{abstract}
\maketitle

{\it Introduction.--} In distributed communication scenarios, the successful implementation of most information processing tasks critically depends on the relentless transmission of signals across distant points \cite{Wilde2013, Marinescu2012, Khatri2024}. However, in the domain of long-range communication, signal quality often deteriorates as the separation between sender and receiver increases \cite{Briegel1998,Dur1999}. To address this crucial challenge, potential remedy involves a repeater-based approach, wherein the communication line is divided into multiple segments or short-range communication lines \cite{Sangouard2011, Azuma2023}. Signals received from each segment are then replicated, amplified, and relayed to neighboring stations (nodes) until they reach the intended endpoints. However, conventional signal replication is rendered unfeasible in quantum information processing due to the {\it no-cloning} principle \cite{Wootters1982}. Nonetheless, thanks to the advent of quantum teleportation protocols \cite{Bennett1993}, the way has been paved for the development of a viable quantum repeater-based protocol through entanglement swapping \cite{Zukowski1993, Pan1998, Bouda2001, Bayrakci2022}. This approach provides the most effective and scalable solution for achieving long-range quantum communication, as it mitigates signal degradation most compelling over extended distances \cite{Briegel1998,Dur1999,Sangouard2011,Azuma2023}. Noteworthy, the protocol requires equipping each segment with maximally entangled states (MES) and performing maximally entangled measurements (MEM) at intermediary nodes to facilitate seamless teleportation from sender to receiver.

However, establishing maximal entanglement poses significant challenges, and environmental action often results in noisy entangled states \cite{Horodecki2009}. Addressing this challenge involves various strategies, one of the kind which entails distilling MES from asymptotic copies of the identical noisy entangled states through local operations and classical communication (LOCC) \cite{Briegel1998,Dur1999,Sangouard2011, Azuma2023, Bennett1996, BennettBrassard1996, Deutsch1996, Razavi2009, Broadfoot2010,Pal2014}. Although the asymptotic approach provides a fundamental theoretical limit for quantum resources, practical scenarios often demand considering finite copies of resources, making such considerations highly relevant and necessary \cite{Lo2001,Morikoshi2001,Martin2003,Hayashi2006,Buscemi2010,Buscemi2013,Brandao2011,Datta2015,Fang2019,Regula2019}. The distribution of quantum entanglement among distant parties has also garnered significant interest in the finite-copy regime \cite{Hardy2000,Gour2004, Acin2007,Perseguers2008,Cuquet2009,Perseguers2010,Meng2021, Riera2021}. In practice, when repeater segments share even a single non-maximally entangled resource, achieving optimal entanglement distribution between the desired parties becomes particularly challenging due to the complexity of non-universal LOCC protocols. Because of the monotonicity of entanglement under LOCC, one can always establish an upper bound for entanglement measures, which corresponds to the least entangled state shared across the segments.

 However, achieving this upper bound is also quite nontrivial, and no entanglement measure is known to achieve this in a typical repeater scenario. For a given noisy segment with a noisy entangled state, a trivial achievability condition arises for all entanglement measures if all other segments are equipped with MES resources. This effectively reduces the problem to an entanglement-swapping (E-SWAP) protocol, recreating the same noisy entangled state between the desired parties.
Under specific conditions, such as when each segment shares the same pure non-maximally entangled (NMES) resource and the scenario involves a single intermediary node, the singlet conversion probability (SCP) of the distributed entangled state equals the SCP of the segment's entangled state\cite{Hardy2000,Acin2007}. However, in multi-node quantum repeater setups, this question remains unresolved. Notably, other entanglement measures, such as concurrence \cite{Gour2004} and worst-case entanglement (WCE)\cite{Perseguers2008}, are strictly less than the entanglement of the noisy state shared in a segment. The ability to achieve the optimal bound under nontrivial conditions (with non-maximally entangled resources) for a given entanglement measure holds significant practical importance. Such an achievement would reduce entanglement consumption compared to the traditional E-SWAP protocol, making it a compelling advancement in quantum communication.

This Letter directly addresses the concerns outlined above. We demonstrate that teleportation of an unknown quantum state, quantified by an entanglement measure known as the optimal teleportation fidelity, exhibits above phenomena while the repeater have a certain class of two-qubit noisy states in it's noisy segment. The protocol utilizes non-maximal resources in free-segments and nodes, resulting in lower resource consumption compared to the standard E-SWAP protocol. Furthermore, in a given quantum repeater scenario, as the noisy segment shares a noisier state, the resource consumption in the repeater reduces accordingly. Notably, this advantage persists even if the number of nodes approaches infinity.  

{\it Optimal teleportation fidelity and Optimal fully entangled fraction.--} Teleportation of an unknown quantum state $\ket{\psi}$ using a shared quantum state $\rho_{AB}$ between the sender (Alice) and receiver (Bob) is typically achieved via an LOCC protocol $\mathcal{T}$. The performance of teleportation is characterized by the average teleportation fidelity, defined as $f_{\mathcal{T}}(\rho_{AB})=\int d\psi ~\bra{\psi} \rho_{\psi}\ket{\psi}$, where $\rho_{\psi}$ represents the teleported state at Bob's end. For a given state $\rho_{AB}$, the teleportation fidelity explicitly depends on the chosen protocol $\mathcal{T}$. The optimal teleportation fidelity (OTF) is then defined as \cite{Horodecki1999}:
\begin{align}
    f^{\star}(\rho_{AB}) = \max_{\mathcal{T}}f_{\mathcal{T}}(\rho_{AB}). \label{OTF}
\end{align} 
Another significant quantity of interest is the  optimal fully entangled fraction (OFEF), 
\begin{align}
F^{\star}(\rho_{AB}) = \max_{\Lambda\in \Lambda_{LOCC}} F(\Lambda (\rho_{AB})), \label{OFEF}
\end{align}
where $F(X) = \max_{\ket{\phi}} \bra{\phi} X \ket{\phi}$, is the fully entangled fraction \cite{Horodecki1999, Badzia2000}, and $\ket{\phi}$ represents an arbitrary $d\times d$ MES. The LOCC optimization in the definition of $F^\star$ makes it an entanglement monotone, but it also presents computational challenges \cite{Verstraete2001,Verstraete2002,Verstraete2003}, similar to determining the OTF. For arbitrary $d \times d$ states, these two measures are related by \cite{Horodecki1999}: $f^{\star}(\rho_{AB})=(d~F^{\star}(\rho_{AB})+1)/(d+1)$ \cite{Horodecki1999}. Due to this equivalence, OTF and OFEF are often used interchangeably in our subsequent studies.

Notably, for separable or PPT-entangled states \cite{Horodecki-PPT1996}, the OTF and OFEF attain the classical bounds of $2/d+1$ and $1/d$, respectively \cite{Bennett1996, Horodecki1999, Chitambar2024}. Exceeding these bounds demonstrates the state's usefulness in teleportation and its potential for quantum information processing. It's noteworthy that for an arbitrary two-qubit state $\rho_{AB}$, the OFEF is always upper-bounded by \cite{Verstraete2002}
\begin{align}
F^{\star}(\rho_{AB}) \leq \dfrac{1+\mathcal{C}(\rho_{AB})}{2}, \label{OFEF-N-C}
\end{align}
where $\mathcal{C}(X)$ is the concurrence of the state $X$ \cite{Wootters1998}. 
While this upper bound is achieved for pure entangled states and certain classes of mixed entangled states \cite{Verstraete2002, Ghosal2020}, it is generally not true for arbitrary states. This distinction underscores the independence between $F^{\star}$ and $\mathcal{C}$ \cite{Somshubhro2012, Pal2018}.

 \begin{figure}[h!]
    \includegraphics[height=140px, width=250px]{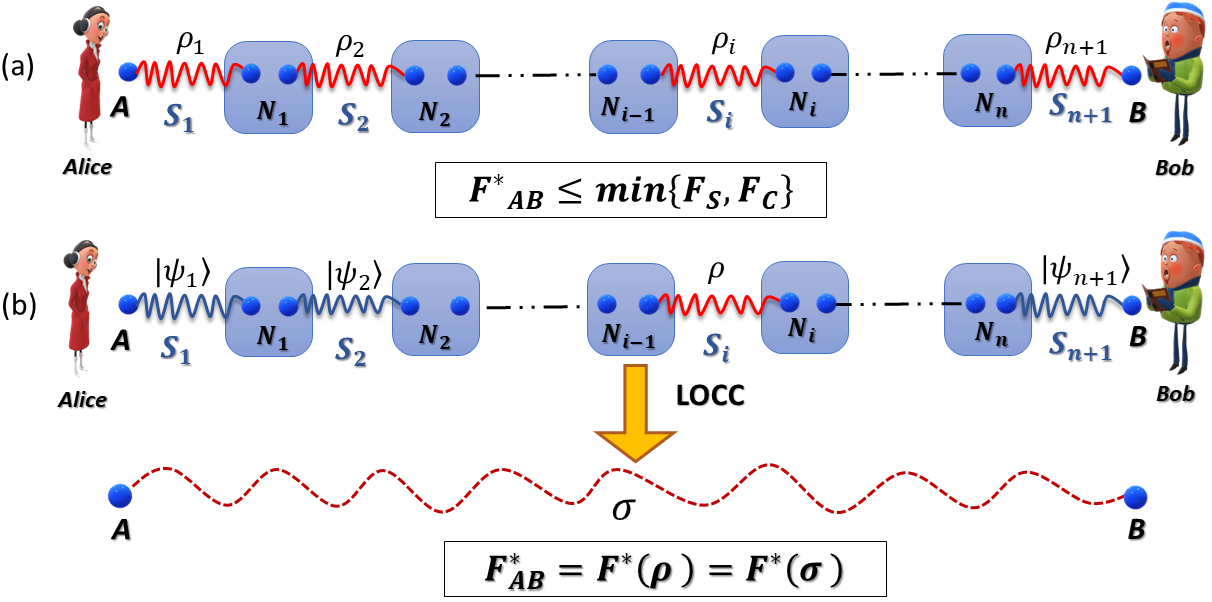}
    \caption{{\it Quantum Repeater-Based Teleportation Scheme.--} (\textbf{a})  Considering \( n+1 \) segments (\( n \) nodes) between the sender (Alice) and receiver (Bob), where each segment \( S_i \) shares a two-qubit entangled state \(\rho_i\). The general strategy is to prepare a two-qubit state between Alice and Bob in such a way that it achieves the highest possible fully entangled fraction (\( F^\star_{AB} \)), optimized over LOCC. (\textbf{b}) In this scenario, one of the segments ($ S_i$) shares a noisy state \(\rho\) ({\it called noise segment}). The objective is to establish \( F^\star_{AB} = F^\star(\rho) \), assuming the other segments ({\it noise-free or called free segment}) can share arbitrary entangled states. We aim to find a LOCC protocol that ultimately results in a two-qubit state \(\sigma\) between Alice and Bob with \( F^\star(\sigma) = F^\star(\rho) \). This protocol leads to the optimal fully entangled fraction and saturates the upper bound.}
    \label{Fig:1}
\end{figure}

 {\it Repeater-based teleportation protocol.--}
 In the realm of long-range quantum communication, quantum repeater-based protocols stand out for their effectiveness in mitigating uncontrollable noise \cite{Azuma2023}. This approach entails dividing the communication line into multiple short-range segments, denoted as $S_i$, where instead of direct communication from the sender (Alice) to a receiver (Bob), quantum information is relayed through intermediary nodes ${N_i}$ along these segments. Let's assume there are $(n+1)$ segments, where $\{S_1,S_i,S_{n+1}\}^{n}_{i=2}$ are associated to $A-N_1$, $\{N_{i-1}-N_{i}\}_{i=2}^{n}$, and $N_n-B$ parties, respectively, as depicted in Fig \ref{Fig:1}(a).  If individual segments $S_i$ share some arbitrary state $\rho_i$, what is the optimal fidelity achievable between the end-to-end parties (AB)? While general studies are quite nontrivial, even for two-qubit noises in the given repeater scenario, the upper bound of OFEF can be expressed as follows:
\begin{align}
F^{\star}_{AB} \leq \min \{F_S, F_C\}. \label{min}
\end{align}
The bound $F_S = \min_i \{F^\star (\rho_i)\}$ arises from $F^{\star}$ being an entanglement measure, where $F^\star (\rho_i)$ denotes the OFEF of the state $\rho_i$. Meanwhile, the bound $F_C = (1+\Pi^{n+1}_{i=1}\mathcal{C}_i)/2$ is obtained using Eq.$~$\eqref{OFEF-N-C}, where $\mathcal{C}_i$ is the concurrence of the state $\rho_i$. Interestingly, these two bounds are typically independent of each other, however, in some special cases, they boil down to being equal \cite{ghosal2024optimal}. It's crucial to note that having separable states in any segment makes $F^\star_{AB}=1/2$, while $F^\star_{AB}=1$ requires maximal entanglement across every segment. However, when a single segment hosts a fixed noisy state $\rho$, the end-to-end OFEF becomes $F^\star_{AB}\leq F^\star(\rho)$. The upper bound can always be achievable by the E-SWAP protocol utilizing maximal resources in free-segments and nodes. Considering concurrence as a figure of merit, previous studies suggest that E-SWAP is essential to achieve $\mathcal{E}_{AB} = \mathcal{E}(\rho)$ between Alice and Bob \cite{Gour2004, Perseguers2008}. However, in this work, we focus on OFEF, an independent entanglement quantifier with significant operational relevance in quantum information processing. Our goal is to find an efficient LOCC protocol that consumes less entanglement than the E-SWAP protocol while still achieving the same fidelity advantage, i.e., $F^\star_{AB} = F^\star(\rho)$ (see Fig \ref{Fig:1}(b)). We begin with the following lemma, which establishes that for certain noisy states, no entanglement consumption advantage exists beyond the E-SWAP protocol. Detailed calculations are provided in Appendix A \cite{supp}.

\begin{Lemma} \label{Lemma1}
   In a single-node scenario if segment $S_1$ shares a two-qubit noisy state $\rho_1$ with the constraint $F^\star(\rho_{1})=(1+\mathcal{C}_{1})/2$, then entanglement swapping is considered to be a necessary protocol to establish $F_{AB}^\star=F^\star(\rho_{1})$.  
\end{Lemma}

Interestingly, as demonstrated in the subsequent theorem, there exists a nontrivial set of noisy entangled states for which non-maximally entangled states and non-maximal entangled measurements suffice to achieve the upper bound for $F^\star$. Consider a two-dimensional subspace $\mathcal{S}$ in $\mathcal{H}_A \otimes \mathcal{H}_B$ for a two-qubit composite system, spanned by a product state $\mathrm{\ket{P}}:=\ket{\psi}\ket{\phi}$ and its orthogonal entangled state $\ket{\zeta(\delta)}:=\sqrt{\delta}\ket{\psi}\ket{\phi^{\perp}}+\sqrt{1-\delta}\ket{\psi^{\perp}}\ket{\phi}$, where $\delta \in \left[1/2,1\right)$ is its Schmidt coefficient. We define a noisy state $\rho 
 (p,\delta)=p\mathrm{P}+(1-p)\zeta (\delta) \in \mathcal{S}$, where $\mathrm{P}:=\mathrm{\ket{P}}\mathrm{\bra{P}},~ \zeta (\delta):=\ket{\zeta (\delta)}\bra{\zeta (\delta)},$ and $p \in \left[0,1\right]$.
Now, let's introduce an arbitrary two-qubit pure NMES, $\ket{\psi(\alpha)}:= \{\sqrt{\alpha}, \sqrt{1-\alpha}\}$ where $\alpha$ is the Schmidt coefficient, will be utilized in the free segments. Additionally, we define a fine-grained two-qubit projective entangled measurement $M (\beta)\equiv \{ \ket{\phi_i}\bra{\phi_i}\}^4_{i=1}$, where $\{\sqrt{\beta}, \sqrt{1-\beta}\}$ are the Schmidt coefficients of all $\ket{\phi_i}$. Without loss of generality consider that $\alpha, \beta \in [1/2,1)$.
\begin{theorem}\label{theorem1}
    In a single-node scenario, suppose noisy segment $S_1$ has a mixed state $\rho (p,\delta)$. Meanwhile, the free segment $S_2$ shares an arbitrary pure state $\ket{\psi(\alpha)}$, and node $N_1$ performs a measurement $M(\beta)$. We observe that within a certain range of the parameters $\lbrace p, \delta,\alpha,\beta\rbrace$, it is indeed possible to establish $F^\star_{AB}=F^\star \left[\rho(p,\delta)\right]$ between the sender and receiver. 
\end{theorem}

 A detailed proof is deferred to Appendix B \cite{supp}. It's worth highlighting that when $\alpha=\beta=1/2$, this leads to the standard E-SWAP protocol, and then any arbitrary state shared in $S_1$ can be readily prepared between the sender and receiver.  However, Theorem \ref {theorem1} suggests that for different choices of $\alpha$ and $\beta$, the exact state may not be preparable and would inherently be noisier than $\rho (p,\delta)$. Nevertheless, it remains feasible to achieve equal fidelity for some specific range of $\lbrace p,\delta,\alpha,\beta\rbrace$. For instance, if we choose non-maximal resources ($\alpha= 0.75, \beta=0.60$) for $\delta=0.5 ~ (\delta = 0.7)$, we get $F^\star_{AB}=F^\star \left[\rho(p,\delta)\right]>1/2$ for $0.51\leq p< 1 (0.49 \leq p < 1)$. This eventually implies that the proposed protocol consumes a lesser amount of resources than the standard protocol.
\par
Given our understanding of the necessity of quantum repeater scenarios, it's vital to establish the equivalent result in an $n$-node scenario. In the subsequent corollary, we will demonstrate that the above studies can be readily extended when one of the end segments ($S_1$ or $S_{n+1}$) is associated with the noisy state, which effectively boils down to a {\it single}-node scenario.

\begin{corollary}\label{corollary}
Consider the free segments $\{S_i\}^{n+1}_{i=2}$, where $S_i$ shares the state $\ket{\psi_i}\equiv \{\sqrt{\alpha_i}, \sqrt{1-\alpha_i}\}$ with $\alpha_i \geq 1/2 ~,~ \forall i$. It is always possible to find noisy states $\rho (p, \delta)=p \mathrm{P}+(1-p)\zeta (\delta)$ shared in the first segment $S_{1}$, with non-vanishing region of $\lbrace p, \delta \rbrace$, such that $F^{\star}_{AB}=F^{\star}\left[\rho (p,\delta)\right]$ can be achieved between the sender and receiver.
\end{corollary}
From this, we observe an intriguing trend: for a constant NMES $(\alpha_i = \alpha > 1/2)$ shared across all free segments, the region of noise parameters $\{p, \delta\},$ where $F^{\star}_{AB} = F^{\star}\left[\rho(p, \delta)\right]$, shrinks as the number of nodes increases. In the limit of many nodes, the noisy state $\rho$ approaches a separable state (see proof of Corollary~\ref{corollary} in Appendix C). However, irrespective of the number of nodes, can one still achieve end-to-end teleportation fidelity identical to the noisy state $\rho(p, \delta)$ for some fixed noise parameters $\{p, \delta\}$? The E-SWAP protocol offers a straightforward solution, consuming 1-ebit of entanglement per free segment to meet this fidelity requirement for any noisy state. Interestingly, as highlighted in Corollary \ref{corollary}, for specific ranges of $\{p, \delta\}$, the same fidelity can be achieved using strictly NMES in the free segments, thereby significantly reducing entanglement consumption. To quantify the resource savings in our protocol for an \( n \)-node setup (with \( n \) free segments), we define 
$R_v = n(1 - \mathcal{C})$,  
where \( 1 - \mathcal{C} \) denotes the entanglement saved per free segment when using a NMES with concurrence \( \mathcal{C} \) instead of a MES with concurrence \( 1 \). For an intuitive picture, we refer the reader to  Appendix I \cite{supp}. Since concurrence is equivalent to the entanglement of formation for arbitrary two-qubit states, it serves as an ideal metric to quantify entanglement consumption. Notably, our findings reveal that the saved resource $R_v$ is strictly positive for all $n$. Two key observations emerge from this analysis (detailed proofs in Appendix D \cite{supp}):

\begin{observation} \label{obs1}
    As the number of nodes $n\rightarrow \infty$, the saved resource $R_v$ tends toward a constant value for a given noisy state $\rho (p, \delta)$, as shown in Fig \ref{Fig:2a}.
\end{observation}
\begin{observation} \label{obs2}
    For a given number of nodes $n$, as the noisy state $\rho(p, \delta)$ tends towards a separable state, $R_v$ increases accordingly, as illustrated in Fig \ref{Fig:2b}. 
\end{observation}
\begin{figure}
     \centering
     \begin{subfigure}[b]{0.5\linewidth}
         \centering
         \includegraphics[width=0.95\linewidth]{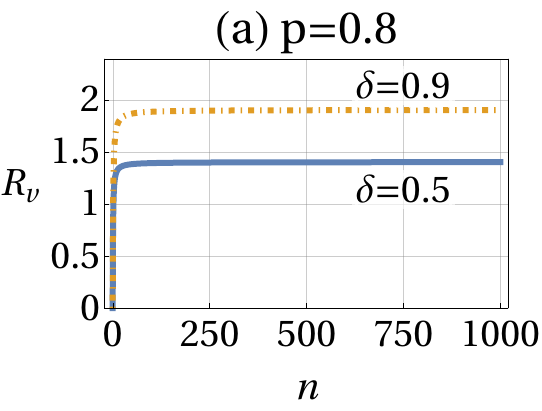}
         \captionlistentry{}
        \label{Fig:2a}
     \end{subfigure}%
     \begin{subfigure}[b]{0.5\linewidth}
        \centering
        \includegraphics[width=0.95\linewidth]{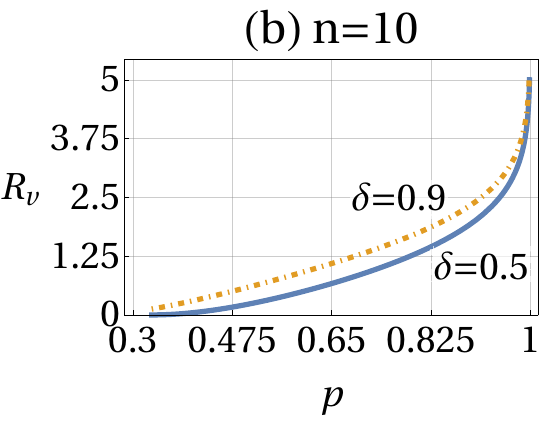}
        \captionlistentry{}
         \label{Fig:2b}
     \end{subfigure}
        \caption{(\textbf{a}) {\it $R_v$ vs. $n$ :} The saved resource remains positive for all \(n\) and approach saturation as the number of nodes increases. (\textbf{b})  {\it $R_v$ vs. $p$ :} Higher noise levels in a given repeater scenario result in greater resource savings under our protocol.}
\end{figure}
So far, we have explored scenarios where noise is localized at one of the end segments. However,  consider the segments $\{S_i\}^{n+1}_{i=1}$ share a composition of states $\{\{\psi(\alpha_i)\}^{i=m-1}_{i=1}, \rho_m, \{\psi(\alpha_i)\}_{i=m+1}^{i=n+1}\}$, where the $m^{th}$ segment is noisy and shares a fixed noisy state $\rho_m$, while the other segments are considered free, each with an arbitrary pure entangled state $\ket{\psi(\alpha_i)}$, with $\alpha_i\geq 1/2$. We demonstrate that this scenario effectively reduces to a {\it two}-node scenario for any value of $m\in \{2,\cdots,n\}$. By setting the $m^{th}$ segment as a reference, the left segments $\{S_i\}^{m-1}_{i=1}$ and the right segments $\{S_i\}_{i=m+1}^{i=n+1}$ can be composed into new pair of segments $S_1$ and $S_3$, respectively, while segment $S_2$ will be the noisy one in the {\it two}-node scenario. Precisely, the segments $\{S_i\}^3_{i=1}$ will have a composition of $\{\psi(\alpha^\prime_l),\rho_2,\psi(\alpha^\prime_r)\}$, where $\psi(\alpha^\prime_{l(r)})$ represents the effective pure entangled state at segment $S_1(S_3)$ and $\rho_2=\rho_m$. In the subsequent theorem, we will derive the condition $F^\star_{AB}=F^\star (\rho_2)$, with a detailed proof deferred to Appendix E \cite{supp}.

\begin{theorem}\label{theorem2}
 Let's assume a two-node scenario, where the intermediary segment $S_2$ shares a noisy state $\rho(p,\delta) = p ~\mathrm{P}+(1-p)\zeta (\delta) \in \mathcal{S}$, while the free segments ($i\in \{1,3\}$) share pure entangled states $\psi (\alpha^\prime_l)$ and $\psi (\alpha^\prime_r)$ respectively. We establish a nontrivial LOCC protocol that achieves $F^\star_{AB}=F^\star \left[\rho(p, \delta)\right]$ when the condition $\alpha^\prime_l \alpha^\prime_r/(1-\alpha^\prime_l)(1-\alpha^\prime_r) \leq p^2/\delta(1-\delta)(1-p)^2$ is satisfied.
\end{theorem}

In a quantum repeater scenario, the position of a noisy segment exhibiting the narrowest (or broadest) range of $\{p, \delta\}$, as outlined in Theorem \ref{theorem2}, is explicitly influenced by the distribution of $\{\alpha_i\}$ across the free segments. A detailed discussion in Appendix F explores how the noise's location within the repeater chain impacts $R_v$ and governs resource optimization in quantum repeater architectures. However, in the presence of noise across multiple segments, it remains uncertain whether we could achieve \(F^\star_{AB} = \min_{i}\{F^\star(\rho_i)\}\). Interestingly, we found an example where a noisy state \(\rho\), when shared in a single segment, allows \(F^\star_{AB} = F^\star(\rho)\). However, if two segments are affected by this same noisy state, \(F^\star_{AB} < F^\star(\rho)\) (see Appendix G \cite{supp}).

 \begin{figure}[h!]
     \centering
         \centering
         \includegraphics[height=130px, width=180px]{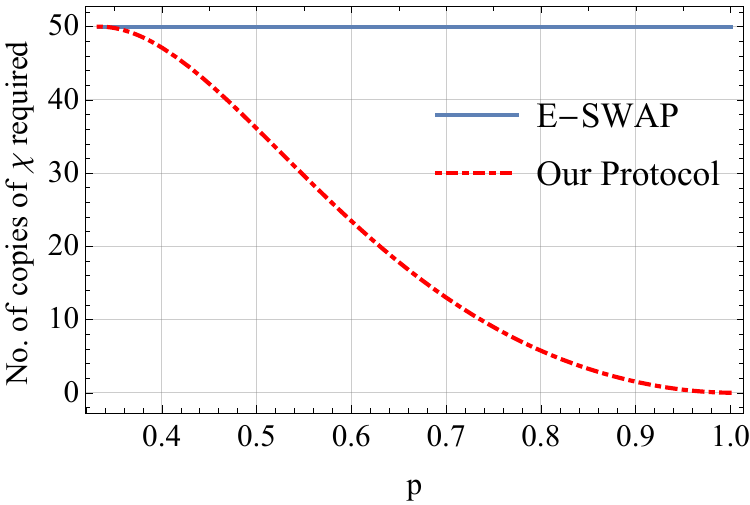}
         \caption{{\it Noise vs. Copies}: When the free segment contains a noisy state $\chi$, the traditional approach suggests distilling $\chi^{\otimes j}$ into a maximally entangled state $(\phi^+)^{\otimes k}$ asymptotically as $j \rightarrow \infty$, and then distributing the resulting state $\rho(p, \delta)$ between the end parties using E-SWAP. However, our analysis (Theorem \ref{theorem1}) shows that distilling a desired non-maximally entangled state $(\chi^{\otimes j}\rightarrow\psi^{\otimes k^\prime})$ is sufficient to achieve the same OFEF. Importantly, in the asymptotic limit, the number of copies required to distill a non-maximally entangled state is always lower than for a maximally entangled state, i.e., $j/k^\prime < j/k$. Furthermore, as illustrated in the above plot, as the noise p in the noisy segment increases, the number of required copies decreases for our protocol, in contrast to the traditional  E-SWAP protocol. For detailed calculations, we refer the reader to Appendix I \cite{supp}}
         \label{f11}
\end{figure}

{\it Noise in free segment.--} The analysis so far has primarily focused on using pure non-maximally entangled states in the free segments. Here, we extend this investigation to study how the OFEF changes when the resources in the free segments are replaced with noisy, mixed entangled states (denoted as $\chi$). Specifically, we limit our study to the one-node scenario, where the noisy segment is characterized by a mixed state $\rho(p, \delta)$. To assess the impact of noise in the free segments, we conduct a detailed analysis in Appendix I \cite{supp} under the following two scenarios: (i) {\bf Asymptotic copies of $\chi$:} This scenario explores the full potential of the noisy state in an idealized setting, where infinitely many copies of $\chi$ are available. The analysis demonstrates how fewer copies can achieve the desired teleportation fidelity compared to traditional approaches (see Fig. \ref{f11}). (ii) {\bf Single-copies of $\chi$:} This scenario explores the practical significance of utilizing a single instance of a noisy state. To ensure a realistic analysis, we incorporate practical noise models, including white noise, photon-loss noise, and measurement noise, demonstrating that the average teleportation fidelity remains robust under these conditions. These investigations offer valuable insights into the theoretical limits and practical feasibility of entanglement distribution across distant parties.

{\it Summary \& Outlook.--} In a quantum repeater scenario, where multiple segments share non-maximally entangled resources, effective long-range quantum communication depends on the optimal distribution of entanglement between the desired parties. Here, we uncover a remarkable feature of non-maximally entangled states: their potential in quantum tasks such as teleportation in a repeater scenario—a capability often overshadowed by their performance in conventional entanglement measures like Concurrence, singlet conversion probability, and worst-case entanglement distribution. Notably, this improvement also leads to a significant reduction in entanglement consumption, which is a critical advantage for long-distance communication protocols that rely on entanglement as a key resource. These findings underscore the importance of effectively distributing and utilizing entanglement, which forms the backbone of many quantum communication protocols where teleportation plays a vital role. This includes applications in long-distance quantum cryptography \cite{Bennett1984, Ekert1991, Gisin2002,Yin2020,Hu2023}, position-based cryptography \cite{Buhrman2014}, the establishment of a quantum internet \cite{Kimble2008, Wehner2018}, distributed quantum computation \cite{Cirac1999,CALEFFI2024110672,Main2025}, measurement-based quantum computation \cite{Gottesman1999,Leung2001,Nielsen2003,Robert2003,Leung2004,Jozsa2006}, blind quantum computation\cite{Broadbent2009}, and many other emerging areas \cite{Sangouard2011, Azuma2023, Curty2019, Li2021, Scarani2009, Klobus2012, Azuma2012}. In the supplementary material, we provide an overview of the broader implications of our findings in practical noisy setups (See Appendix J \cite{supp} ). Beyond its immediate results, this work raises important questions about the operational role of entanglement in quantum systems. It paves the way for further exploration into whether similar advantages can be realized using other entanglement measures, such as coherent information or private capacity, which define entanglement’s ability to transmit quantum information or generate secure keys.

A notable strength of this work lies in its focus on noisy states that are both theoretically interesting and experimentally relevant. These states, arising from physical processes such as atomic state decay or photon loss in photon-number encoding, represent dominant noise mechanisms in many practical quantum setups (see Appendix H  \cite{supp}). However, extending this analysis to other classes of noisy states and entanglement measures remains an open and exciting direction for future research. Moreover, our investigation demonstrates a noise-robust approach by accounting for imperfections in free segments, achieving substantial teleportation fidelity even under significant impurity (see Appendix I  \cite{supp}). This provides a solid foundation for designing practical and reliable quantum communication protocols. In summary, this study bridges foundational insights with operational significance, offering a robust framework for advancing quantum communication technologies under realistic noise conditions. It not only highlights the untapped potential of non-maximally entangled states but also lays the groundwork for further exploration of entanglement's role in distributed quantum systems, advancing the quest for efficient and noise-resilient quantum communication networks. \nocite {Horodecki1995,Horodecki199654,Horodecki199621,PhysRevA.78.032321,Nielsen2010,riera2021entanglement,BennettPure1996,Bennett1996}

\begin{acknowledgements}
 {\bf Acknowledgements:} We acknowledge Dr. Manik Banik, Prof. Antonio Acin, Dr. Ananda G. Maity, Dr. Some Sankar Bhattacharya, and Dr. Amit Mukherjee for many fruitful discussions and suggestions. M.A. acknowledges funding from the European Union’s
Horizon 2023 research and innovation programme under the
Marie Skłodowska-Curie grant agreement No. 101153001
as part of the project ‘Harnessing Quantum Resources in
Network Communication and Their Thermodynamic Underpinnings (QURES). The views and opinions expressed by are,
however the sole responsibility of the author(s) only.
\end{acknowledgements}
\bibliography{Repeater}

\begin{thebibliography}{81}%
\makeatletter
\providecommand \@ifxundefined [1]{%
 \@ifx{#1\undefined}
}%
\providecommand \@ifnum [1]{%
 \ifnum #1\expandafter \@firstoftwo
 \else \expandafter \@secondoftwo
 \fi
}%
\providecommand \@ifx [1]{%
 \ifx #1\expandafter \@firstoftwo
 \else \expandafter \@secondoftwo
 \fi
}%
\providecommand \natexlab [1]{#1}%
\providecommand \enquote  [1]{``#1''}%
\providecommand \bibnamefont  [1]{#1}%
\providecommand \bibfnamefont [1]{#1}%
\providecommand \citenamefont [1]{#1}%
\providecommand \href@noop [0]{\@secondoftwo}%
\providecommand \href [0]{\begingroup \@sanitize@url \@href}%
\providecommand \@href[1]{\@@startlink{#1}\@@href}%
\providecommand \@@href[1]{\endgroup#1\@@endlink}%
\providecommand \@sanitize@url [0]{\catcode `\\12\catcode `\$12\catcode
  `\&12\catcode `\#12\catcode `\^12\catcode `\_12\catcode `\%12\relax}%
\providecommand \@@startlink[1]{}%
\providecommand \@@endlink[0]{}%
\providecommand \url  [0]{\begingroup\@sanitize@url \@url }%
\providecommand \@url [1]{\endgroup\@href {#1}{\urlprefix }}%
\providecommand \urlprefix  [0]{URL }%
\providecommand \Eprint [0]{\href }%
\providecommand \doibase [0]{https://doi.org/}%
\providecommand \selectlanguage [0]{\@gobble}%
\providecommand \bibinfo  [0]{\@secondoftwo}%
\providecommand \bibfield  [0]{\@secondoftwo}%
\providecommand \translation [1]{[#1]}%
\providecommand \BibitemOpen [0]{}%
\providecommand \bibitemStop [0]{}%
\providecommand \bibitemNoStop [0]{.\EOS\space}%
\providecommand \EOS [0]{\spacefactor3000\relax}%
\providecommand \BibitemShut  [1]{\csname bibitem#1\endcsname}%
\let\auto@bib@innerbib\@empty
\bibitem [{\citenamefont {Wilde}(2013)}]{Wilde2013}%
  \BibitemOpen
  \bibfield  {author} {\bibinfo {author} {\bibfnamefont {M.~M.}\ \bibnamefont
  {Wilde}},\ }\href@noop {} {\emph {\bibinfo {title} {Quantum Information
  Theory}}}\ (\bibinfo  {publisher} {Cambridge university press},\ \bibinfo
  {year} {2013})\BibitemShut {NoStop}%
\bibitem [{\citenamefont {Marinescu}\ and\ \citenamefont
  {Marinescu}(2012)}]{Marinescu2012}%
  \BibitemOpen
  \bibfield  {author} {\bibinfo {author} {\bibfnamefont {D.~C.}\ \bibnamefont
  {Marinescu}}\ and\ \bibinfo {author} {\bibfnamefont {G.~M.}\ \bibnamefont
  {Marinescu}},\ }\bibfield  {title} {\bibinfo {title} {Chapter 3 - classical
  and quantum information theory},\ }in\ \href
  {https://doi.org/https://doi.org/10.1016/B978-0-12-383874-2.00003-5} {\emph
  {\bibinfo {booktitle} {Classical and Quantum Information}}},\ \bibinfo
  {editor} {edited by\ \bibinfo {editor} {\bibfnamefont {D.~C.}\ \bibnamefont
  {Marinescu}}\ and\ \bibinfo {editor} {\bibfnamefont {G.~M.}\ \bibnamefont
  {Marinescu}}}\ (\bibinfo  {publisher} {Academic Press},\ \bibinfo {address}
  {Boston},\ \bibinfo {year} {2012})\ pp.\ \bibinfo {pages}
  {221--344}\BibitemShut {NoStop}%
\bibitem [{\citenamefont {Khatri}\ and\ \citenamefont
  {Wilde}(2024)}]{Khatri2024}%
  \BibitemOpen
  \bibfield  {author} {\bibinfo {author} {\bibfnamefont {S.}~\bibnamefont
  {Khatri}}\ and\ \bibinfo {author} {\bibfnamefont {M.~M.}\ \bibnamefont
  {Wilde}},\ }\href@noop {} {\bibinfo {title} {Principles of quantum
  communication theory: A modern approach}} (\bibinfo {year} {2024}),\ \Eprint
  {https://arxiv.org/abs/2011.04672} {arXiv:2011.04672 [quant-ph]} \BibitemShut
  {NoStop}%
\bibitem [{\citenamefont {Briegel}\ \emph {et~al.}(1998)\citenamefont
  {Briegel}, \citenamefont {D\"ur}, \citenamefont {Cirac},\ and\ \citenamefont
  {Zoller}}]{Briegel1998}%
  \BibitemOpen
  \bibfield  {author} {\bibinfo {author} {\bibfnamefont {H.-J.}\ \bibnamefont
  {Briegel}}, \bibinfo {author} {\bibfnamefont {W.}~\bibnamefont {D\"ur}},
  \bibinfo {author} {\bibfnamefont {J.~I.}\ \bibnamefont {Cirac}},\ and\
  \bibinfo {author} {\bibfnamefont {P.}~\bibnamefont {Zoller}},\ }\bibfield
  {title} {\bibinfo {title} {Quantum repeaters: The role of imperfect local
  operations in quantum communication},\ }\href
  {https://doi.org/10.1103/PhysRevLett.81.5932} {\bibfield  {journal} {\bibinfo
   {journal} {Phys. Rev. Lett.}\ }\textbf {\bibinfo {volume} {81}},\ \bibinfo
  {pages} {5932} (\bibinfo {year} {1998})}\BibitemShut {NoStop}%
\bibitem [{\citenamefont {D\"ur}\ \emph {et~al.}(1999)\citenamefont {D\"ur},
  \citenamefont {Briegel}, \citenamefont {Cirac},\ and\ \citenamefont
  {Zoller}}]{Dur1999}%
  \BibitemOpen
  \bibfield  {author} {\bibinfo {author} {\bibfnamefont {W.}~\bibnamefont
  {D\"ur}}, \bibinfo {author} {\bibfnamefont {H.-J.}\ \bibnamefont {Briegel}},
  \bibinfo {author} {\bibfnamefont {J.~I.}\ \bibnamefont {Cirac}},\ and\
  \bibinfo {author} {\bibfnamefont {P.}~\bibnamefont {Zoller}},\ }\bibfield
  {title} {\bibinfo {title} {Quantum repeaters based on entanglement
  purification},\ }\href {https://doi.org/10.1103/PhysRevA.59.169} {\bibfield
  {journal} {\bibinfo  {journal} {Phys. Rev. A}\ }\textbf {\bibinfo {volume}
  {59}},\ \bibinfo {pages} {169} (\bibinfo {year} {1999})}\BibitemShut
  {NoStop}%
\bibitem [{\citenamefont {Sangouard}\ \emph {et~al.}(2011)\citenamefont
  {Sangouard}, \citenamefont {Simon}, \citenamefont {de~Riedmatten},\ and\
  \citenamefont {Gisin}}]{Sangouard2011}%
  \BibitemOpen
  \bibfield  {author} {\bibinfo {author} {\bibfnamefont {N.}~\bibnamefont
  {Sangouard}}, \bibinfo {author} {\bibfnamefont {C.}~\bibnamefont {Simon}},
  \bibinfo {author} {\bibfnamefont {H.}~\bibnamefont {de~Riedmatten}},\ and\
  \bibinfo {author} {\bibfnamefont {N.}~\bibnamefont {Gisin}},\ }\bibfield
  {title} {\bibinfo {title} {Quantum repeaters based on atomic ensembles and
  linear optics},\ }\href {https://doi.org/10.1103/RevModPhys.83.33} {\bibfield
   {journal} {\bibinfo  {journal} {Rev. Mod. Phys.}\ }\textbf {\bibinfo
  {volume} {83}},\ \bibinfo {pages} {33} (\bibinfo {year} {2011})}\BibitemShut
  {NoStop}%
\bibitem [{\citenamefont {Azuma}\ \emph {et~al.}(2023)\citenamefont {Azuma},
  \citenamefont {Economou}, \citenamefont {Elkouss}, \citenamefont {Hilaire},
  \citenamefont {Jiang}, \citenamefont {Lo},\ and\ \citenamefont
  {Tzitrin}}]{Azuma2023}%
  \BibitemOpen
  \bibfield  {author} {\bibinfo {author} {\bibfnamefont {K.}~\bibnamefont
  {Azuma}}, \bibinfo {author} {\bibfnamefont {S.~E.}\ \bibnamefont {Economou}},
  \bibinfo {author} {\bibfnamefont {D.}~\bibnamefont {Elkouss}}, \bibinfo
  {author} {\bibfnamefont {P.}~\bibnamefont {Hilaire}}, \bibinfo {author}
  {\bibfnamefont {L.}~\bibnamefont {Jiang}}, \bibinfo {author} {\bibfnamefont
  {H.-K.}\ \bibnamefont {Lo}},\ and\ \bibinfo {author} {\bibfnamefont
  {I.}~\bibnamefont {Tzitrin}},\ }\bibfield  {title} {\bibinfo {title} {Quantum
  repeaters: From quantum networks to the quantum internet},\ }\href
  {https://doi.org/10.1103/RevModPhys.95.045006} {\bibfield  {journal}
  {\bibinfo  {journal} {Rev. Mod. Phys.}\ }\textbf {\bibinfo {volume} {95}},\
  \bibinfo {pages} {045006} (\bibinfo {year} {2023})}\BibitemShut {NoStop}%
\bibitem [{\citenamefont {Wootters}\ and\ \citenamefont
  {Zurek}(1982)}]{Wootters1982}%
  \BibitemOpen
  \bibfield  {author} {\bibinfo {author} {\bibfnamefont {W.~K.}\ \bibnamefont
  {Wootters}}\ and\ \bibinfo {author} {\bibfnamefont {W.~H.}\ \bibnamefont
  {Zurek}},\ }\bibfield  {title} {\bibinfo {title} {A single quantum cannot be
  cloned},\ }\href {https://doi.org/10.1038/299802a0} {\bibfield  {journal}
  {\bibinfo  {journal} {Nature}\ }\textbf {\bibinfo {volume} {299}},\ \bibinfo
  {pages} {802} (\bibinfo {year} {1982})}\BibitemShut {NoStop}%
\bibitem [{\citenamefont {Bennett}\ \emph {et~al.}(1993)\citenamefont
  {Bennett}, \citenamefont {Brassard}, \citenamefont {Cr\'epeau}, \citenamefont
  {Jozsa}, \citenamefont {Peres},\ and\ \citenamefont
  {Wootters}}]{Bennett1993}%
  \BibitemOpen
  \bibfield  {author} {\bibinfo {author} {\bibfnamefont {C.~H.}\ \bibnamefont
  {Bennett}}, \bibinfo {author} {\bibfnamefont {G.}~\bibnamefont {Brassard}},
  \bibinfo {author} {\bibfnamefont {C.}~\bibnamefont {Cr\'epeau}}, \bibinfo
  {author} {\bibfnamefont {R.}~\bibnamefont {Jozsa}}, \bibinfo {author}
  {\bibfnamefont {A.}~\bibnamefont {Peres}},\ and\ \bibinfo {author}
  {\bibfnamefont {W.~K.}\ \bibnamefont {Wootters}},\ }\bibfield  {title}
  {\bibinfo {title} {Teleporting an unknown quantum state via dual classical
  and einstein-podolsky-rosen channels},\ }\href
  {https://doi.org/10.1103/PhysRevLett.70.1895} {\bibfield  {journal} {\bibinfo
   {journal} {Phys. Rev. Lett.}\ }\textbf {\bibinfo {volume} {70}},\ \bibinfo
  {pages} {1895} (\bibinfo {year} {1993})}\BibitemShut {NoStop}%
\bibitem [{\citenamefont {\ifmmode~\dot{Z}\else \.{Z}\fi{}ukowski}\ \emph
  {et~al.}(1993)\citenamefont {\ifmmode~\dot{Z}\else \.{Z}\fi{}ukowski},
  \citenamefont {Zeilinger}, \citenamefont {Horne},\ and\ \citenamefont
  {Ekert}}]{Zukowski1993}%
  \BibitemOpen
  \bibfield  {author} {\bibinfo {author} {\bibfnamefont {M.}~\bibnamefont
  {\ifmmode~\dot{Z}\else \.{Z}\fi{}ukowski}}, \bibinfo {author} {\bibfnamefont
  {A.}~\bibnamefont {Zeilinger}}, \bibinfo {author} {\bibfnamefont {M.~A.}\
  \bibnamefont {Horne}},\ and\ \bibinfo {author} {\bibfnamefont {A.~K.}\
  \bibnamefont {Ekert}},\ }\bibfield  {title} {\bibinfo {title}
  {``event-ready-detectors'' bell experiment via entanglement swapping},\
  }\href {https://doi.org/10.1103/PhysRevLett.71.4287} {\bibfield  {journal}
  {\bibinfo  {journal} {Phys. Rev. Lett.}\ }\textbf {\bibinfo {volume} {71}},\
  \bibinfo {pages} {4287} (\bibinfo {year} {1993})}\BibitemShut {NoStop}%
\bibitem [{\citenamefont {Pan}\ \emph {et~al.}(1998)\citenamefont {Pan},
  \citenamefont {Bouwmeester}, \citenamefont {Weinfurter},\ and\ \citenamefont
  {Zeilinger}}]{Pan1998}%
  \BibitemOpen
  \bibfield  {author} {\bibinfo {author} {\bibfnamefont {J.-W.}\ \bibnamefont
  {Pan}}, \bibinfo {author} {\bibfnamefont {D.}~\bibnamefont {Bouwmeester}},
  \bibinfo {author} {\bibfnamefont {H.}~\bibnamefont {Weinfurter}},\ and\
  \bibinfo {author} {\bibfnamefont {A.}~\bibnamefont {Zeilinger}},\ }\bibfield
  {title} {\bibinfo {title} {Experimental entanglement swapping: Entangling
  photons that never interacted},\ }\href
  {https://doi.org/10.1103/PhysRevLett.80.3891} {\bibfield  {journal} {\bibinfo
   {journal} {Phys. Rev. Lett.}\ }\textbf {\bibinfo {volume} {80}},\ \bibinfo
  {pages} {3891} (\bibinfo {year} {1998})}\BibitemShut {NoStop}%
\bibitem [{\citenamefont {Bouda}\ and\ \citenamefont
  {Buzek}(2001)}]{Bouda2001}%
  \BibitemOpen
  \bibfield  {author} {\bibinfo {author} {\bibfnamefont {J.}~\bibnamefont
  {Bouda}}\ and\ \bibinfo {author} {\bibfnamefont {V.}~\bibnamefont {Buzek}},\
  }\bibfield  {title} {\bibinfo {title} {Entanglement swapping between
  multi-qudit systems},\ }\href {https://doi.org/10.1088/0305-4470/34/20/304}
  {\bibfield  {journal} {\bibinfo  {journal} {Journal of Physics A:
  Mathematical and General}\ }\textbf {\bibinfo {volume} {34}},\ \bibinfo
  {pages} {4301} (\bibinfo {year} {2001})}\BibitemShut {NoStop}%
\bibitem [{\citenamefont {Bayrakci}\ and\ \citenamefont
  {Ozaydin}(2022)}]{Bayrakci2022}%
  \BibitemOpen
  \bibfield  {author} {\bibinfo {author} {\bibfnamefont {V.}~\bibnamefont
  {Bayrakci}}\ and\ \bibinfo {author} {\bibfnamefont {F.}~\bibnamefont
  {Ozaydin}},\ }\bibfield  {title} {\bibinfo {title} {Quantum zeno repeaters},\
  }\href {https://doi.org/10.1038/s41598-022-19170-z} {\bibfield  {journal}
  {\bibinfo  {journal} {Scientific Reports}\ }\textbf {\bibinfo {volume}
  {12}},\ \bibinfo {pages} {15302} (\bibinfo {year} {2022})}\BibitemShut
  {NoStop}%
\bibitem [{\citenamefont {Horodecki}\ \emph {et~al.}(2009)\citenamefont
  {Horodecki}, \citenamefont {Horodecki}, \citenamefont {Horodecki},\ and\
  \citenamefont {Horodecki}}]{Horodecki2009}%
  \BibitemOpen
  \bibfield  {author} {\bibinfo {author} {\bibfnamefont {R.}~\bibnamefont
  {Horodecki}}, \bibinfo {author} {\bibfnamefont {P.}~\bibnamefont
  {Horodecki}}, \bibinfo {author} {\bibfnamefont {M.}~\bibnamefont
  {Horodecki}},\ and\ \bibinfo {author} {\bibfnamefont {K.}~\bibnamefont
  {Horodecki}},\ }\bibfield  {title} {\bibinfo {title} {Quantum entanglement},\
  }\href {https://doi.org/10.1103/RevModPhys.81.865} {\bibfield  {journal}
  {\bibinfo  {journal} {Rev. Mod. Phys.}\ }\textbf {\bibinfo {volume} {81}},\
  \bibinfo {pages} {865} (\bibinfo {year} {2009})}\BibitemShut {NoStop}%
\bibitem [{\citenamefont {Bennett}\ \emph
  {et~al.}(1996{\natexlab{a}})\citenamefont {Bennett}, \citenamefont
  {DiVincenzo}, \citenamefont {Smolin},\ and\ \citenamefont
  {Wootters}}]{Bennett1996}%
  \BibitemOpen
  \bibfield  {author} {\bibinfo {author} {\bibfnamefont {C.~H.}\ \bibnamefont
  {Bennett}}, \bibinfo {author} {\bibfnamefont {D.~P.}\ \bibnamefont
  {DiVincenzo}}, \bibinfo {author} {\bibfnamefont {J.~A.}\ \bibnamefont
  {Smolin}},\ and\ \bibinfo {author} {\bibfnamefont {W.~K.}\ \bibnamefont
  {Wootters}},\ }\bibfield  {title} {\bibinfo {title} {Mixed-state entanglement
  and quantum error correction},\ }\href
  {https://doi.org/10.1103/PhysRevA.54.3824} {\bibfield  {journal} {\bibinfo
  {journal} {Phys. Rev. A}\ }\textbf {\bibinfo {volume} {54}},\ \bibinfo
  {pages} {3824} (\bibinfo {year} {1996}{\natexlab{a}})}\BibitemShut {NoStop}%
\bibitem [{\citenamefont {Bennett}\ \emph
  {et~al.}(1996{\natexlab{b}})\citenamefont {Bennett}, \citenamefont
  {Brassard}, \citenamefont {Popescu}, \citenamefont {Schumacher},
  \citenamefont {Smolin},\ and\ \citenamefont
  {Wootters}}]{BennettBrassard1996}%
  \BibitemOpen
  \bibfield  {author} {\bibinfo {author} {\bibfnamefont {C.~H.}\ \bibnamefont
  {Bennett}}, \bibinfo {author} {\bibfnamefont {G.}~\bibnamefont {Brassard}},
  \bibinfo {author} {\bibfnamefont {S.}~\bibnamefont {Popescu}}, \bibinfo
  {author} {\bibfnamefont {B.}~\bibnamefont {Schumacher}}, \bibinfo {author}
  {\bibfnamefont {J.~A.}\ \bibnamefont {Smolin}},\ and\ \bibinfo {author}
  {\bibfnamefont {W.~K.}\ \bibnamefont {Wootters}},\ }\bibfield  {title}
  {\bibinfo {title} {Purification of noisy entanglement and faithful
  teleportation via noisy channels},\ }\href
  {https://doi.org/10.1103/PhysRevLett.76.722} {\bibfield  {journal} {\bibinfo
  {journal} {Phys. Rev. Lett.}\ }\textbf {\bibinfo {volume} {76}},\ \bibinfo
  {pages} {722} (\bibinfo {year} {1996}{\natexlab{b}})}\BibitemShut {NoStop}%
\bibitem [{\citenamefont {Deutsch}\ \emph {et~al.}(1996)\citenamefont
  {Deutsch}, \citenamefont {Ekert}, \citenamefont {Jozsa}, \citenamefont
  {Macchiavello}, \citenamefont {Popescu},\ and\ \citenamefont
  {Sanpera}}]{Deutsch1996}%
  \BibitemOpen
  \bibfield  {author} {\bibinfo {author} {\bibfnamefont {D.}~\bibnamefont
  {Deutsch}}, \bibinfo {author} {\bibfnamefont {A.}~\bibnamefont {Ekert}},
  \bibinfo {author} {\bibfnamefont {R.}~\bibnamefont {Jozsa}}, \bibinfo
  {author} {\bibfnamefont {C.}~\bibnamefont {Macchiavello}}, \bibinfo {author}
  {\bibfnamefont {S.}~\bibnamefont {Popescu}},\ and\ \bibinfo {author}
  {\bibfnamefont {A.}~\bibnamefont {Sanpera}},\ }\bibfield  {title} {\bibinfo
  {title} {Quantum privacy amplification and the security of quantum
  cryptography over noisy channels},\ }\href
  {https://doi.org/10.1103/PhysRevLett.77.2818} {\bibfield  {journal} {\bibinfo
   {journal} {Phys. Rev. Lett.}\ }\textbf {\bibinfo {volume} {77}},\ \bibinfo
  {pages} {2818} (\bibinfo {year} {1996})}\BibitemShut {NoStop}%
\bibitem [{\citenamefont {Razavi}\ \emph {et~al.}(2009)\citenamefont {Razavi},
  \citenamefont {Piani},\ and\ \citenamefont {L\"utkenhaus}}]{Razavi2009}%
  \BibitemOpen
  \bibfield  {author} {\bibinfo {author} {\bibfnamefont {M.}~\bibnamefont
  {Razavi}}, \bibinfo {author} {\bibfnamefont {M.}~\bibnamefont {Piani}},\ and\
  \bibinfo {author} {\bibfnamefont {N.}~\bibnamefont {L\"utkenhaus}},\
  }\bibfield  {title} {\bibinfo {title} {Quantum repeaters with imperfect
  memories: Cost and scalability},\ }\href
  {https://doi.org/10.1103/PhysRevA.80.032301} {\bibfield  {journal} {\bibinfo
  {journal} {Phys. Rev. A}\ }\textbf {\bibinfo {volume} {80}},\ \bibinfo
  {pages} {032301} (\bibinfo {year} {2009})}\BibitemShut {NoStop}%
\bibitem [{\citenamefont {Broadfoot}\ \emph {et~al.}(2010)\citenamefont
  {Broadfoot}, \citenamefont {Dorner},\ and\ \citenamefont
  {Jaksch}}]{Broadfoot2010}%
  \BibitemOpen
  \bibfield  {author} {\bibinfo {author} {\bibfnamefont {S.}~\bibnamefont
  {Broadfoot}}, \bibinfo {author} {\bibfnamefont {U.}~\bibnamefont {Dorner}},\
  and\ \bibinfo {author} {\bibfnamefont {D.}~\bibnamefont {Jaksch}},\
  }\bibfield  {title} {\bibinfo {title} {Singlet generation in mixed-state
  quantum networks},\ }\href {https://doi.org/10.1103/PhysRevA.81.042316}
  {\bibfield  {journal} {\bibinfo  {journal} {Phys. Rev. A}\ }\textbf {\bibinfo
  {volume} {81}},\ \bibinfo {pages} {042316} (\bibinfo {year}
  {2010})}\BibitemShut {NoStop}%
\bibitem [{\citenamefont {Pal}\ \emph {et~al.}(2014)\citenamefont {Pal},
  \citenamefont {Bandyopadhyay},\ and\ \citenamefont {Ghosh}}]{Pal2014}%
  \BibitemOpen
  \bibfield  {author} {\bibinfo {author} {\bibfnamefont {R.}~\bibnamefont
  {Pal}}, \bibinfo {author} {\bibfnamefont {S.}~\bibnamefont {Bandyopadhyay}},\
  and\ \bibinfo {author} {\bibfnamefont {S.}~\bibnamefont {Ghosh}},\ }\bibfield
   {title} {\bibinfo {title} {Entanglement sharing through noisy qubit
  channels: One-shot optimal singlet fraction},\ }\href
  {https://doi.org/10.1103/PhysRevA.90.052304} {\bibfield  {journal} {\bibinfo
  {journal} {Phys. Rev. A}\ }\textbf {\bibinfo {volume} {90}},\ \bibinfo
  {pages} {052304} (\bibinfo {year} {2014})}\BibitemShut {NoStop}%
\bibitem [{\citenamefont {Lo}\ and\ \citenamefont {Popescu}(2001)}]{Lo2001}%
  \BibitemOpen
  \bibfield  {author} {\bibinfo {author} {\bibfnamefont {H.-K.}\ \bibnamefont
  {Lo}}\ and\ \bibinfo {author} {\bibfnamefont {S.}~\bibnamefont {Popescu}},\
  }\bibfield  {title} {\bibinfo {title} {Concentrating entanglement by local
  actions: Beyond mean values},\ }\href
  {https://doi.org/10.1103/PhysRevA.63.022301} {\bibfield  {journal} {\bibinfo
  {journal} {Phys. Rev. A}\ }\textbf {\bibinfo {volume} {63}},\ \bibinfo
  {pages} {022301} (\bibinfo {year} {2001})}\BibitemShut {NoStop}%
\bibitem [{\citenamefont {Morikoshi}\ and\ \citenamefont
  {Koashi}(2001)}]{Morikoshi2001}%
  \BibitemOpen
  \bibfield  {author} {\bibinfo {author} {\bibfnamefont {F.}~\bibnamefont
  {Morikoshi}}\ and\ \bibinfo {author} {\bibfnamefont {M.}~\bibnamefont
  {Koashi}},\ }\bibfield  {title} {\bibinfo {title} {Deterministic entanglement
  concentration},\ }\href {https://doi.org/10.1103/PhysRevA.64.022316}
  {\bibfield  {journal} {\bibinfo  {journal} {Phys. Rev. A}\ }\textbf {\bibinfo
  {volume} {64}},\ \bibinfo {pages} {022316} (\bibinfo {year}
  {2001})}\BibitemShut {NoStop}%
\bibitem [{\citenamefont {Mart\'{\i}n-Delgado}\ and\ \citenamefont
  {Navascu\'es}(2003)}]{Martin2003}%
  \BibitemOpen
  \bibfield  {author} {\bibinfo {author} {\bibfnamefont {M.~A.}\ \bibnamefont
  {Mart\'{\i}n-Delgado}}\ and\ \bibinfo {author} {\bibfnamefont
  {M.}~\bibnamefont {Navascu\'es}},\ }\bibfield  {title} {\bibinfo {title}
  {Single-step distillation protocol with generalized beam splitters},\ }\href
  {https://doi.org/10.1103/PhysRevA.68.012322} {\bibfield  {journal} {\bibinfo
  {journal} {Phys. Rev. A}\ }\textbf {\bibinfo {volume} {68}},\ \bibinfo
  {pages} {012322} (\bibinfo {year} {2003})}\BibitemShut {NoStop}%
\bibitem [{\citenamefont {Hayashi}(2006)}]{Hayashi2006}%
  \BibitemOpen
  \bibfield  {author} {\bibinfo {author} {\bibfnamefont {M.}~\bibnamefont
  {Hayashi}},\ }\href {https://doi.org/10.1109/TIT.2006.872976} {\bibinfo
  {title} {General formulas for fixed-length quantum entanglement
  concentration}} (\bibinfo {year} {2006})\BibitemShut {NoStop}%
\bibitem [{\citenamefont {Buscemi}\ and\ \citenamefont
  {Datta}(2010)}]{Buscemi2010}%
  \BibitemOpen
  \bibfield  {author} {\bibinfo {author} {\bibfnamefont {F.}~\bibnamefont
  {Buscemi}}\ and\ \bibinfo {author} {\bibfnamefont {N.}~\bibnamefont
  {Datta}},\ }\bibfield  {title} {\bibinfo {title} {Distilling entanglement
  from arbitrary resources},\ }\bibfield  {journal} {\bibinfo  {journal}
  {Journal of Mathematical Physics}\ }\textbf {\bibinfo {volume} {51}},\ \href
  {https://doi.org/https://doi.org/10.1063/1.3483717}
  {https://doi.org/10.1063/1.3483717} (\bibinfo {year} {2010})\BibitemShut
  {NoStop}%
\bibitem [{\citenamefont {Buscemi}\ and\ \citenamefont
  {Datta}(2013)}]{Buscemi2013}%
  \BibitemOpen
  \bibfield  {author} {\bibinfo {author} {\bibfnamefont {F.}~\bibnamefont
  {Buscemi}}\ and\ \bibinfo {author} {\bibfnamefont {N.}~\bibnamefont
  {Datta}},\ }\bibfield  {title} {\bibinfo {title} {General theory of
  environment-assisted entanglement distillation},\ }\href
  {https://doi.org/10.1109/TIT.2012.2227673} {\bibfield  {journal} {\bibinfo
  {journal} {IEEE Transactions on Information Theory}\ }\textbf {\bibinfo
  {volume} {59}},\ \bibinfo {pages} {1940} (\bibinfo {year}
  {2013})}\BibitemShut {NoStop}%
\bibitem [{\citenamefont {Brandao}\ and\ \citenamefont
  {Datta}(2011)}]{Brandao2011}%
  \BibitemOpen
  \bibfield  {author} {\bibinfo {author} {\bibfnamefont {F.~G. S.~L.}\
  \bibnamefont {Brandao}}\ and\ \bibinfo {author} {\bibfnamefont
  {N.}~\bibnamefont {Datta}},\ }\bibfield  {title} {\bibinfo {title} {One-shot
  rates for entanglement manipulation under non-entangling maps},\ }\href
  {https://doi.org/10.1109/TIT.2011.2104531} {\bibfield  {journal} {\bibinfo
  {journal} {IEEE Transactions on Information Theory}\ }\textbf {\bibinfo
  {volume} {57}},\ \bibinfo {pages} {1754} (\bibinfo {year}
  {2011})}\BibitemShut {NoStop}%
\bibitem [{\citenamefont {Datta}\ and\ \citenamefont
  {Leditzky}(2015)}]{Datta2015}%
  \BibitemOpen
  \bibfield  {author} {\bibinfo {author} {\bibfnamefont {N.}~\bibnamefont
  {Datta}}\ and\ \bibinfo {author} {\bibfnamefont {F.}~\bibnamefont
  {Leditzky}},\ }\bibfield  {title} {\bibinfo {title} {Second-order asymptotics
  for source coding, dense coding, and pure-state entanglement conversions},\
  }\href {https://doi.org/10.1109/TIT.2014.2366994} {\bibfield  {journal}
  {\bibinfo  {journal} {IEEE Transactions on Information Theory}\ }\textbf
  {\bibinfo {volume} {61}},\ \bibinfo {pages} {582} (\bibinfo {year}
  {2015})}\BibitemShut {NoStop}%
\bibitem [{\citenamefont {Fang}\ \emph {et~al.}(2019)\citenamefont {Fang},
  \citenamefont {Wang}, \citenamefont {Tomamichel},\ and\ \citenamefont
  {Duan}}]{Fang2019}%
  \BibitemOpen
  \bibfield  {author} {\bibinfo {author} {\bibfnamefont {K.}~\bibnamefont
  {Fang}}, \bibinfo {author} {\bibfnamefont {X.}~\bibnamefont {Wang}}, \bibinfo
  {author} {\bibfnamefont {M.}~\bibnamefont {Tomamichel}},\ and\ \bibinfo
  {author} {\bibfnamefont {R.}~\bibnamefont {Duan}},\ }\bibfield  {title}
  {\bibinfo {title} {Non-asymptotic entanglement distillation},\ }\href
  {https://doi.org/10.1109/TIT.2019.2914688} {\bibfield  {journal} {\bibinfo
  {journal} {IEEE Transactions on Information Theory}\ }\textbf {\bibinfo
  {volume} {65}},\ \bibinfo {pages} {6454} (\bibinfo {year}
  {2019})}\BibitemShut {NoStop}%
\bibitem [{\citenamefont {Regula}\ \emph {et~al.}(2019)\citenamefont {Regula},
  \citenamefont {Fang}, \citenamefont {Wang},\ and\ \citenamefont
  {Gu}}]{Regula2019}%
  \BibitemOpen
  \bibfield  {author} {\bibinfo {author} {\bibfnamefont {B.}~\bibnamefont
  {Regula}}, \bibinfo {author} {\bibfnamefont {K.}~\bibnamefont {Fang}},
  \bibinfo {author} {\bibfnamefont {X.}~\bibnamefont {Wang}},\ and\ \bibinfo
  {author} {\bibfnamefont {M.}~\bibnamefont {Gu}},\ }\bibfield  {title}
  {\bibinfo {title} {One-shot entanglement distillation beyond local operations
  and classical communication},\ }\href
  {https://doi.org/10.1088/1367-2630/ab4732} {\bibfield  {journal} {\bibinfo
  {journal} {New Journal of Physics}\ }\textbf {\bibinfo {volume} {21}},\
  \bibinfo {pages} {103017} (\bibinfo {year} {2019})}\BibitemShut {NoStop}%
\bibitem [{\citenamefont {Hardy}\ and\ \citenamefont {Song}(2000)}]{Hardy2000}%
  \BibitemOpen
  \bibfield  {author} {\bibinfo {author} {\bibfnamefont {L.}~\bibnamefont
  {Hardy}}\ and\ \bibinfo {author} {\bibfnamefont {D.~D.}\ \bibnamefont
  {Song}},\ }\bibfield  {title} {\bibinfo {title} {Entanglement-swapping chains
  for general pure states},\ }\href
  {https://doi.org/10.1103/PhysRevA.62.052315} {\bibfield  {journal} {\bibinfo
  {journal} {Phys. Rev. A}\ }\textbf {\bibinfo {volume} {62}},\ \bibinfo
  {pages} {052315} (\bibinfo {year} {2000})}\BibitemShut {NoStop}%
\bibitem [{\citenamefont {Gour}\ and\ \citenamefont
  {Sanders}(2004)}]{Gour2004}%
  \BibitemOpen
  \bibfield  {author} {\bibinfo {author} {\bibfnamefont {G.}~\bibnamefont
  {Gour}}\ and\ \bibinfo {author} {\bibfnamefont {B.~C.}\ \bibnamefont
  {Sanders}},\ }\bibfield  {title} {\bibinfo {title} {Remote preparation and
  distribution of bipartite entangled states},\ }\href
  {https://doi.org/10.1103/PhysRevLett.93.260501} {\bibfield  {journal}
  {\bibinfo  {journal} {Phys. Rev. Lett.}\ }\textbf {\bibinfo {volume} {93}},\
  \bibinfo {pages} {260501} (\bibinfo {year} {2004})}\BibitemShut {NoStop}%
\bibitem [{\citenamefont {Ac{\'{\i}}n}\ \emph {et~al.}(2007)\citenamefont
  {Ac{\'{\i}}n}, \citenamefont {Cirac},\ and\ \citenamefont
  {Lewenstein}}]{Acin2007}%
  \BibitemOpen
  \bibfield  {author} {\bibinfo {author} {\bibfnamefont {A.}~\bibnamefont
  {Ac{\'{\i}}n}}, \bibinfo {author} {\bibfnamefont {J.~I.}\ \bibnamefont
  {Cirac}},\ and\ \bibinfo {author} {\bibfnamefont {M.}~\bibnamefont
  {Lewenstein}},\ }\bibfield  {title} {\bibinfo {title} {Entanglement
  percolation in quantum~networks},\ }\href {https://doi.org/10.1038/nphys549}
  {\bibfield  {journal} {\bibinfo  {journal} {Nature Physics}\ }\textbf
  {\bibinfo {volume} {3}},\ \bibinfo {pages} {256} (\bibinfo {year}
  {2007})}\BibitemShut {NoStop}%
\bibitem [{\citenamefont {Perseguers}\ \emph {et~al.}(2008)\citenamefont
  {Perseguers}, \citenamefont {Cirac}, \citenamefont {Ac\'{\i}n}, \citenamefont
  {Lewenstein},\ and\ \citenamefont {Wehr}}]{Perseguers2008}%
  \BibitemOpen
  \bibfield  {author} {\bibinfo {author} {\bibfnamefont {S.}~\bibnamefont
  {Perseguers}}, \bibinfo {author} {\bibfnamefont {J.~I.}\ \bibnamefont
  {Cirac}}, \bibinfo {author} {\bibfnamefont {A.}~\bibnamefont {Ac\'{\i}n}},
  \bibinfo {author} {\bibfnamefont {M.}~\bibnamefont {Lewenstein}},\ and\
  \bibinfo {author} {\bibfnamefont {J.}~\bibnamefont {Wehr}},\ }\bibfield
  {title} {\bibinfo {title} {Entanglement distribution in pure-state quantum
  networks},\ }\href {https://doi.org/10.1103/PhysRevA.77.022308} {\bibfield
  {journal} {\bibinfo  {journal} {Phys. Rev. A}\ }\textbf {\bibinfo {volume}
  {77}},\ \bibinfo {pages} {022308} (\bibinfo {year} {2008})}\BibitemShut
  {NoStop}%
\bibitem [{\citenamefont {Cuquet}\ and\ \citenamefont
  {Calsamiglia}(2009)}]{Cuquet2009}%
  \BibitemOpen
  \bibfield  {author} {\bibinfo {author} {\bibfnamefont {M.}~\bibnamefont
  {Cuquet}}\ and\ \bibinfo {author} {\bibfnamefont {J.}~\bibnamefont
  {Calsamiglia}},\ }\bibfield  {title} {\bibinfo {title} {Entanglement
  percolation in quantum complex networks},\ }\href
  {https://doi.org/10.1103/PhysRevLett.103.240503} {\bibfield  {journal}
  {\bibinfo  {journal} {Phys. Rev. Lett.}\ }\textbf {\bibinfo {volume} {103}},\
  \bibinfo {pages} {240503} (\bibinfo {year} {2009})}\BibitemShut {NoStop}%
\bibitem [{\citenamefont {Perseguers}\ \emph {et~al.}(2010)\citenamefont
  {Perseguers}, \citenamefont {Cavalcanti}, \citenamefont {Lapeyre},
  \citenamefont {Lewenstein},\ and\ \citenamefont
  {Ac\'{\i}n}}]{Perseguers2010}%
  \BibitemOpen
  \bibfield  {author} {\bibinfo {author} {\bibfnamefont {S.}~\bibnamefont
  {Perseguers}}, \bibinfo {author} {\bibfnamefont {D.}~\bibnamefont
  {Cavalcanti}}, \bibinfo {author} {\bibfnamefont {G.~J.}\ \bibnamefont
  {Lapeyre}}, \bibinfo {author} {\bibfnamefont {M.}~\bibnamefont
  {Lewenstein}},\ and\ \bibinfo {author} {\bibfnamefont {A.}~\bibnamefont
  {Ac\'{\i}n}},\ }\bibfield  {title} {\bibinfo {title} {Multipartite
  entanglement percolation},\ }\href
  {https://doi.org/10.1103/PhysRevA.81.032327} {\bibfield  {journal} {\bibinfo
  {journal} {Phys. Rev. A}\ }\textbf {\bibinfo {volume} {81}},\ \bibinfo
  {pages} {032327} (\bibinfo {year} {2010})}\BibitemShut {NoStop}%
\bibitem [{\citenamefont {Meng}\ \emph {et~al.}(2021)\citenamefont {Meng},
  \citenamefont {Gao},\ and\ \citenamefont {Havlin}}]{Meng2021}%
  \BibitemOpen
  \bibfield  {author} {\bibinfo {author} {\bibfnamefont {X.}~\bibnamefont
  {Meng}}, \bibinfo {author} {\bibfnamefont {J.}~\bibnamefont {Gao}},\ and\
  \bibinfo {author} {\bibfnamefont {S.}~\bibnamefont {Havlin}},\ }\bibfield
  {title} {\bibinfo {title} {Concurrence percolation in quantum networks},\
  }\href {https://doi.org/10.1103/PhysRevLett.126.170501} {\bibfield  {journal}
  {\bibinfo  {journal} {Phys. Rev. Lett.}\ }\textbf {\bibinfo {volume} {126}},\
  \bibinfo {pages} {170501} (\bibinfo {year} {2021})}\BibitemShut {NoStop}%
\bibitem [{\citenamefont {Riera-S\`abat}\ \emph {et~al.}(2021)\citenamefont
  {Riera-S\`abat}, \citenamefont {Sekatski}, \citenamefont {Pirker},\ and\
  \citenamefont {D\"ur}}]{Riera2021}%
  \BibitemOpen
  \bibfield  {author} {\bibinfo {author} {\bibfnamefont {F.}~\bibnamefont
  {Riera-S\`abat}}, \bibinfo {author} {\bibfnamefont {P.}~\bibnamefont
  {Sekatski}}, \bibinfo {author} {\bibfnamefont {A.}~\bibnamefont {Pirker}},\
  and\ \bibinfo {author} {\bibfnamefont {W.}~\bibnamefont {D\"ur}},\ }\bibfield
   {title} {\bibinfo {title} {Entanglement-assisted entanglement
  purification},\ }\href {https://doi.org/10.1103/PhysRevLett.127.040502}
  {\bibfield  {journal} {\bibinfo  {journal} {Phys. Rev. Lett.}\ }\textbf
  {\bibinfo {volume} {127}},\ \bibinfo {pages} {040502} (\bibinfo {year}
  {2021})}\BibitemShut {NoStop}%
\bibitem [{\citenamefont {Horodecki}\ \emph {et~al.}(1999)\citenamefont
  {Horodecki}, \citenamefont {Horodecki},\ and\ \citenamefont
  {Horodecki}}]{Horodecki1999}%
  \BibitemOpen
  \bibfield  {author} {\bibinfo {author} {\bibfnamefont {M.}~\bibnamefont
  {Horodecki}}, \bibinfo {author} {\bibfnamefont {P.}~\bibnamefont
  {Horodecki}},\ and\ \bibinfo {author} {\bibfnamefont {R.}~\bibnamefont
  {Horodecki}},\ }\bibfield  {title} {\bibinfo {title} {General teleportation
  channel, singlet fraction, and quasidistillation},\ }\href
  {https://doi.org/10.1103/PhysRevA.60.1888} {\bibfield  {journal} {\bibinfo
  {journal} {Phys. Rev. A}\ }\textbf {\bibinfo {volume} {60}},\ \bibinfo
  {pages} {1888} (\bibinfo {year} {1999})}\BibitemShut {NoStop}%
\bibitem [{\citenamefont {Badzia\ifmmode~\mbox{\c{}}\else \c{}\fi{}g}\ \emph
  {et~al.}(2000)\citenamefont {Badzia\ifmmode~\mbox{\c{}}\else \c{}\fi{}g},
  \citenamefont {Horodecki}, \citenamefont {Horodecki},\ and\ \citenamefont
  {Horodecki}}]{Badzia2000}%
  \BibitemOpen
  \bibfield  {author} {\bibinfo {author} {\bibfnamefont {P.}~\bibnamefont
  {Badzia\ifmmode~\mbox{\c{}}\else \c{}\fi{}g}}, \bibinfo {author}
  {\bibfnamefont {M.}~\bibnamefont {Horodecki}}, \bibinfo {author}
  {\bibfnamefont {P.}~\bibnamefont {Horodecki}},\ and\ \bibinfo {author}
  {\bibfnamefont {R.}~\bibnamefont {Horodecki}},\ }\bibfield  {title} {\bibinfo
  {title} {Local environment can enhance fidelity of quantum teleportation},\
  }\href {https://doi.org/10.1103/PhysRevA.62.012311} {\bibfield  {journal}
  {\bibinfo  {journal} {Phys. Rev. A}\ }\textbf {\bibinfo {volume} {62}},\
  \bibinfo {pages} {012311} (\bibinfo {year} {2000})}\BibitemShut {NoStop}%
\bibitem [{\citenamefont {Verstraete}\ \emph {et~al.}(2001)\citenamefont
  {Verstraete}, \citenamefont {Dehaene},\ and\ \citenamefont
  {DeMoor}}]{Verstraete2001}%
  \BibitemOpen
  \bibfield  {author} {\bibinfo {author} {\bibfnamefont {F.}~\bibnamefont
  {Verstraete}}, \bibinfo {author} {\bibfnamefont {J.}~\bibnamefont
  {Dehaene}},\ and\ \bibinfo {author} {\bibfnamefont {B.}~\bibnamefont
  {DeMoor}},\ }\bibfield  {title} {\bibinfo {title} {Local filtering operations
  on two qubits},\ }\href {https://doi.org/10.1103/PhysRevA.64.010101}
  {\bibfield  {journal} {\bibinfo  {journal} {Phys. Rev. A}\ }\textbf {\bibinfo
  {volume} {64}},\ \bibinfo {pages} {010101} (\bibinfo {year}
  {2001})}\BibitemShut {NoStop}%
\bibitem [{\citenamefont {Verstraete}\ and\ \citenamefont
  {Verschelde}(2002)}]{Verstraete2002}%
  \BibitemOpen
  \bibfield  {author} {\bibinfo {author} {\bibfnamefont {F.}~\bibnamefont
  {Verstraete}}\ and\ \bibinfo {author} {\bibfnamefont {H.}~\bibnamefont
  {Verschelde}},\ }\bibfield  {title} {\bibinfo {title} {Fidelity of mixed
  states of two qubits},\ }\href {https://doi.org/10.1103/PhysRevA.66.022307}
  {\bibfield  {journal} {\bibinfo  {journal} {Phys. Rev. A}\ }\textbf {\bibinfo
  {volume} {66}},\ \bibinfo {pages} {022307} (\bibinfo {year}
  {2002})}\BibitemShut {NoStop}%
\bibitem [{\citenamefont {Verstraete}\ and\ \citenamefont
  {Verschelde}(2003)}]{Verstraete2003}%
  \BibitemOpen
  \bibfield  {author} {\bibinfo {author} {\bibfnamefont {F.}~\bibnamefont
  {Verstraete}}\ and\ \bibinfo {author} {\bibfnamefont {H.}~\bibnamefont
  {Verschelde}},\ }\bibfield  {title} {\bibinfo {title} {Optimal teleportation
  with a mixed state of two qubits},\ }\href
  {https://doi.org/10.1103/PhysRevLett.90.097901} {\bibfield  {journal}
  {\bibinfo  {journal} {Phys. Rev. Lett.}\ }\textbf {\bibinfo {volume} {90}},\
  \bibinfo {pages} {097901} (\bibinfo {year} {2003})}\BibitemShut {NoStop}%
\bibitem [{\citenamefont {Horodecki}\ \emph
  {et~al.}(1996{\natexlab{a}})\citenamefont {Horodecki}, \citenamefont
  {Horodecki},\ and\ \citenamefont {Horodecki}}]{Horodecki-PPT1996}%
  \BibitemOpen
  \bibfield  {author} {\bibinfo {author} {\bibfnamefont {M.}~\bibnamefont
  {Horodecki}}, \bibinfo {author} {\bibfnamefont {P.}~\bibnamefont
  {Horodecki}},\ and\ \bibinfo {author} {\bibfnamefont {R.}~\bibnamefont
  {Horodecki}},\ }\bibfield  {title} {\bibinfo {title} {Separability of mixed
  states: necessary and sufficient conditions},\ }\href
  {https://doi.org/https://doi.org/10.1016/S0375-9601(96)00706-2} {\bibfield
  {journal} {\bibinfo  {journal} {Physics Letters A}\ }\textbf {\bibinfo
  {volume} {223}},\ \bibinfo {pages} {1} (\bibinfo {year}
  {1996}{\natexlab{a}})}\BibitemShut {NoStop}%
\bibitem [{\citenamefont {Chitambar}\ and\ \citenamefont
  {Leditzky}(2024)}]{Chitambar2024}%
  \BibitemOpen
  \bibfield  {author} {\bibinfo {author} {\bibfnamefont {E.}~\bibnamefont
  {Chitambar}}\ and\ \bibinfo {author} {\bibfnamefont {F.}~\bibnamefont
  {Leditzky}},\ }\bibfield  {title} {\bibinfo {title} {On the duality of
  teleportation and dense coding},\ }\href
  {https://doi.org/10.1109/TIT.2023.3331821} {\bibfield  {journal} {\bibinfo
  {journal} {IEEE Transactions on Information Theory}\ }\textbf {\bibinfo
  {volume} {70}},\ \bibinfo {pages} {3529} (\bibinfo {year}
  {2024})}\BibitemShut {NoStop}%
\bibitem [{\citenamefont {Wootters}(1998)}]{Wootters1998}%
  \BibitemOpen
  \bibfield  {author} {\bibinfo {author} {\bibfnamefont {W.~K.}\ \bibnamefont
  {Wootters}},\ }\bibfield  {title} {\bibinfo {title} {Entanglement of
  formation of an arbitrary state of two qubits},\ }\href
  {https://doi.org/10.1103/PhysRevLett.80.2245} {\bibfield  {journal} {\bibinfo
   {journal} {Phys. Rev. Lett.}\ }\textbf {\bibinfo {volume} {80}},\ \bibinfo
  {pages} {2245} (\bibinfo {year} {1998})}\BibitemShut {NoStop}%
\bibitem [{\citenamefont {Ghosal}\ \emph {et~al.}(2020)\citenamefont {Ghosal},
  \citenamefont {Das}, \citenamefont {Roy},\ and\ \citenamefont
  {Bandyopadhyay}}]{Ghosal2020}%
  \BibitemOpen
  \bibfield  {author} {\bibinfo {author} {\bibfnamefont {A.}~\bibnamefont
  {Ghosal}}, \bibinfo {author} {\bibfnamefont {D.}~\bibnamefont {Das}},
  \bibinfo {author} {\bibfnamefont {S.}~\bibnamefont {Roy}},\ and\ \bibinfo
  {author} {\bibfnamefont {S.}~\bibnamefont {Bandyopadhyay}},\ }\bibfield
  {title} {\bibinfo {title} {Optimal two-qubit states for quantum teleportation
  vis-\`a-vis state properties},\ }\href
  {https://doi.org/10.1103/PhysRevA.101.012304} {\bibfield  {journal} {\bibinfo
   {journal} {Phys. Rev. A}\ }\textbf {\bibinfo {volume} {101}},\ \bibinfo
  {pages} {012304} (\bibinfo {year} {2020})}\BibitemShut {NoStop}%
\bibitem [{\citenamefont {Bandyopadhyay}\ and\ \citenamefont
  {Ghosh}(2012)}]{Somshubhro2012}%
  \BibitemOpen
  \bibfield  {author} {\bibinfo {author} {\bibfnamefont {S.}~\bibnamefont
  {Bandyopadhyay}}\ and\ \bibinfo {author} {\bibfnamefont {A.}~\bibnamefont
  {Ghosh}},\ }\bibfield  {title} {\bibinfo {title} {Optimal fidelity for a
  quantum channel may be attained by nonmaximally entangled states},\ }\href
  {https://doi.org/10.1103/PhysRevA.86.020304} {\bibfield  {journal} {\bibinfo
  {journal} {Phys. Rev. A}\ }\textbf {\bibinfo {volume} {86}},\ \bibinfo
  {pages} {020304} (\bibinfo {year} {2012})}\BibitemShut {NoStop}%
\bibitem [{\citenamefont {Pal}\ and\ \citenamefont
  {Bandyopadhyay}(2018)}]{Pal2018}%
  \BibitemOpen
  \bibfield  {author} {\bibinfo {author} {\bibfnamefont {R.}~\bibnamefont
  {Pal}}\ and\ \bibinfo {author} {\bibfnamefont {S.}~\bibnamefont
  {Bandyopadhyay}},\ }\bibfield  {title} {\bibinfo {title} {Entanglement
  sharing via qudit channels: Nonmaximally entangled states may be necessary
  for one-shot optimal singlet fraction and negativity},\ }\href
  {https://doi.org/10.1103/PhysRevA.97.032322} {\bibfield  {journal} {\bibinfo
  {journal} {Phys. Rev. A}\ }\textbf {\bibinfo {volume} {97}},\ \bibinfo
  {pages} {032322} (\bibinfo {year} {2018})}\BibitemShut {NoStop}%
\bibitem [{\citenamefont {Ghosal}\ \emph {et~al.}(2024)\citenamefont {Ghosal},
  \citenamefont {Ghai}, \citenamefont {Saha}, \citenamefont {Alimuddin},\ and\
  \citenamefont {Ghosh}}]{ghosal2024optimal}%
  \BibitemOpen
  \bibfield  {author} {\bibinfo {author} {\bibfnamefont {A.}~\bibnamefont
  {Ghosal}}, \bibinfo {author} {\bibfnamefont {J.}~\bibnamefont {Ghai}},
  \bibinfo {author} {\bibfnamefont {T.}~\bibnamefont {Saha}}, \bibinfo {author}
  {\bibfnamefont {M.}~\bibnamefont {Alimuddin}},\ and\ \bibinfo {author}
  {\bibfnamefont {S.}~\bibnamefont {Ghosh}},\ }\href@noop {} {\bibinfo {title}
  {Optimal quantum teleportation of collaboration}} (\bibinfo {year} {2024}),\
  \Eprint {https://arxiv.org/abs/2401.17201} {arXiv:2401.17201 [quant-ph]}
  \BibitemShut {NoStop}%
\bibitem [{sup()}]{supp}%
  \BibitemOpen
  \href@noop {} {}\bibinfo {note} {See Supplemental Material for technical
  details and proofs, which contains additional Refs. [75-81].}\BibitemShut
  {Stop}%
\bibitem [{\citenamefont {Bennett}\ and\ \citenamefont
  {Brassard}(2014)}]{Bennett1984}%
  \BibitemOpen
  \bibfield  {author} {\bibinfo {author} {\bibfnamefont {C.~H.}\ \bibnamefont
  {Bennett}}\ and\ \bibinfo {author} {\bibfnamefont {G.}~\bibnamefont
  {Brassard}},\ }\bibfield  {title} {\bibinfo {title} {Quantum cryptography:
  Public key distribution and coin tossing},\ }\href
  {https://doi.org/https://doi.org/10.1016/j.tcs.2014.05.025} {\bibfield
  {journal} {\bibinfo  {journal} {Theoretical Computer Science}\ }\textbf
  {\bibinfo {volume} {560}},\ \bibinfo {pages} {7} (\bibinfo {year} {2014})},\
  \bibinfo {note} {theoretical Aspects of Quantum Cryptography – celebrating
  30 years of BB84}\BibitemShut {NoStop}%
\bibitem [{\citenamefont {Ekert}(1991)}]{Ekert1991}%
  \BibitemOpen
  \bibfield  {author} {\bibinfo {author} {\bibfnamefont {A.~K.}\ \bibnamefont
  {Ekert}},\ }\bibfield  {title} {\bibinfo {title} {Quantum cryptography based
  on bell's theorem},\ }\href {https://doi.org/10.1103/PhysRevLett.67.661}
  {\bibfield  {journal} {\bibinfo  {journal} {Phys. Rev. Lett.}\ }\textbf
  {\bibinfo {volume} {67}},\ \bibinfo {pages} {661} (\bibinfo {year}
  {1991})}\BibitemShut {NoStop}%
\bibitem [{\citenamefont {Gisin}\ \emph {et~al.}(2002)\citenamefont {Gisin},
  \citenamefont {Ribordy}, \citenamefont {Tittel},\ and\ \citenamefont
  {Zbinden}}]{Gisin2002}%
  \BibitemOpen
  \bibfield  {author} {\bibinfo {author} {\bibfnamefont {N.}~\bibnamefont
  {Gisin}}, \bibinfo {author} {\bibfnamefont {G.}~\bibnamefont {Ribordy}},
  \bibinfo {author} {\bibfnamefont {W.}~\bibnamefont {Tittel}},\ and\ \bibinfo
  {author} {\bibfnamefont {H.}~\bibnamefont {Zbinden}},\ }\bibfield  {title}
  {\bibinfo {title} {Quantum cryptography},\ }\href
  {https://doi.org/10.1103/RevModPhys.74.145} {\bibfield  {journal} {\bibinfo
  {journal} {Rev. Mod. Phys.}\ }\textbf {\bibinfo {volume} {74}},\ \bibinfo
  {pages} {145} (\bibinfo {year} {2002})}\BibitemShut {NoStop}%
\bibitem [{\citenamefont {Yin}\ \emph {et~al.}(2020)\citenamefont {Yin},
  \citenamefont {Li}, \citenamefont {Liao}, \citenamefont {Yang}, \citenamefont
  {Cao}, \citenamefont {Zhang}, \citenamefont {Ren}, \citenamefont {Cai},
  \citenamefont {Liu}, \citenamefont {Li}, \citenamefont {Shu}, \citenamefont
  {Huang}, \citenamefont {Deng}, \citenamefont {Li}, \citenamefont {Zhang},
  \citenamefont {Liu}, \citenamefont {Chen}, \citenamefont {Lu}, \citenamefont
  {Wang}, \citenamefont {Xu}, \citenamefont {Wang}, \citenamefont {Peng},
  \citenamefont {Ekert},\ and\ \citenamefont {Pan}}]{Yin2020}%
  \BibitemOpen
  \bibfield  {author} {\bibinfo {author} {\bibfnamefont {J.}~\bibnamefont
  {Yin}}, \bibinfo {author} {\bibfnamefont {Y.-H.}\ \bibnamefont {Li}},
  \bibinfo {author} {\bibfnamefont {S.-K.}\ \bibnamefont {Liao}}, \bibinfo
  {author} {\bibfnamefont {M.}~\bibnamefont {Yang}}, \bibinfo {author}
  {\bibfnamefont {Y.}~\bibnamefont {Cao}}, \bibinfo {author} {\bibfnamefont
  {L.}~\bibnamefont {Zhang}}, \bibinfo {author} {\bibfnamefont {J.-G.}\
  \bibnamefont {Ren}}, \bibinfo {author} {\bibfnamefont {W.-Q.}\ \bibnamefont
  {Cai}}, \bibinfo {author} {\bibfnamefont {W.-Y.}\ \bibnamefont {Liu}},
  \bibinfo {author} {\bibfnamefont {S.-L.}\ \bibnamefont {Li}}, \bibinfo
  {author} {\bibfnamefont {R.}~\bibnamefont {Shu}}, \bibinfo {author}
  {\bibfnamefont {Y.-M.}\ \bibnamefont {Huang}}, \bibinfo {author}
  {\bibfnamefont {L.}~\bibnamefont {Deng}}, \bibinfo {author} {\bibfnamefont
  {L.}~\bibnamefont {Li}}, \bibinfo {author} {\bibfnamefont {Q.}~\bibnamefont
  {Zhang}}, \bibinfo {author} {\bibfnamefont {N.-L.}\ \bibnamefont {Liu}},
  \bibinfo {author} {\bibfnamefont {Y.-A.}\ \bibnamefont {Chen}}, \bibinfo
  {author} {\bibfnamefont {C.-Y.}\ \bibnamefont {Lu}}, \bibinfo {author}
  {\bibfnamefont {X.-B.}\ \bibnamefont {Wang}}, \bibinfo {author}
  {\bibfnamefont {F.}~\bibnamefont {Xu}}, \bibinfo {author} {\bibfnamefont
  {J.-Y.}\ \bibnamefont {Wang}}, \bibinfo {author} {\bibfnamefont {C.-Z.}\
  \bibnamefont {Peng}}, \bibinfo {author} {\bibfnamefont {A.~K.}\ \bibnamefont
  {Ekert}},\ and\ \bibinfo {author} {\bibfnamefont {J.-W.}\ \bibnamefont
  {Pan}},\ }\bibfield  {title} {\bibinfo {title} {Entanglement-based secure
  quantum cryptography over 1, 120 kilometres},\ }\href
  {https://doi.org/10.1038/s41586-020-2401-y} {\bibfield  {journal} {\bibinfo
  {journal} {Nature}\ }\textbf {\bibinfo {volume} {582}},\ \bibinfo {pages}
  {501–505} (\bibinfo {year} {2020})}\BibitemShut {NoStop}%
\bibitem [{\citenamefont {Hu}\ \emph {et~al.}(2023)\citenamefont {Hu},
  \citenamefont {Guo}, \citenamefont {Liu}, \citenamefont {Li},\ and\
  \citenamefont {Guo}}]{Hu2023}%
  \BibitemOpen
  \bibfield  {author} {\bibinfo {author} {\bibfnamefont {X.-M.}\ \bibnamefont
  {Hu}}, \bibinfo {author} {\bibfnamefont {Y.}~\bibnamefont {Guo}}, \bibinfo
  {author} {\bibfnamefont {B.-H.}\ \bibnamefont {Liu}}, \bibinfo {author}
  {\bibfnamefont {C.-F.}\ \bibnamefont {Li}},\ and\ \bibinfo {author}
  {\bibfnamefont {G.-C.}\ \bibnamefont {Guo}},\ }\bibfield  {title} {\bibinfo
  {title} {Progress in quantum teleportation},\ }\href
  {https://doi.org/10.1038/s42254-023-00588-x} {\bibfield  {journal} {\bibinfo
  {journal} {Nature Reviews Physics}\ }\textbf {\bibinfo {volume} {5}},\
  \bibinfo {pages} {339} (\bibinfo {year} {2023})}\BibitemShut {NoStop}%
\bibitem [{\citenamefont {Buhrman}\ \emph {et~al.}(2014)\citenamefont
  {Buhrman}, \citenamefont {Chandran}, \citenamefont {Fehr}, \citenamefont
  {Gelles}, \citenamefont {Goyal}, \citenamefont {Ostrovsky},\ and\
  \citenamefont {Schaffner}}]{Buhrman2014}%
  \BibitemOpen
  \bibfield  {author} {\bibinfo {author} {\bibfnamefont {H.}~\bibnamefont
  {Buhrman}}, \bibinfo {author} {\bibfnamefont {N.}~\bibnamefont {Chandran}},
  \bibinfo {author} {\bibfnamefont {S.}~\bibnamefont {Fehr}}, \bibinfo {author}
  {\bibfnamefont {R.}~\bibnamefont {Gelles}}, \bibinfo {author} {\bibfnamefont
  {V.}~\bibnamefont {Goyal}}, \bibinfo {author} {\bibfnamefont
  {R.}~\bibnamefont {Ostrovsky}},\ and\ \bibinfo {author} {\bibfnamefont
  {C.}~\bibnamefont {Schaffner}},\ }\bibfield  {title} {\bibinfo {title}
  {Position-based quantum cryptography: Impossibility and constructions},\
  }\href {https://doi.org/10.1137/130913687} {\bibfield  {journal} {\bibinfo
  {journal} {SIAM Journal on Computing}\ }\textbf {\bibinfo {volume} {43}},\
  \bibinfo {pages} {150} (\bibinfo {year} {2014})},\ \Eprint
  {https://arxiv.org/abs/https://doi.org/10.1137/130913687}
  {https://doi.org/10.1137/130913687} \BibitemShut {NoStop}%
\bibitem [{\citenamefont {Kimble}(2008)}]{Kimble2008}%
  \BibitemOpen
  \bibfield  {author} {\bibinfo {author} {\bibfnamefont {H.~J.}\ \bibnamefont
  {Kimble}},\ }\bibfield  {title} {\bibinfo {title} {The quantum internet},\
  }\href {https://doi.org/https://doi.org/10.1038/nature07127} {\bibfield
  {journal} {\bibinfo  {journal} {Nature}\ }\textbf {\bibinfo {volume} {453}},\
  \bibinfo {pages} {1023} (\bibinfo {year} {2008})}\BibitemShut {NoStop}%
\bibitem [{\citenamefont {Wehner}\ \emph {et~al.}(2018)\citenamefont {Wehner},
  \citenamefont {Elkouss},\ and\ \citenamefont {Hanson}}]{Wehner2018}%
  \BibitemOpen
  \bibfield  {author} {\bibinfo {author} {\bibfnamefont {S.}~\bibnamefont
  {Wehner}}, \bibinfo {author} {\bibfnamefont {D.}~\bibnamefont {Elkouss}},\
  and\ \bibinfo {author} {\bibfnamefont {R.}~\bibnamefont {Hanson}},\
  }\bibfield  {title} {\bibinfo {title} {Quantum internet: A vision for the
  road ahead},\ }\href {https://doi.org/10.1126/science.aam9288} {\bibfield
  {journal} {\bibinfo  {journal} {Science}\ }\textbf {\bibinfo {volume}
  {362}},\ \bibinfo {pages} {eaam9288} (\bibinfo {year} {2018})},\ \Eprint
  {https://arxiv.org/abs/https://www.science.org/doi/pdf/10.1126/science.aam9288}
  {https://www.science.org/doi/pdf/10.1126/science.aam9288} \BibitemShut
  {NoStop}%
\bibitem [{\citenamefont {Cirac}\ \emph {et~al.}(1999)\citenamefont {Cirac},
  \citenamefont {Ekert}, \citenamefont {Huelga},\ and\ \citenamefont
  {Macchiavello}}]{Cirac1999}%
  \BibitemOpen
  \bibfield  {author} {\bibinfo {author} {\bibfnamefont {J.~I.}\ \bibnamefont
  {Cirac}}, \bibinfo {author} {\bibfnamefont {A.~K.}\ \bibnamefont {Ekert}},
  \bibinfo {author} {\bibfnamefont {S.~F.}\ \bibnamefont {Huelga}},\ and\
  \bibinfo {author} {\bibfnamefont {C.}~\bibnamefont {Macchiavello}},\
  }\bibfield  {title} {\bibinfo {title} {Distributed quantum computation over
  noisy channels},\ }\href {https://doi.org/10.1103/PhysRevA.59.4249}
  {\bibfield  {journal} {\bibinfo  {journal} {Phys. Rev. A}\ }\textbf {\bibinfo
  {volume} {59}},\ \bibinfo {pages} {4249} (\bibinfo {year}
  {1999})}\BibitemShut {NoStop}%
\bibitem [{\citenamefont {Caleffi}\ \emph {et~al.}(2024)\citenamefont
  {Caleffi}, \citenamefont {Amoretti}, \citenamefont {Ferrari}, \citenamefont
  {Illiano}, \citenamefont {Manzalini},\ and\ \citenamefont
  {Cacciapuoti}}]{CALEFFI2024110672}%
  \BibitemOpen
  \bibfield  {author} {\bibinfo {author} {\bibfnamefont {M.}~\bibnamefont
  {Caleffi}}, \bibinfo {author} {\bibfnamefont {M.}~\bibnamefont {Amoretti}},
  \bibinfo {author} {\bibfnamefont {D.}~\bibnamefont {Ferrari}}, \bibinfo
  {author} {\bibfnamefont {J.}~\bibnamefont {Illiano}}, \bibinfo {author}
  {\bibfnamefont {A.}~\bibnamefont {Manzalini}},\ and\ \bibinfo {author}
  {\bibfnamefont {A.~S.}\ \bibnamefont {Cacciapuoti}},\ }\bibfield  {title}
  {\bibinfo {title} {Distributed quantum computing: A survey},\ }\href
  {https://doi.org/https://doi.org/10.1016/j.comnet.2024.110672} {\bibfield
  {journal} {\bibinfo  {journal} {Computer Networks}\ }\textbf {\bibinfo
  {volume} {254}},\ \bibinfo {pages} {110672} (\bibinfo {year}
  {2024})}\BibitemShut {NoStop}%
\bibitem [{\citenamefont {Main}\ \emph {et~al.}(2025)\citenamefont {Main},
  \citenamefont {Drmota}, \citenamefont {Nadlinger}, \citenamefont {Ainley},
  \citenamefont {Agrawal}, \citenamefont {Nichol}, \citenamefont {Srinivas},
  \citenamefont {Araneda},\ and\ \citenamefont {Lucas}}]{Main2025}%
  \BibitemOpen
  \bibfield  {author} {\bibinfo {author} {\bibfnamefont {D.}~\bibnamefont
  {Main}}, \bibinfo {author} {\bibfnamefont {P.}~\bibnamefont {Drmota}},
  \bibinfo {author} {\bibfnamefont {D.}~\bibnamefont {Nadlinger}}, \bibinfo
  {author} {\bibfnamefont {E.}~\bibnamefont {Ainley}}, \bibinfo {author}
  {\bibfnamefont {A.}~\bibnamefont {Agrawal}}, \bibinfo {author} {\bibfnamefont
  {B.}~\bibnamefont {Nichol}}, \bibinfo {author} {\bibfnamefont
  {R.}~\bibnamefont {Srinivas}}, \bibinfo {author} {\bibfnamefont
  {G.}~\bibnamefont {Araneda}},\ and\ \bibinfo {author} {\bibfnamefont
  {D.}~\bibnamefont {Lucas}},\ }\bibfield  {title} {\bibinfo {title}
  {Distributed quantum computing across an optical network link},\ }\href
  {https://doi.org/10.1038/s41586-024-08404-x} {\bibfield  {journal} {\bibinfo
  {journal} {Nature}\ }\textbf {\bibinfo {volume} {638}},\ \bibinfo {pages}
  {383} (\bibinfo {year} {2025})}\BibitemShut {NoStop}%
\bibitem [{\citenamefont {Gottesman}\ and\ \citenamefont
  {Chuang}(1999)}]{Gottesman1999}%
  \BibitemOpen
  \bibfield  {author} {\bibinfo {author} {\bibfnamefont {D.}~\bibnamefont
  {Gottesman}}\ and\ \bibinfo {author} {\bibfnamefont {I.~L.}\ \bibnamefont
  {Chuang}},\ }\bibfield  {title} {\bibinfo {title} {Demonstrating the
  viability of universal quantum computation using teleportation and
  single-qubit operations},\ }\href {https://doi.org/10.1038/46503} {\bibfield
  {journal} {\bibinfo  {journal} {Nature}\ }\textbf {\bibinfo {volume} {402}},\
  \bibinfo {pages} {390} (\bibinfo {year} {1999})}\BibitemShut {NoStop}%
\bibitem [{\citenamefont {Leung}(2002)}]{Leung2001}%
  \BibitemOpen
  \bibfield  {author} {\bibinfo {author} {\bibfnamefont {D.~W.}\ \bibnamefont
  {Leung}},\ }\href@noop {} {\bibinfo {title} {Two-qubit projective
  measurements are universal for quantum computation}} (\bibinfo {year}
  {2002}),\ \Eprint {https://arxiv.org/abs/quant-ph/0111122}
  {arXiv:quant-ph/0111122 [quant-ph]} \BibitemShut {NoStop}%
\bibitem [{\citenamefont {Nielsen}(2003)}]{Nielsen2003}%
  \BibitemOpen
  \bibfield  {author} {\bibinfo {author} {\bibfnamefont {M.~A.}\ \bibnamefont
  {Nielsen}},\ }\bibfield  {title} {\bibinfo {title} {Quantum computation by
  measurement and quantum memory},\ }\href
  {https://doi.org/https://doi.org/10.1016/S0375-9601(02)01803-0} {\bibfield
  {journal} {\bibinfo  {journal} {Physics Letters A}\ }\textbf {\bibinfo
  {volume} {308}},\ \bibinfo {pages} {96} (\bibinfo {year} {2003})}\BibitemShut
  {NoStop}%
\bibitem [{\citenamefont {Raussendorf}\ \emph {et~al.}(2003)\citenamefont
  {Raussendorf}, \citenamefont {Browne},\ and\ \citenamefont
  {Briegel}}]{Robert2003}%
  \BibitemOpen
  \bibfield  {author} {\bibinfo {author} {\bibfnamefont {R.}~\bibnamefont
  {Raussendorf}}, \bibinfo {author} {\bibfnamefont {D.~E.}\ \bibnamefont
  {Browne}},\ and\ \bibinfo {author} {\bibfnamefont {H.~J.}\ \bibnamefont
  {Briegel}},\ }\bibfield  {title} {\bibinfo {title} {Measurement-based quantum
  computation on cluster states},\ }\href
  {https://doi.org/10.1103/PhysRevA.68.022312} {\bibfield  {journal} {\bibinfo
  {journal} {Phys. Rev. A}\ }\textbf {\bibinfo {volume} {68}},\ \bibinfo
  {pages} {022312} (\bibinfo {year} {2003})}\BibitemShut {NoStop}%
\bibitem [{\citenamefont {Leung}(2004)}]{Leung2004}%
  \BibitemOpen
  \bibfield  {author} {\bibinfo {author} {\bibfnamefont {D.~W.}\ \bibnamefont
  {Leung}},\ }\bibfield  {title} {\bibinfo {title} {Quantum computation by
  measurements},\ }\href {https://doi.org/10.1142/S0219749904000055} {\bibfield
   {journal} {\bibinfo  {journal} {International Journal of Quantum
  Information}\ }\textbf {\bibinfo {volume} {2}},\ \bibinfo {pages} {33}
  (\bibinfo {year} {2004})}\BibitemShut {NoStop}%
\bibitem [{\citenamefont {Jozsa}(2006)}]{Jozsa2006}%
  \BibitemOpen
  \bibfield  {author} {\bibinfo {author} {\bibfnamefont {R.}~\bibnamefont
  {Jozsa}},\ }\bibfield  {title} {\bibinfo {title} {An introduction to
  measurement based quantum computation},\ }\href
  {https://books.google.co.in/books?hl=en&lr=&id=PgTvAgAAQBAJ&oi=fnd&pg=PA137&ots=0-pTSWsaOb&sig=19lbu8liMORhyuHaw96EDyxTbfw&redir_esc=y#v=onepage&q&f=false}
  {\bibfield  {journal} {\bibinfo  {journal} {NATO Science Series, III:
  Computer and Systems Sciences. Quantum Information Processing-From Theory to
  Experiment}\ }\textbf {\bibinfo {volume} {199}},\ \bibinfo {pages} {137}
  (\bibinfo {year} {2006})}\BibitemShut {NoStop}%
\bibitem [{\citenamefont {Broadbent}\ \emph {et~al.}(2009)\citenamefont
  {Broadbent}, \citenamefont {Fitzsimons},\ and\ \citenamefont
  {Kashefi}}]{Broadbent2009}%
  \BibitemOpen
  \bibfield  {author} {\bibinfo {author} {\bibfnamefont {A.}~\bibnamefont
  {Broadbent}}, \bibinfo {author} {\bibfnamefont {J.}~\bibnamefont
  {Fitzsimons}},\ and\ \bibinfo {author} {\bibfnamefont {E.}~\bibnamefont
  {Kashefi}},\ }\bibfield  {title} {\bibinfo {title} {Universal blind quantum
  computation},\ }in\ \href {https://doi.org/10.1109/FOCS.2009.36} {\emph
  {\bibinfo {booktitle} {2009 50th Annual IEEE Symposium on Foundations of
  Computer Science}}}\ (\bibinfo  {publisher} {IEEE},\ \bibinfo {year} {2009})\
  p.\ \bibinfo {pages} {517–526}\BibitemShut {NoStop}%
\bibitem [{\citenamefont {Curty}\ \emph {et~al.}(2019)\citenamefont {Curty},
  \citenamefont {Azuma},\ and\ \citenamefont {Lo}}]{Curty2019}%
  \BibitemOpen
  \bibfield  {author} {\bibinfo {author} {\bibfnamefont {M.}~\bibnamefont
  {Curty}}, \bibinfo {author} {\bibfnamefont {K.}~\bibnamefont {Azuma}},\ and\
  \bibinfo {author} {\bibfnamefont {H.-K.}\ \bibnamefont {Lo}},\ }\bibfield
  {title} {\bibinfo {title} {Simple security proof of twin-field type quantum
  key distribution protocol},\ }\href
  {https://doi.org/10.1038/s41534-019-0175-6} {\bibfield  {journal} {\bibinfo
  {journal} {npj Quantum Information}\ }\textbf {\bibinfo {volume} {5}},\
  \bibinfo {pages} {64} (\bibinfo {year} {2019})}\BibitemShut {NoStop}%
\bibitem [{\citenamefont {Li}\ \emph {et~al.}(2021)\citenamefont {Li},
  \citenamefont {Fang}, \citenamefont {Zhang}, \citenamefont {Tabia},
  \citenamefont {Lu},\ and\ \citenamefont {Liang}}]{Li2021}%
  \BibitemOpen
  \bibfield  {author} {\bibinfo {author} {\bibfnamefont {J.-Y.}\ \bibnamefont
  {Li}}, \bibinfo {author} {\bibfnamefont {X.-X.}\ \bibnamefont {Fang}},
  \bibinfo {author} {\bibfnamefont {T.}~\bibnamefont {Zhang}}, \bibinfo
  {author} {\bibfnamefont {G.~N.~M.}\ \bibnamefont {Tabia}}, \bibinfo {author}
  {\bibfnamefont {H.}~\bibnamefont {Lu}},\ and\ \bibinfo {author}
  {\bibfnamefont {Y.-C.}\ \bibnamefont {Liang}},\ }\bibfield  {title} {\bibinfo
  {title} {Activating hidden teleportation power: Theory and experiment},\
  }\href {https://doi.org/10.1103/PhysRevResearch.3.023045} {\bibfield
  {journal} {\bibinfo  {journal} {Phys. Rev. Res.}\ }\textbf {\bibinfo {volume}
  {3}},\ \bibinfo {pages} {023045} (\bibinfo {year} {2021})}\BibitemShut
  {NoStop}%
\bibitem [{\citenamefont {Scarani}\ \emph {et~al.}(2009)\citenamefont
  {Scarani}, \citenamefont {Bechmann-Pasquinucci}, \citenamefont {Cerf},
  \citenamefont {Du\ifmmode~\check{s}\else \v{s}\fi{}ek}, \citenamefont
  {L\"utkenhaus},\ and\ \citenamefont {Peev}}]{Scarani2009}%
  \BibitemOpen
  \bibfield  {author} {\bibinfo {author} {\bibfnamefont {V.}~\bibnamefont
  {Scarani}}, \bibinfo {author} {\bibfnamefont {H.}~\bibnamefont
  {Bechmann-Pasquinucci}}, \bibinfo {author} {\bibfnamefont {N.~J.}\
  \bibnamefont {Cerf}}, \bibinfo {author} {\bibfnamefont {M.}~\bibnamefont
  {Du\ifmmode~\check{s}\else \v{s}\fi{}ek}}, \bibinfo {author} {\bibfnamefont
  {N.}~\bibnamefont {L\"utkenhaus}},\ and\ \bibinfo {author} {\bibfnamefont
  {M.}~\bibnamefont {Peev}},\ }\bibfield  {title} {\bibinfo {title} {The
  security of practical quantum key distribution},\ }\href
  {https://doi.org/10.1103/RevModPhys.81.1301} {\bibfield  {journal} {\bibinfo
  {journal} {Rev. Mod. Phys.}\ }\textbf {\bibinfo {volume} {81}},\ \bibinfo
  {pages} {1301} (\bibinfo {year} {2009})}\BibitemShut {NoStop}%
\bibitem [{\citenamefont {K\l{}obus}\ \emph {et~al.}(2012)\citenamefont
  {K\l{}obus}, \citenamefont {Laskowski}, \citenamefont {Markiewicz},\ and\
  \citenamefont {Grudka}}]{Klobus2012}%
  \BibitemOpen
  \bibfield  {author} {\bibinfo {author} {\bibfnamefont {W.}~\bibnamefont
  {K\l{}obus}}, \bibinfo {author} {\bibfnamefont {W.}~\bibnamefont
  {Laskowski}}, \bibinfo {author} {\bibfnamefont {M.}~\bibnamefont
  {Markiewicz}},\ and\ \bibinfo {author} {\bibfnamefont {A.}~\bibnamefont
  {Grudka}},\ }\bibfield  {title} {\bibinfo {title} {Nonlocality activation in
  entanglement-swapping chains},\ }\href
  {https://doi.org/10.1103/PhysRevA.86.020302} {\bibfield  {journal} {\bibinfo
  {journal} {Phys. Rev. A}\ }\textbf {\bibinfo {volume} {86}},\ \bibinfo
  {pages} {020302} (\bibinfo {year} {2012})}\BibitemShut {NoStop}%
\bibitem [{\citenamefont {Azuma}\ \emph {et~al.}(2012)\citenamefont {Azuma},
  \citenamefont {Takeda}, \citenamefont {Koashi},\ and\ \citenamefont
  {Imoto}}]{Azuma2012}%
  \BibitemOpen
  \bibfield  {author} {\bibinfo {author} {\bibfnamefont {K.}~\bibnamefont
  {Azuma}}, \bibinfo {author} {\bibfnamefont {H.}~\bibnamefont {Takeda}},
  \bibinfo {author} {\bibfnamefont {M.}~\bibnamefont {Koashi}},\ and\ \bibinfo
  {author} {\bibfnamefont {N.}~\bibnamefont {Imoto}},\ }\bibfield  {title}
  {\bibinfo {title} {Quantum repeaters and computation by a single module:
  Remote nondestructive parity measurement},\ }\href
  {https://doi.org/10.1103/PhysRevA.85.062309} {\bibfield  {journal} {\bibinfo
  {journal} {Phys. Rev. A}\ }\textbf {\bibinfo {volume} {85}},\ \bibinfo
  {pages} {062309} (\bibinfo {year} {2012})}\BibitemShut {NoStop}%
\bibitem [{\citenamefont {Horodecki}\ \emph {et~al.}(1995)\citenamefont
  {Horodecki}, \citenamefont {Horodecki},\ and\ \citenamefont
  {Horodecki}}]{Horodecki1995}%
  \BibitemOpen
  \bibfield  {author} {\bibinfo {author} {\bibfnamefont {R.}~\bibnamefont
  {Horodecki}}, \bibinfo {author} {\bibfnamefont {P.}~\bibnamefont
  {Horodecki}},\ and\ \bibinfo {author} {\bibfnamefont {M.}~\bibnamefont
  {Horodecki}},\ }\bibfield  {title} {\bibinfo {title} {Violating bell
  inequality by mixed spin-12 states: necessary and sufficient condition},\
  }\href {https://doi.org/https://doi.org/10.1016/0375-9601(95)00214-N}
  {\bibfield  {journal} {\bibinfo  {journal} {Physics Letters A}\ }\textbf
  {\bibinfo {volume} {200}},\ \bibinfo {pages} {340} (\bibinfo {year}
  {1995})}\BibitemShut {NoStop}%
\bibitem [{\citenamefont {Horodecki}\ and\ \citenamefont
  {Horodecki}(1996)}]{Horodecki199654}%
  \BibitemOpen
  \bibfield  {author} {\bibinfo {author} {\bibfnamefont {R.}~\bibnamefont
  {Horodecki}}\ and\ \bibinfo {author} {\bibfnamefont {M.}~\bibnamefont
  {Horodecki}},\ }\bibfield  {title} {\bibinfo {title} {Information-theoretic
  aspects of inseparability of mixed states},\ }\href
  {https://doi.org/10.1103/PhysRevA.54.1838} {\bibfield  {journal} {\bibinfo
  {journal} {Phys. Rev. A}\ }\textbf {\bibinfo {volume} {54}},\ \bibinfo
  {pages} {1838} (\bibinfo {year} {1996})}\BibitemShut {NoStop}%
\bibitem [{\citenamefont {Horodecki}\ \emph
  {et~al.}(1996{\natexlab{b}})\citenamefont {Horodecki}, \citenamefont
  {Horodecki},\ and\ \citenamefont {Horodecki}}]{Horodecki199621}%
  \BibitemOpen
  \bibfield  {author} {\bibinfo {author} {\bibfnamefont {R.}~\bibnamefont
  {Horodecki}}, \bibinfo {author} {\bibfnamefont {M.}~\bibnamefont
  {Horodecki}},\ and\ \bibinfo {author} {\bibfnamefont {P.}~\bibnamefont
  {Horodecki}},\ }\bibfield  {title} {\bibinfo {title} {Teleportation, bell's
  inequalities and inseparability},\ }\href
  {https://doi.org/https://doi.org/10.1016/0375-9601(96)00639-1} {\bibfield
  {journal} {\bibinfo  {journal} {Physics Letters A}\ }\textbf {\bibinfo
  {volume} {222}},\ \bibinfo {pages} {21} (\bibinfo {year}
  {1996}{\natexlab{b}})}\BibitemShut {NoStop}%
\bibitem [{\citenamefont {Mod\l{}awska}\ and\ \citenamefont
  {Grudka}(2008)}]{PhysRevA.78.032321}%
  \BibitemOpen
  \bibfield  {author} {\bibinfo {author} {\bibfnamefont {J.}~\bibnamefont
  {Mod\l{}awska}}\ and\ \bibinfo {author} {\bibfnamefont {A.}~\bibnamefont
  {Grudka}},\ }\bibfield  {title} {\bibinfo {title} {Increasing singlet
  fraction with entanglement swapping},\ }\href
  {https://doi.org/10.1103/PhysRevA.78.032321} {\bibfield  {journal} {\bibinfo
  {journal} {Phys. Rev. A}\ }\textbf {\bibinfo {volume} {78}},\ \bibinfo
  {pages} {032321} (\bibinfo {year} {2008})}\BibitemShut {NoStop}%
\bibitem [{\citenamefont {Nielsen}\ and\ \citenamefont
  {Chuang}(2010)}]{Nielsen2010}%
  \BibitemOpen
  \bibfield  {author} {\bibinfo {author} {\bibfnamefont {M.~A.}\ \bibnamefont
  {Nielsen}}\ and\ \bibinfo {author} {\bibfnamefont {I.~L.}\ \bibnamefont
  {Chuang}},\ }\href@noop {} {\emph {\bibinfo {title} {Quantum Computation and
  Quantum Information}}}\ (\bibinfo  {publisher} {Cambridge University Press},\
  \bibinfo {year} {2010})\BibitemShut {NoStop}%
\bibitem [{\citenamefont {Riera-S{\`a}bat}\ \emph {et~al.}(2021)\citenamefont
  {Riera-S{\`a}bat}, \citenamefont {Sekatski}, \citenamefont {Pirker},\ and\
  \citenamefont {D{\"u}r}}]{riera2021entanglement}%
  \BibitemOpen
  \bibfield  {author} {\bibinfo {author} {\bibfnamefont {F.}~\bibnamefont
  {Riera-S{\`a}bat}}, \bibinfo {author} {\bibfnamefont {P.}~\bibnamefont
  {Sekatski}}, \bibinfo {author} {\bibfnamefont {A.}~\bibnamefont {Pirker}},\
  and\ \bibinfo {author} {\bibfnamefont {W.}~\bibnamefont {D{\"u}r}},\
  }\bibfield  {title} {\bibinfo {title} {Entanglement-assisted entanglement
  purification},\ }\href
  {https://journals.aps.org/prl/pdf/10.1103/PhysRevLett.127.040502} {\bibfield
  {journal} {\bibinfo  {journal} {Physical Review Letters}\ }\textbf {\bibinfo
  {volume} {127}},\ \bibinfo {pages} {040502} (\bibinfo {year}
  {2021})}\BibitemShut {NoStop}%
\bibitem [{\citenamefont {Bennett}\ \emph
  {et~al.}(1996{\natexlab{c}})\citenamefont {Bennett}, \citenamefont
  {Bernstein}, \citenamefont {Popescu},\ and\ \citenamefont
  {Schumacher}}]{BennettPure1996}%
  \BibitemOpen
  \bibfield  {author} {\bibinfo {author} {\bibfnamefont {C.~H.}\ \bibnamefont
  {Bennett}}, \bibinfo {author} {\bibfnamefont {H.~J.}\ \bibnamefont
  {Bernstein}}, \bibinfo {author} {\bibfnamefont {S.}~\bibnamefont {Popescu}},\
  and\ \bibinfo {author} {\bibfnamefont {B.}~\bibnamefont {Schumacher}},\
  }\bibfield  {title} {\bibinfo {title} {Concentrating partial entanglement by
  local operations},\ }\href {https://doi.org/10.1103/PhysRevA.53.2046}
  {\bibfield  {journal} {\bibinfo  {journal} {Phys. Rev. A}\ }\textbf {\bibinfo
  {volume} {53}},\ \bibinfo {pages} {2046} (\bibinfo {year}
  {1996}{\natexlab{c}})}\BibitemShut {NoStop}%
\end{thebibliography}%

 \onecolumngrid
\appendix
\section{Proof of Lemma 1} \label{AppendixA}
\begin{proof}
Let the segment $S_1$ shares a noisy two-qubit state $\rho_{1}$ with concurrence $\mathcal{C}_1$ and {\it FEF} 
\begin{align}
    F(\rho_1)=\dfrac{1+\mathcal{C}_1}{2}, \label{l1}
\end{align}
respectively. Clearly, Eq.$~$(\ref{l1}) implies that the optimal FEF value of $\rho_1$, {\it i.e.}, $F^\star(\rho_1)$ is equal with $F(\rho_1)$. Any two-qubit entangled state satisfying Eq.$~$(\ref{l1}) must hold the property \cite{Verstraete2002}
\begin{align}
    \rho_1^{\Gamma} \ket{\phi} =  \lambda_{min} \ket{\phi}, \label{l2}
\end{align}
where $\Gamma$ is the partial transposition in computational basis, and $\ket{\phi}$ is a maximally entangled state which is also the eigenvector corresponding to the smallest eigenvalue $\lambda_{min}$. Note that if $\rho_1$ is entangled then $\lambda_{min}$ must be negative and can also be expressed as $\lambda_{min}=-\frac{N(\rho_1)}{2}$, where $N(\rho_1)$ is the negativity of $\rho_1$ \cite{Verstraete2002}.  However, the negativity of any two-qubit state $\rho_1$ satisfying Eq.$~$(\ref{l2}) is equal with its concurrence $C(\rho_1)=\mathcal{C}_1$ \cite{Verstraete2002, Ghosal2020}. Hence, without any loss of generality one can express $\lambda_{min}=-\frac{\mathcal{C}_1}{2}$. Under local unitary operation $U\otimes V$, any two-qubit state satisfying Eq.$~$(\ref{l1}) can be transformed to an unique canonical state \cite{Badzia2000, Ghosal2020} with its Hilbert-Schmidt form,

\begin{align}
    \tilde{\rho}_1 &= (U \otimes V) ~\rho_1 ~(U^{\dagger} \otimes V^{\dagger}) \nonumber \\
    &= \dfrac{1}{4}\left[ \mathbb{I}\otimes \mathbb{I} + (\mathbf{r \cdot \sigma}) \otimes \mathbb{I} - \mathbb{I} \otimes (\mathbf{r\cdot \sigma})- \sum_{i=1}^3 |\tau_i| ~(\sigma_i \otimes \sigma_i)\right], \label{a3}
\end{align}
where $\mathbb{I}$ is identity operation on a qubit state and $\lbrace \sigma_i \rbrace$ are the Pauli matrices. The matrix representation in Eq.$~$(\ref{a3}) is the Hilbert-Schmidt representation \cite{Horodecki1995, Horodecki199654, Horodecki199621}, where $\textbf{r}=\lbrace r_1, ~r_2,~r_3\rbrace$ is the local vector of $\tilde{\rho}_1$ such that $r_i=Tr\left[ \tilde{\rho}_1~(\sigma_i \otimes \mathbb{I})\right]=Tr\left[ \tilde{\rho}_1~(\mathbb{I} \otimes \sigma_i)\right]$ for $i=1,2,3$. The parameter $\tau_i$ in Eq.$~$(\ref{a3}) is the correlation matrix element such that $\tau_i=\left[ \tilde{\rho}_1~(\sigma_i \otimes \sigma_i)\right]$ for $i=1,2,3$.

It is now trivial to show that $F(\rho_1)=F(\tilde{\rho}_1)= F^\star(\rho_1)$ holds because of the local unitary equivalence between $\tilde{\rho}_1$ and $\rho_1$. One can always show that Eq.$~$(\ref{l1}) holds for any two-qubit state {\it iff} 
\begin{align}
    \tilde{\rho}_1^{\Gamma} \ket{\Phi_0}= -\dfrac{\mathcal{C}_1}{2} \ket{\Phi_0} 
\end{align}
is satisfied (see {\bf Lemma 4} in Ref. \cite{Ghosal2020}), where $\sqrt{2}~\ket{\Phi_0} =\ket{00} +\ket{11}$. 

Now let us suppose the free segment $S_2$ shares a two-qubit state $\rho_2$ with concurrence value $\mathcal{C}_2$. From Ref. \cite{Gour2004} one can always argue that no three-party trace-preserving LOCC between $A-N_1-B$ can prepare a two-qubit state between Alice-Bob with concurrence greater than $\mathcal{C}_1\mathcal{C}_2$. Hence, from Ref. \cite{Gour2004} and \cite{Verstraete2002} one can conclude that the optimal FEF between Alice-Bob is upper bounded as 
\begin{align}
    F^\star_{AB}\leq \dfrac{1+\mathcal{C}_1 \mathcal{C}_2}{2} \leq \dfrac{1+\mathcal{C}_1}{2}. \label{l5}
\end{align}
Hence, from Eq.$~$(\ref{l5}) it is clear that $F^\star_{AB}< F^\star(\rho_1)$ if $\mathcal{C}_2<1$. Therefore, the equality in Eq.$~$(\ref{l5}) will hold {\it iff} we consider $\mathcal{C}_2=1$ which is a maximally entangled state shared in $S_2$. This completes the proof of {\bf Lemma 1}. 
\end{proof}

\section{Proof of Theorem 1} \label{AppendixB}
\begin{proof}   
       Let segment $S_1$ shares a rank two mixed state $\rho_{1} \in \mathcal{S}$, where $\mathcal{S}$ is a two-dimensional subspace in 
 two-qubit space spanned by a product state $\mathrm{\ket{P}}$ and its orthogonal entangled state $\ket{\zeta (\delta)}$. Without loss of generality we could take $\rho_1=\rho(p,\delta)= p \mathrm{P}+(1-p)\zeta (\delta) $, where $\zeta (\delta):=\ket{\zeta (\delta)}\bra{\zeta (\delta)}$ with $\ket{\zeta (\delta)}=\sqrt{\delta}\ket{00}+\sqrt{1-\delta}\ket{11}$, and $\mathrm{P}:=\mathrm{\ket{P}\bra{P}}=\ket{01}\bra{01}$.
   The optimal fully entangled fraction (OFEF) of the preshared state $\rho_{1} $ can be computed by solving a semi-definite program \cite{PhysRevA.78.032321, Somshubhro2012, Verstraete2003} as
   \begin{align}
        F^\star(\rho_{1})&=\dfrac{1}{2}(1+\mathcal{C}(\rho_{1})-p), \quad ~~ \textit{if} ~~\mathcal{C}(\rho_{1}) > 2p, \nonumber \\
       &=\dfrac{1}{2}\left( 1+\dfrac{C(\rho_1)^2}{4p}\right), \quad ~~ \textit{if} ~~ \mathcal{C}(\rho_{1})\leq2p, \label{a1} 
   \end{align}
   where $\mathcal{C}(\rho_{1})=2\sqrt{\delta~(1-\delta)}(1-p)$ is the concurrence of $\rho_{1}$. Furthermore from Eq.$~$(\ref{a1}) one can easily conclude that 
   \begin{align}
       F^\star(\rho_{1})\leq \dfrac{1+\mathcal{C}(\rho_{1})}{2}. 
   \end{align}
We now assume a scenario, where the free segment $S_2$ has a preshared pure two-qubit entangled state, say, $\ket{\psi_2}=\sqrt{\alpha}\ket{00}+\sqrt{1-\alpha}\ket{11}$, such that $\frac{1}{2}\leq \alpha <1$. Since we want to establish end-to-end OFEF $F_{AB}^\star = F^\star (\rho_1)$, prerequisite entanglement in the free segment $S_2$ has to maintain the following constraint
\begin{align}
    F^\star(\rho_{1})\leq  F^\star(\psi_2)=\dfrac{1+\mathcal{C}(\psi_2)}{2}. \nonumber
\end{align}
 The joint state $(A-N_1-B)$ shared by two segments $S_1$ and $S_2$ can thus be written as:
\begin{align}
\eta_{AN_1B}=\eta_{12}&=\rho_{1}\otimes\ket{\psi_2}\bra{\psi_2}\nonumber\\
    &= \Big(p\ket{01}\bra{01}+(1-p)\ket{\zeta (\delta)}\bra{\zeta (\delta)}\Big)\otimes\ket{\psi_2}\bra{\psi_2}.
\end{align} 
Now we will show that under some specific range of state parameters of $\rho_{1}$ and $\ket{\psi_2}$, there exist deterministic LOCC protocols for which one obtains 
\begin{align}
        F^\star_{AB}=F^\star(\rho_{1}).
   \end{align}  

\paragraph*{{\bf PVM is performed at node $N_1$}:}
Now let us assume that the protocol starts at node $N_1$ by performing a projective measurement (PVM) in the following basis: 
\begin{align}
    & \ket{\phi_0}= \sqrt{\beta} \ket{00}+ \sqrt{(1-\beta)} \ket{11}, \quad \ket{\phi_1}= \sqrt{(1-\beta)} \ket{00}-\sqrt{\beta} \ket{11}, \nonumber \\
    & \ket{\phi_2}= \sqrt{\beta} \ket{01}+ \sqrt{(1-\beta)} \ket{10}, \quad \ket{\phi_3}= \sqrt{(1-\beta)} \ket{01}-\sqrt{\beta} \ket{10}, \label{a4}
\end{align}
where $ \frac{1}{2}\leq \beta <1$. Post measurement state between Alice and Bob, corresponding to the projector $\{\ket{\phi_i}\bra{\phi_i}\}$, is defined as,
\begin{equation}
    \rho_{AB,\phi_i}=\dfrac{1}{p_i}~Tr_{N_1}\left((\ket{\phi_i}_{N_1}\bra{\phi_i}\otimes \mathbb{I}_{AB})\eta_{12} \right),
\end{equation}
where $p_i=Tr(\rho_{AB,\phi_i})$ is the probability of outcome $i$. The explicit expressions of  the post-measurement states for each of the measurement outcomes can thus be evaluated as
\begin{align}
    & \rho_{AB,\phi_0}= \dfrac{\big\{\alpha\beta\delta+(1-\alpha)(1-\beta)(1-\delta)\big\}(1-p)}{(1-\alpha)(1-\beta)-(1-\alpha-\beta)\delta(1-p)}\ket{\kappa_0}\bra{\kappa_0}+\dfrac{(1-\alpha)(1-\beta)p}{(1-\alpha)(1-\beta)-(1-\alpha-\beta)\delta(1-p)}\ket{01}\bra{01}, \label{a6} \\\nonumber\\ 
    & \rho_{AB,\phi_1}= \dfrac{\big\{\alpha(1-\beta)\delta+(1-\alpha)\beta(1-\delta)\big\}(1-p)}{(1-\alpha)\beta+(\alpha-\beta)\delta(1-p)}\ket{\kappa_1}\bra{\kappa_1}+\dfrac{(1-\alpha)\beta p}{(1-\alpha)\beta+(\alpha-\beta)\delta(1-p)}\ket{01}\bra{01},\label{a7} \\\nonumber\\ 
   & \rho_{AB,\phi_2}= \dfrac{\big\{(1-\alpha)\beta\delta+\alpha(1-\beta)(1-\delta)\big\}(1-p)}{\alpha(1-\beta)-(\alpha-\beta)\delta(1-p)}\ket{\kappa_2}\bra{\kappa_2}+\dfrac{\alpha(1-\beta)p}{\alpha(1-\beta)-(\alpha-\beta)\delta(1-p)}\ket{00}\bra{00},\label{a8}  \\\nonumber\\
    & \rho_{AB,\phi_3}= \dfrac{\big\{(1-\alpha)(1-\beta)\delta+\alpha\beta(1-\delta)\big\}(1-p)}{\alpha\beta+(1-\alpha-\beta)\delta(1-p)}\ket{\kappa_3}\bra{\kappa_3}+\dfrac{\alpha\beta p}{\alpha\beta+(1-\alpha-\beta)\delta(1-p)}\ket{00}\bra{00}\label{a9}, 
\end{align}
where $\ket{\kappa_i}$s are given by
\begin{align}
    &\ket{\kappa_0}=\dfrac{\sqrt{\alpha\beta\delta}\ket{00}+\sqrt{(1-\alpha)(1-\beta)(1-\delta)}\ket{11}}{\sqrt{\alpha\beta\delta+(1-\alpha)(1-\beta)(1-\delta)}},\\ ~\nonumber\\
    &\ket{\kappa_1}=\dfrac{\sqrt{\alpha(1-\beta)\delta}\ket{00}-\sqrt{(1-\alpha)\beta(1-\delta)}\ket{11}}{\sqrt{\alpha(1-\beta)\delta+(1-\alpha)\beta(1-\delta)}},\\~\nonumber\\
    &\ket{\kappa_2}=\dfrac{\sqrt{(1-\alpha)\beta\delta}\ket{01}+\sqrt{\alpha(1-\beta)(1-\delta)}\ket{10}}{\sqrt{(1-\alpha)\beta\delta+\alpha(1-\beta)(1-\delta)}},\\~\nonumber\\
    &\ket{\kappa_3}=\dfrac{\sqrt{(1-\alpha)(1-\beta)\delta}\ket{01}-\sqrt{\alpha\beta(1-\delta)}\ket{10}}{\sqrt{(1-\alpha)(1-\beta)\delta+\alpha\beta(1-\delta)}}.
\end{align}
Now from the knowledge of $\lbrace \ket{\kappa_i}\rbrace$, one can figure out OFEF of every $\rho_{AB,\phi_i}$ as 
\begin{align}
    F^{\star}(\rho_{AB,\phi_0})&=\frac{(1-\alpha) \alpha \delta p (1-p)}{2\big\{1-\alpha-\delta(1-p)\big\}\big\{(1-\alpha)(1-\beta)-(1-\alpha-\beta)\delta(1-p)\big\} }\nonumber\\
    &\qquad \qquad \qquad \qquad \qquad+\sqrt{\frac{(1-\alpha) \alpha (1-\beta) \beta (1-\delta) \delta (1-p)^2}{\big\{(1-\alpha)(1-\beta)-(1-\alpha-\beta)\delta(1-p)\big\}^2}}+\frac{
   (1-\alpha-\delta)(1-p)}{2 \big\{1-\alpha-\delta(1-p)\big\}},\nonumber\\
   & \qquad \qquad \qquad \qquad \qquad \qquad \qquad \qquad \qquad \qquad \qquad \qquad \qquad  \text{if,} \quad \alpha\beta\delta(1-\delta)(1-p)^2> (1-\alpha)(1-\beta)p^2,\nonumber \\~\nonumber\\
    F^{\star}(\rho_{AB,\phi_1})&=\frac{\big\{(1-\alpha)\beta+(\alpha-\beta)\delta\big\}(1-p)}{2\big\{(1-\alpha)\beta+(\alpha-\beta)\delta(1-p)\big\}}+\sqrt{\frac{(1-\alpha)\alpha(1-\beta) \beta (1-\delta) \delta (1-p)^2}{\big\{(1-\alpha)\beta+(\alpha-\beta)\delta(1-p)\big\}^2}}, \nonumber\\
    &\qquad \qquad \qquad \qquad \qquad \qquad \qquad \qquad \qquad \qquad \qquad \qquad \qquad \text{if,} \quad \alpha(1-\beta)\delta(1-\delta)(1-p)^2> (1-\alpha)\beta p^2, \nonumber \\
    &\text{and,} \\~\nonumber\\
    F^{\star}(\rho_{AB,\phi_0}) &=\frac{\big\{1-\delta(1-p)\big\}\big\{(1-\alpha)(1-\beta)p +\alpha\beta\delta(1-p)\big\}}{2p\big\{(1-\alpha)(1-\beta)-(1-\alpha-\beta)\delta(1-p)\big\}} , ~~~~~~~~~~~~~~~~ \text{if,} \quad \alpha\beta\delta(1-\delta)(1-p)^2\leq (1-\alpha)(1-\beta)p^2, \nonumber \\~\nonumber\\
    F^{\star}(\rho_{AB,\phi_1})&=\frac{\big\{1-\delta(1-p)\big\}\big\{(1-\alpha)\beta p +\alpha(1-\beta)\delta(1-p)\big\}}{2p\big\{(1-\alpha)\beta+(\alpha-\beta)\delta(1-p)\big\}}, ~~~~~~~~~~~~~~~~~~~~~~~\text{if,} \quad \alpha(1-\beta)\delta(1-\delta)(1-p)^2\leq (1-\alpha)\beta p^2,\label{a15}
    \end{align}
    whereas, 
    \begin{align}
     F^{\star}(\rho_{AB,\phi_2})&=\frac{1}{2}-\frac{(1-\beta)p}{2\big\{(1-\beta)-(1-2\beta)\delta(1-p)\big\}}+\sqrt{\frac{(1-\beta) \beta (1-\delta) \delta (1-p)^2}{\big\{(1-\beta)-(1-2\beta)\delta(1-p)\big\}^2}}, \nonumber\\
    &\qquad \qquad \qquad \qquad \qquad \qquad \qquad \qquad \qquad \qquad \qquad \qquad \qquad \text{if,} \quad (1-\alpha)\beta\delta(1-\delta)(1-p)^2> \alpha(1-\beta) p^2, \nonumber \\~\nonumber\\
    F^{\star}(\rho_{AB,\phi_3})&=\sqrt{\frac{(1-\alpha) \alpha (1-\beta) \beta (1-\delta) \delta (1-p)^2}{\big\{\alpha\beta+(1-\alpha-\beta)\delta(1-p)\big\}^2}}+\frac{\big\{\alpha\beta -(1-\alpha-\beta)\delta\big\}(1-p)}{2\big\{\alpha\beta+(1-\alpha-\beta)\delta(1-p)\big\}}, \nonumber\\
    &\qquad \qquad \qquad \qquad \qquad \qquad \qquad \qquad \qquad \qquad \qquad \qquad \qquad \text{if,} \quad (1-\alpha)(1-\beta)\delta(1-\delta)(1-p)^2> \alpha\beta p^2, \nonumber \\
    &\text{and,} \\~\nonumber\\
     F^{\star}(\rho_{AB,\phi_2}) &=\frac{\big\{1-\delta(1-p)\big\}\big\{\alpha(1-\beta)p+(1-\alpha)\beta\delta(1-p)\big\}}{2p\big\{\alpha(1-\beta)-(\alpha-\beta)\delta(1-p)\big\}}, ~~~~~~~~~~~~~~~~ \text{if,} \quad (1-\alpha)\beta\delta(1-\delta)(1-p)^2\leq \alpha(1-\beta) p^2, \nonumber \\~\nonumber\\
     F^{\star}(\rho_{AB,\phi_3})&=\frac{\big\{1-\delta(1-p)\big\}\big\{\alpha\beta p+(1-\alpha)(1-\beta)\delta(1-p)\big\}}{2p\big\{\alpha\beta+(1-\alpha-\beta)\delta(1-p)\big\}}, ~~~~~~~~~~~~~~~~ \text{if,} \quad (1-\alpha)(1-\beta)\delta(1-\delta)(1-p)^2\leq \alpha\beta p^2. \label{a17}
\end{align}

Now let us consider a particular domain where all four inequalities, 
\begin{align}
    & \alpha\beta\delta(1-\delta)(1-p)^2\leq (1-\alpha)(1-\beta)p^2,\label{a18} \\
    & \alpha(1-\beta)\delta(1-\delta)(1-p)^2\leq (1-\alpha)\beta p^2, \label{a19}\\
    & (1-\alpha)\beta\delta(1-\delta)(1-p)^2\leq \alpha(1-\beta) p^2,\label{a20} \\
    & (1-\alpha)(1-\beta)\delta(1-\delta)(1-p)^2\leq \alpha\beta p^2, \label{a21} 
\end{align}
are simultaneously satisfied. From these conditions, one can argue that the following inequalities should be satisfied, 
\begin{align}
   &0<p\leq\frac{1}{3},~~ \frac{1}{2}\Bigg(1+\sqrt{1-\frac{4p^2}{(1-p)^2}}\Bigg)<\delta<1,~~ \frac{1}{2}\leq \alpha<\frac{p^2}{p^2+\delta(1-\delta)(1-p)^2},~~ \frac{1}{2}\leq \beta\leq\frac{(1-\alpha)p^2}{(1-\alpha)p^2+\alpha \delta (1-\delta)(1-p)^2},\nonumber\\
   &\qquad \qquad \qquad \qquad \qquad \qquad \qquad \qquad \qquad \qquad \qquad \qquad  \text{or,}\nonumber\\
    &\qquad \qquad   \frac{1}{3}<p<1,~~~~ \frac{1}{2}\leq \delta<1,~~~~ \frac{1}{2}\leq \alpha<\frac{p^2}{p^2+\delta(1-\delta)(1-p)^2},~~~~ \frac{1}{2}\leq \beta\leq\frac{(1-\alpha)p^2}{(1-\alpha)p^2+\alpha \delta (1-\delta)(1-p)^2}.\label{a25}
\end{align}
It can be checked easily that this set of inequalities arises from Eqs.$~$(\ref{a18} -- \ref{a21}). Hence Eqs.$~$(\ref{a18} -- \ref{a21}) are the minimal constraints required for OFEF. So while all these conditions hold, from Eqs. (\ref{a15}) and (\ref{a17}) one can easily find out that the average OFEF, {\it i.e.,} 
\begin{align}
    F^\star_{AB}&= \sum_{i=0}^3 p_i ~F^\star(\rho_{AB,\phi_i}) \nonumber \\
&=\dfrac{1}{2}\left( 1+\delta(1-\delta)\dfrac{(1-p)^2}{p}\right) =  F^\star(\rho_{1})=F^\star(\rho(p,\delta)), \label{a26}
\end{align}
which ends the proof.

A major point to be noted here is that the established OFEF between Alice and Bob in Eq.$~$(\ref{a26}) is hitting the LOCC upper bound `$\min \lbrace F^\star(\rho_{1}), F^\star(\psi_{2})\rbrace$' which means Eq.$~$(\ref{a25}) must obey the inequality condition, $F^\star(\rho_{1})\leq F^\star(\psi_{2})$. This means that the following inequality must hold, {\it i.e.,} 
\begin{align}
    \dfrac{1}{2}\left( 1+\delta(1-\delta)\dfrac{(1-p)^2}{p}\right) \leq F^\star(\psi_2)=\dfrac{1+2\sqrt{\alpha~(1-\alpha)}}{2}, \label{a27}
\end{align}
if Eq.$~$(\ref{a25}) is satisfied. From Eq.$~$(\ref{a27}) one can find a domain of $\alpha$ as \begin{align}
   \frac{1}{2}\leq \alpha<\frac{p^2}{p^2+\delta(1-\delta)(1-p)^2} <\dfrac{1}{2}\left[ 1+ \sqrt{1-\delta^2(1-\delta)^2\dfrac{(1-p)^4}{p^2}}\right], \nonumber
\end{align}
which is consistent with Eq.$~$(\ref{a25}). 
\end{proof}

\section{Proof of Corollary 1}\label{AppendixC}
\begin{proof}
Let's consider a quantum repeater set-up, composed of $n$ nodes with $n+1$ segments between Alice and Bob. The segments $\{S_i\}^{n+1}_{i=2}$, each shares the state $\ket{\psi_i}=\sqrt{\alpha_i}\ket{00}+\sqrt{1-\alpha_i}\ket{11}$, having concurrence $\mathcal{C}(\psi_i)=2\sqrt{\alpha_i(1-\alpha_i)}$ and the segment $S_{1}$ shares a noisy state $\rho_{1}=\rho(p,\delta)= p \mathrm{P}+(1-p)\zeta (\delta)$, where $\zeta (\delta):=\ket{\zeta (\delta)}\bra{\zeta (\delta)}$ with $\ket{\zeta (\delta)}=\sqrt{\delta}\ket{00}+\sqrt{1-\delta}\ket{11}$ and $\mathrm{P}=\ket{01}\bra{01}$. The joint state shared by $n+1$ segments can thus be written as:
\begin{align}
    \eta_{AN_1\cdots N_{n}B}=\eta&=\rho_{1}\otimes\ket{\psi_2}\bra{\psi_2}\otimes\ket{\psi_3}\bra{\psi_3}\otimes.......\otimes\ket{\psi_{n+1}}\bra{\psi_{n+1}}\nonumber\\
    &=\Big(p \ket{01}\bra{01}+(1-p)\ket{\zeta (\delta)}\bra{\zeta (\delta)}\Big)\otimes \Big(\bigotimes_{i=2}^{n+1}\ket{\psi_i}\bra{\psi_i}\Big).
\end{align} 

It is known that given two pure non-maximally entangled states with concurrences, say $\mathcal{C}_1 \text{and} ~\mathcal{C}_2$, we can deterministically prepare another non-maximally entangled state with concurrence $\mathcal{C}_1\mathcal{C}_2$ through the RPBES protocol \cite{Gour2004}. The state prepared will have the form
\begin{equation}
     \ket{\tau} = \sqrt{\tau} \ket{00} + \sqrt{1-\tau} \ket{11},
\end{equation}
where $\tau=\dfrac{1+\sqrt{1-\mathcal{C}_1^2\mathcal{C}_2^2}}{2}$ is the Schmidt coefficients such that $\tau\geq \frac{1}{2}$ and the concurrence of $\ket{\tau}$ is $\mathcal{C}(\tau)=\mathcal{C}_1\mathcal{C}_2\leq 1$.\\ 

First, we apply this RPBES protocol sequentially at nodes $N_2, N_3,.......N_{n}$ for our given chain. The resultant state between node $N_1$ and Bob is
\begin{equation}
    \ket{\Psi_n}=\sqrt{\alpha'_n}\ket{00}+\sqrt{1-\alpha'_n}\ket{11},
\end{equation}
where $\alpha'_n=\dfrac{1+\sqrt{1-\mathcal{C}^{2}(\Psi_n)}}{2}$, and $\mathcal{C}(\Psi_n)=\prod_{i=2}^{n+1}\mathcal{C}(\psi_i)$. Now we are in a single node scenario where we have $\rho_{1}$ and $\ket{\psi}_{N_1B}\bra{\psi}=\ket{\Psi_n}\bra{\Psi_n}$. If we do only a Bell measurement at node $N_1$ ($\beta=\frac{1}{2}$) then from Appendix \ref{AppendixB} we know that to obtain optimal fidelity the noise parameters $p~\text{and}~\delta$ of the noisy state $\rho_{1}$, and the Schmidt coefficient of the state $\ket{\Psi_n}$ have to satisfy the following inequalities:
\begin{align}
   &0<p\leq\frac{1}{3},~~~~~ \frac{1}{2}\Bigg(1+\sqrt{1-\frac{4p^2}{(1-p)^2}}\Bigg)<\delta<1,~~~~~ \frac{1}{2}\leq \alpha'_n\leq\frac{p^2}{p^2+\delta(1-\delta)(1-p)^2},\nonumber\\
   &\qquad \qquad \qquad \qquad \qquad \qquad \qquad \qquad  \text{or,}\nonumber\\
    &\qquad \qquad   \frac{1}{3}<p<1,~~~~~ \frac{1}{2}\leq \delta<1,~~~~~ \frac{1}{2}\leq \alpha'_n\leq\frac{p^2}{p^2+\delta(1-\delta)(1-p)^2}.\label{b3}
\end{align}
Thus from the third inequality we have 
\begin{equation}\label{b4}
    \dfrac{1+\sqrt{1-\mathcal{C}^2(\Psi_n)}}{2}\leq\frac{p^2}{p^2+\delta(1-\delta)(1-p)^2}.
\end{equation}
Hence we see that for non-maximally entangled states with given concurrences distributed among $n$ segments $\{S_i\}^{n+1}_{i=2}$, there exists a certain range of the noise parameters $p~\text{and}~\delta$ of the noisy state $\rho_1$ in segment $S_1$ for which $F^\star_{AB}=F^\star(\rho_1)=F^\star(\rho(p,\delta)).$
\end{proof}
\section{Resource saved in the proposed LOCC protocol and comparison with E-SWAP}\label{AppendixD}

    If all the segments $\{S_i\}^{n+1}_{i=2}$ share the same non-maximally entangled state, \textit{i.e.,} $\alpha_i=\alpha$ for all $i\in[2,n+1]$ where $i$ is a natural number then $\mathcal{C}(\Psi_n)=\mathcal{C}^{n}$, where $\mathcal{C}=2\sqrt{\alpha(1-\alpha)}$. Substituting it in inequality (\ref{b4}) we have
\begin{equation}
    \dfrac{1+\sqrt{1-\mathcal{C}^{2n}}}{2}\leq\frac{p^2}{p^2+\delta(1-\delta)(1-p)^2}.
\end{equation}
With some basic manipulations, we arrive at the upper bound on the number of nodes $n$ for the given values of the noise parameters $\{p,\delta\}$ and the concurrence $\mathcal{C}$ of the states in segments $\{S_i\}^{n+1}_{i=2}$ as
\begin{equation}
    n\leq\dfrac{\Bigg|\log\bigg(1-\Big(\frac{p^2-\delta(1-\delta)(1-p)^2}{p^2+\delta(1-\delta)(1-p)^2}\Big)^2\bigg)\Bigg|}{2|\log(\mathcal{C})|}.
\end{equation}
Let us define a quantity saved resource: 
\begin{equation}\label{d3}
    R_v=n(1-\mathcal{C}).
\end{equation}
We know that one-{\it ebit} resource (a maximally entangled state with concurrence $\mathcal{C}=1$) is required per free segment in the E-SWAP protocol. Our proposed protocol, which utilizes non-maximally entangled state with concurrence $\mathcal{C}<1$, will suffice to satisfy the condition $F^\star_{AB}=F^\star(\rho(p,\delta))$. Thus, $R_{v}$ quantifies how much less resource is consumed in our proposed protocol compared to the E-SWAP protocol to achieve optimal fidelity in an $n$-node repeater scenario. In terms of the noise parameter and number of nodes, it is upper-bounded as
\begin{equation}
   R_v\leq n-n \left(1-\Bigg(\frac{p^2- \delta(1-\delta)(1-p)^2}{p^2+ \delta(1-\delta)
   (1-p)^2}\Bigg)^2\right)^{\frac{1}{2n}}.
\end{equation}

As the number of nodes increases the variance in resource consumption ($R_v$) saturates. In fact, for $n\rightarrow\infty$, we have
\begin{equation}
    R_v\rightarrow-\frac{1}{2} \log \left(\frac{\delta(1-\delta) (1-p)^2 p^2}{\left\{\delta(1-\delta)(1-p)^2 + p^2\right\}^2}\right)-\log (2).
\end{equation}

\section{Proof of Theorem 2}\label{AppendixE}
\begin{proof}
We consider a repeater scenario of $n+1$ number of segments ($n$ number of nodes) between Alice and Bob such that an intermediate segment, $S_m$, shares a noisy state of the form $\rho_{m}= \rho(p,\delta)=p \mathrm{P}+(1-p)\zeta (\delta)$, where $\zeta (\delta):=\ket{\zeta (\delta)}\bra{\zeta (\delta)}$ with $\ket{\zeta (\delta)}=\sqrt{\delta}\ket{00}+\sqrt{1-\delta}\ket{11}$, and $\mathrm{P}=\ket{01}\bra{01}$. All the other segments $\{S_i\}^{n+1}_{i=1}$, except for $i=m$, each shares a pure state $\ket{\psi_i}=\sqrt{\alpha_i}\ket{00}+\sqrt{1-\alpha_i}\ket{11}$ with concurrence $\mathcal{C}(\psi_i)$. So the initial state of the composite network can be expressed as 
\begin{align}
    \rho = \ket{\eta}_{1(m-1)}\bra{\eta} \otimes \rho_{m}\otimes \ket{\eta'}_{(m+1)(n+1)}\bra{\eta'}, \nonumber 
\end{align}
where we define
\begin{align}
    & \ket{\eta}_{1(m-1)}=  \bigotimes_{i=1}^{m-1}\ket{\psi_i}, \nonumber \\
    \text{and,}~~~&\ket{\eta'}_{(m+1)(n+1)}=  \bigotimes_{i=m+1}^{n+1} \ket{\psi_i}. \nonumber 
\end{align} 
The strategy is simple. We first apply the RPBES protocol \cite{Gour2004} to the first $m-1$ segments $\{S_i\}^{m-1}_{i=1}$. The resultant state will become
\begin{equation}
    \ket{\Psi_l}=\sqrt{\alpha'_l}\ket{00}+\sqrt{1-\alpha'_l}\ket{11},
\end{equation}
where $\alpha'_l=\dfrac{1+\sqrt{1-\mathcal{C}^{2}(\Psi_l)}}{2}$ and $\mathcal{C}(\Psi_l)=\prod_{i=1}^{m-1}\mathcal{C}(\psi_i)$. Similarly, we apply the RPBES protocol on the remaining segments $\{S_i\}^{n+1}_{i=m+1}$ resulting in the state 
\begin{equation}
    \ket{\Psi_{r}}=\sqrt{\alpha'_{r}}\ket{00}+\sqrt{1-\alpha'_{r}}\ket{11},
\end{equation}
where $\alpha'_r=\dfrac{1+\sqrt{1-\mathcal{C}^{2}(\Psi_r)}}{2}$ and $\mathcal{C}(\Psi_r)=\prod_{i=m+1}^{n+1}\mathcal{C}(\psi_i)$. 

So, we are now effectively reduced to a two-node scenario $\{N_{m-1}, N_m\}$ with three segments $S_1 (A-N_{m-1}), S_2 (N_{m-1}-N_m),~\text{and}~ S_3 (N_m-B)$, where the total state can be written as 
\begin{align}
    \rho = \ket{\Psi_l}\bra{\Psi_l} \otimes \rho_{2}\otimes \ket{\Psi_r}\bra{\Psi_r}, \nonumber 
\end{align}
where segment $S_1$ shares a non maximally entangled state $\ket{\Psi_l}$, segment $S_2$ shares the noisy state $\rho_2=\rho_m$, and segment $S_3$ shares another non maximally entangled state $\ket{\Psi_r}$. Let us first consider node $N_m$. We implement a measurement in the Bell basis there. Following the calculations in Appendix \ref{AppendixB}, we will end up with the four post-measurement states $\rho_{N_{m-1}B,\phi_i}$, where $i\in \{1,2,3,4\}$, as in Eqs. (\ref{a6}) to (\ref{a9}) with $\beta=\frac{1}{2}$.

Now, implement a Bell measurement at node $N_{m-1}$. Corresponding to each outcome, {\it i.e.,} $\rho_{N_{m-1}B,\phi_i}$,  of the earlier Bell measurement at node $N_m$, there will be four more outcomes. So, in total, there will be sixteen states. The detailed calculations are as follows:

For the state $\rho_{N_{m-1}B,\phi_0}$ (as in Eq.$~$(\ref{a6}) with $\beta=\frac{1}{2}$), the outcomes corresponding to the second Bell measurement are given by
\begin{align}
    \rho_{AB,\phi_0,\phi_0}=&\frac{\big\{\alpha'_r \alpha'_l \delta+(1-\alpha'_r)(1-\alpha'_l)(1-\delta)\big\}(1-p)}{(1 - \alpha'_r)(1 - \alpha'_l) - (1 - \alpha'_r) (1 - 2 \alpha'_l) p - (1 - \alpha'_r - \alpha'_l)\delta(1-p)}\ket{\zeta_{0,0}}\bra{\zeta_{0,0}}\nonumber\\
    &+\frac{(1-\alpha'_r)\alpha'_l p}{(1 - \alpha'_r)(1 - \alpha'_l) - (1 - \alpha'_r) (1 - 2 \alpha'_l) p - (1 - \alpha'_r - \alpha'_l)\delta(1-p)}\ket{01}\bra{01},\label{c3}\\~\nonumber\\
    \rho_{AB,\phi_0,\phi_1}=&\frac{\big\{\alpha'_r \alpha'_l \delta+(1-\alpha'_r)(1-\alpha'_l)(1-\delta)\big\}(1-p)}{(1 - \alpha'_r)(1 - \alpha'_l) - (1 - \alpha'_r) (1 - 2 \alpha'_l) p - (1 - \alpha'_r - \alpha'_l)\delta(1-p)}\ket{\zeta_{0,1}}\bra{\zeta_{0,1}}\nonumber\\
    &+\frac{(1-\alpha'_r)\alpha'_l p}{(1 - \alpha'_r)(1 - \alpha'_l) - (1 - \alpha'_r) (1 - 2 \alpha'_l) p - (1 - \alpha'_r - \alpha'_l)\delta(1-p)}\ket{01}\bra{01},\label{c4}\\~\nonumber\\
     \rho_{AB,\phi_0,\phi_2}=&\frac{\big\{\alpha'_r(1-\alpha'_l)\delta+(1-\alpha'_r)\alpha'_l(1-\delta)\big\}(1-p)}{(1-\alpha'_r)\alpha'_l+(1-\alpha'_r)(1-2\alpha'_l)p-(\alpha'_l-\alpha'_r)\delta(1-p)}\ket{\zeta_{0,2}}\bra{\zeta_{0,2}}\nonumber\\
    &+\frac{(1-\alpha'_r)(1-\alpha'_l)p}{(1-\alpha'_r)\alpha'_l+(1-\alpha'_r)(1-2\alpha'_l)p-(\alpha'_l-\alpha'_r)\delta(1-p)}\ket{11}\bra{11},\label{c5}\\~\nonumber\\
     \rho_{AB,\phi_0,\phi_3}=&\frac{\big\{\alpha'_r(1-\alpha'_l)\delta+(1-\alpha'_r)\alpha'_l(1-\delta)\big\}(1-p)}{(1-\alpha'_r)\alpha'_l+(1-\alpha'_r)(1-2\alpha'_l)p-(\alpha'_l-\alpha'_r)\delta(1-p)}\ket{\zeta_{0,3}}\bra{\zeta_{0,3}}\nonumber\\
    &+\frac{(1-\alpha'_r)(1-\alpha'_l)p}{(1-\alpha'_r)\alpha'_l+(1-\alpha'_r)(1-2\alpha'_l)p-(\alpha'_l-\alpha'_r)\delta(1-p)}\ket{11}\bra{11},\label{c6}
\end{align}
where,
\begin{align}
    \ket{\zeta_{0,0}}=\frac{\sqrt{\alpha'_r \alpha'_l \delta}\ket{00}+\sqrt{(1-\alpha'_r)(1-\alpha'_l)(1-\delta)}\ket{11}}{\sqrt{\alpha'_r \alpha'_l \delta+(1-\alpha'_r)(1-\alpha'_l)(1-\delta)}},\\~\nonumber\\
    \ket{\zeta_{0,1}}=\frac{\sqrt{\alpha'_r \alpha'_l \delta}\ket{00}-\sqrt{(1-\alpha'_r)(1-\alpha'_l)(1-\delta)}\ket{11}}{\sqrt{\alpha'_r \alpha'_l \delta+(1-\alpha'_r)(1-\alpha'_l)(1-\delta)}},\\~\nonumber\\
    \ket{\zeta_{0,2}}=\frac{\sqrt{(1-\alpha'_r)\alpha'_l(1-\delta)}\ket{01}+\sqrt{\alpha'_r(1-\alpha'_l)\delta}\ket{10}}{\sqrt{(1-\alpha'_r)\alpha'_l(1-\delta)+\alpha'_r(1-\alpha'_l)\delta}},\\~\nonumber\\
    \ket{\zeta_{0,3}}=\frac{\sqrt{(1-\alpha'_r)\alpha'_l(1-\delta)}\ket{01}-\sqrt{\alpha'_r(1-\alpha'_l)\delta}\ket{10}}{\sqrt{(1-\alpha'_r)\alpha'_l(1-\delta)+\alpha'_r(1-\alpha'_l)\delta}}.
\end{align}
 Note that the states in Eqs.$~$(\ref{c3} -- \ref{c6}) are of the same form as in Eqs.$~$(\ref{a6} -- \ref{a9}). Following the similar calculations as in Appendix \ref{AppendixB}, for finite resource consumption we get 
\begin{equation}\label{c11}
    \alpha'_r(1-\alpha'_l)\delta(1-\delta)(1-p)^2\leq(1-\alpha'_r)\alpha'_l p^2,
\end{equation}
and,
\begin{equation}\label{c12}
     \alpha'_r \alpha'_l \delta(1-\delta)(1-p)^2\leq(1-\alpha'_r)(1-\alpha'_l)p^2.
\end{equation}

For the state $\rho_{N_{m-1}B,\phi_1}$ (as in Eq.$~$(\ref{a7}) with $\beta=\frac{1}{2}$), the outcomes corresponding to the second Bell measurement are given by
\begin{align}
    \rho_{AB,\phi_1,\phi_0}=&\frac{\big\{\alpha'_r \alpha'_l \delta+(1-\alpha'_r)(1-\alpha'_l)(1-\delta)\big\}(1-p)}{(1 - \alpha'_r)(1 - \alpha'_l)- (1 - \alpha'_r) (1 - 2 \alpha'_l) p - (1 - \alpha'_r - \alpha'_l) \delta(1- p)}\ket{\zeta_{1,0}}\bra{\zeta_{1,0}}\nonumber\\
    &+\frac{(1-\alpha'_r)\alpha'_l p}{(1 - \alpha'_r)(1 - \alpha'_l)- (1 - \alpha'_r) (1 - 2 \alpha'_l) p - (1 - \alpha'_r - \alpha'_l) \delta(1- p)}\ket{01}\bra{01},\label{c13}\\~\nonumber\\
    \rho_{AB,\phi_1,\phi_1}=&\frac{\big\{\alpha'_r \alpha'_l \delta+(1-\alpha'_r)(1-\alpha'_l)(1-\delta)\big\}(1-p)}{(1 - \alpha'_r)(1 - \alpha'_l)- (1 - \alpha'_r) (1 - 2 \alpha'_l) p - (1 - \alpha'_r - \alpha'_l) \delta(1- p)}\ket{\zeta_{1,1}}\bra{\zeta_{1,1}}\nonumber\\
    &+\frac{(1-\alpha'_r)\alpha'_l p}{(1 - \alpha'_r)(1 - \alpha'_l)- (1 - \alpha'_r) (1 - 2 \alpha'_l) p - (1 - \alpha'_r - \alpha'_l) \delta(1- p)}\ket{01}\bra{01},\label{c14}\\~\nonumber\\
     \rho_{AB,\phi_1,\phi_2}=&\frac{\big\{\alpha'_r(1-\alpha'_l)\delta+(1-\alpha'_r)\alpha'_l(1-\delta)\big\}(1-p)}{(1-\alpha'_r)\alpha'_l+(1-\alpha'_r)(1-2\alpha'_l)p-(\alpha'_l-\alpha'_r)\delta(1-p)}\ket{\zeta_{1,2}}\bra{\zeta_{1,2}}\nonumber\\
    &+\frac{(1-\alpha'_r)(1-\alpha'_l)p}{(1-\alpha'_r)\alpha'_l+(1-\alpha'_r)(1-2\alpha'_l)p-(\alpha'_l-\alpha'_r)\delta(1-p)}\ket{11}\bra{11},\label{c15}\\~\nonumber\\
     \rho_{AB,\phi_1,\phi_3}=&\frac{\big\{\alpha'_r(1-\alpha'_l)\delta+(1-\alpha'_r)\alpha'_l(1-\delta)\big\}(1-p)}{(1-\alpha'_r)\alpha'_l+(1-\alpha'_r)(1-2\alpha'_l)p-(\alpha'_l-\alpha'_r)\delta(1-p)}\ket{\zeta_{1,3}}\bra{\zeta_{1,3}}\nonumber\\
    &+\frac{(1-\alpha'_r)(1-\alpha'_l)p}{(1-\alpha'_r)\alpha'_l+(1-\alpha'_r)(1-2\alpha'_l)p-(\alpha'_l-\alpha'_r)\delta(1-p)}\ket{11}\bra{11},\label{c16}
\end{align}
where,
\begin{align}
    \ket{\zeta_{1,0}}=\frac{\sqrt{\alpha'_r \alpha'_l \delta}\ket{00}-\sqrt{(1-\alpha'_r)(1-\alpha'_l)(1-\delta)}\ket{11}}{\sqrt{\alpha'_r \alpha'_l \delta+(1-\alpha'_r)(1-\alpha'_l)(1-\delta)}},\\~\nonumber\\
    \ket{\zeta_{1,1}}=\frac{\sqrt{\alpha'_r \alpha'_l \delta}\ket{00}+\sqrt{(1-\alpha'_r)(1-\alpha'_l)(1-\delta)}\ket{11}}{\sqrt{\alpha'_r \alpha'_l \delta+(1-\alpha'_r)(1-\alpha'_l)(1-\delta)}},\\~\nonumber\\
    \ket{\zeta_{1,2}}=\frac{\sqrt{(1-\alpha'_r)\alpha'_l(1-\delta)}\ket{01}-\sqrt{\alpha'_r(1-\alpha'_l)\delta}\ket{10}}{\sqrt{(1-\alpha'_r)\alpha'_l(1-\delta)+\alpha'_r(1-\alpha'_l)\delta}},\\~\nonumber\\
    \ket{\zeta_{1,3}}=\frac{\sqrt{(1-\alpha'_r)\alpha'_l(1-\delta)}\ket{01}+\sqrt{\alpha'_r(1-\alpha'_l)\delta}\ket{10}}{\sqrt{(1-\alpha'_r)\alpha'_l(1-\delta)+\alpha'_r(1-\alpha'_l)\delta}}.
\end{align}
We retrieve Eqs.$~$(\ref{c11}) and (\ref{c12}) for finite resource consumption.

 For the state $\rho_{N_{m-1}B,\phi_2}$ (as in Eq.$~$(\ref{a8}) with $\beta=\frac{1}{2}$), the outcomes corresponding to the second Bell measurement are given by
\begin{align}
    \rho_{AB,\phi_2,\phi_0}=&\frac{\big\{(1-\alpha'_r)\alpha'_l \delta+\alpha'_r(1-\alpha'_l)(1-\delta)\big\}(1-p)}{\alpha'_r(1-\alpha'_l)-\alpha'_r(1-2\alpha'_l)p-(\alpha'_r -\alpha'_l)\delta(1-p)}\ket{\zeta_{2,0}}\bra{\zeta_{2,0}}\nonumber\\
    &+\frac{\alpha'_r \alpha'_l p}{\alpha'_r(1-\alpha'_l)-\alpha'_r(1-2\alpha'_l)p-(\alpha'_r -\alpha'_l)\delta(1-p)}\ket{00}\bra{00},\label{c21}\\~\nonumber\\
    \rho_{AB,\phi_2,\phi_1}=&\frac{\big\{(1-\alpha'_r)\alpha'_l \delta+\alpha'_r(1-\alpha'_l)(1-\delta)\big\}(1-p)}{\alpha'_r(1-\alpha'_l)-\alpha'_r(1-2\alpha'_l)p-(\alpha'_r -\alpha'_l)\delta(1-p)}\ket{\zeta_{2,1}}\bra{\zeta_{2,1}}\nonumber\\
    &+\frac{\alpha'_r \alpha'_l p}{\alpha'_r(1-\alpha'_l)-\alpha'_r(1-2\alpha'_l)p-(\alpha'_r -\alpha'_l)\delta(1-p)}\ket{00}\bra{00},\label{c22}\\~\nonumber\\
     \rho_{AB,\phi_2,\phi_2}=&\frac{\big\{(1-\alpha'_r)(1-\alpha'_l)\delta+\alpha'_r \alpha'_l (1-\delta)\big\}(1-p)}{\alpha'_r\alpha'_l+\alpha'_r(1-2\alpha'_l)p+(1-\alpha'_r -\alpha'_l)\delta(1-p)}\ket{\zeta_{2,2}}\bra{\zeta_{2,2}}\nonumber\\
    &+\frac{\alpha'_r(1-\alpha'_l)p}{\alpha'_r\alpha'_l+\alpha'_r(1-2\alpha'_l)p+(1-\alpha'_r -\alpha'_l)\delta(1-p)}\ket{10}\bra{10},\label{c23}\\~\nonumber\\
     \rho_{AB,\phi_2,\phi_3}=&\frac{\big\{(1-\alpha'_r)(1-\alpha'_l)\delta+\alpha'_r \alpha'_l (1-\delta)\big\}(1-p)}{\alpha'_r\alpha'_l+\alpha'_r(1-2\alpha'_l)p+(1-\alpha'_r -\alpha'_l)\delta(1-p)}\ket{\zeta_{2,3}}\bra{\zeta_{2,3}}\nonumber\\
    &+\frac{\alpha'_r(1-\alpha'_l)p}{\alpha'_r\alpha'_l+\alpha'_r(1-2\alpha'_l)p+(1-\alpha'_r -\alpha'_l)\delta(1-p)}\ket{10}\bra{10},\label{c24}
\end{align}
where,
\begin{align}
    \ket{\zeta_{2,0}}=\frac{\sqrt{(1-\alpha'_r)\alpha'_l \delta}\ket{01}+\sqrt{\alpha'_r(1-\alpha'_l)(1-\delta)}\ket{10}}{\sqrt{(1-\alpha'_r)\alpha'_l \delta+\alpha'_r(1-\alpha'_l)(1-\delta)}},\\~\nonumber\\
    \ket{\zeta_{2,1}}=\frac{\sqrt{(1-\alpha'_r)\alpha'_l \delta}\ket{01}-\sqrt{\alpha'_r(1-\alpha'_l)(1-\delta)}\ket{10}}{\sqrt{(1-\alpha'_r)\alpha'_l \delta+\alpha'_r(1-\alpha'_l)(1-\delta)}},\\~\nonumber\\
    \ket{\zeta_{2,2}}=\frac{\sqrt{\alpha'_r \alpha'_l(1-\delta)}\ket{00}+\sqrt{(1-\alpha'_r)(1-\alpha'_l)\delta}\ket{11}}{\sqrt{\alpha'_r \alpha'_l (1-\delta)+(1-\alpha'_r)(1-\alpha'_l)\delta}},\\~\nonumber\\
    \ket{\zeta_{2,3}}=\frac{\sqrt{\alpha'_r \alpha'_l(1-\delta)}\ket{00}-\sqrt{(1-\alpha'_r)(1-\alpha'_l)\delta}\ket{11}}{\sqrt{\alpha'_r \alpha'_l (1-\delta)+(1-\alpha'_r)(1-\alpha'_l)\delta}}.
\end{align}
 Following the same steps as before, for finite resource consumption, we get 
\begin{equation}\label{c29}
     (1-\alpha'_r)(1-\alpha'_l)\delta(1-\delta)(1-p)^2\leq \alpha'_r \alpha'_l p^2,
\end{equation}
and,
\begin{equation}\label{c30}
     (1-\alpha'_r)\alpha'_l \delta(1-\delta)(1-p)^2\leq \alpha'_r(1-\alpha'_l) p^2.
\end{equation}

 Finally, for the state $\rho_{N_{m-1}B,\phi_3}$ (as in Eq.$~$(\ref{a9}) with $\beta=\frac{1}{2}$), the outcomes corresponding to the second Bell measurement are given by
\begin{align}
    \rho_{AB,\phi_3,\phi_0}=&\frac{\big\{(1-\alpha'_r)\alpha'_l \delta+\alpha'_r(1-\alpha'_l)(1-\delta)\big\}(1-p)}{\alpha'_r(1-\alpha'_l)-\alpha'_r(1-2\alpha'_l)p-(\alpha'_r -\alpha'_l)\delta(1-p)}\ket{\zeta_{3,0}}\bra{\zeta_{3,0}}\nonumber\\
    &+\frac{\alpha'_r \alpha'_l p}{\alpha'_r(1-\alpha'_l)-\alpha'_r(1-2\alpha'_l)p-(\alpha'_r -\alpha'_l)\delta(1-p)}\ket{00}\bra{00},\label{c31}\\~\nonumber\\
    \rho_{AB,\phi_3,\phi_1}=&\frac{\big\{(1-\alpha'_r)\alpha'_l \delta+\alpha'_r(1-\alpha'_l)(1-\delta)\big\}(1-p)}{\alpha'_r(1-\alpha'_l)-\alpha'_r(1-2\alpha'_l)p-(\alpha'_r -\alpha'_l)\delta(1-p)}\ket{\zeta_{3,1}}\bra{\zeta_{3,1}}\nonumber\\
    &+\frac{\alpha'_r \alpha'_l p}{\alpha'_r(1-\alpha'_l)-\alpha'_r(1-2\alpha'_l)p-(\alpha'_r -\alpha'_l)\delta(1-p)}\ket{00}\bra{00},\label{c32}\\~\nonumber\\
     \rho_{AB,\phi_3,\phi_2}=&\frac{\big\{(1-\alpha'_r)(1-\alpha'_l)\delta+\alpha'_r \alpha'_l (1-\delta)\big\}(1-p)}{\alpha'_r\alpha'_l+\alpha'_r(1-2\alpha'_l)p+(1-\alpha'_r -\alpha'_l)\delta(1-p)}\ket{\zeta_{3,2}}\bra{\zeta_{3,2}}\nonumber\\
    &+\frac{\alpha'_r(1-\alpha'_l)p}{\alpha'_r\alpha'_l+\alpha'_r(1-2\alpha'_l)p+(1-\alpha'_r -\alpha'_l)\delta(1-p)}\ket{10}\bra{10},\label{c33}\\~\nonumber\\
     \rho_{AB,\phi_3,\phi_3}=&\frac{\big\{(1-\alpha'_r)(1-\alpha'_l)\delta+\alpha'_r \alpha'_l (1-\delta)\big\}(1-p)}{\alpha'_r\alpha'_l+\alpha'_r(1-2\alpha'_l)p+(1-\alpha'_r -\alpha'_l)\delta(1-p)}\ket{\zeta_{3,3}}\bra{\zeta_{3,3}}\nonumber\\
    &+\frac{\alpha'_r(1-\alpha'_l)p}{\alpha'_r\alpha'_l+\alpha'_r(1-2\alpha'_l)p+(1-\alpha'_r -\alpha'_l)\delta(1-p)}\ket{10}\bra{10},\label{c34}
\end{align}
where,
\begin{align}
    \ket{\zeta_{3,0}}=\frac{\sqrt{(1-\alpha'_r)\alpha'_l \delta}\ket{01}-\sqrt{\alpha'_r(1-\alpha'_l)(1-\delta)}\ket{10}}{\sqrt{(1-\alpha'_r)\alpha'_l \delta+\alpha'_r(1-\alpha'_l)(1-\delta)}},\\~\nonumber\\
    \ket{\zeta_{3,1}}=\frac{\sqrt{(1-\alpha'_r)\alpha'_l \delta}\ket{01}+\sqrt{\alpha'_r(1-\alpha'_l)(1-\delta)}\ket{10}}{\sqrt{(1-\alpha'_r)\alpha'_l \delta+\alpha'_r(1-\alpha'_l)(1-\delta)}},\\~\nonumber\\
    \ket{\zeta_{3,2}}=\frac{\sqrt{\alpha'_r \alpha'_l(1-\delta)}\ket{00}-\sqrt{(1-\alpha'_r)(1-\alpha'_l)\delta}\ket{11}}{\sqrt{\alpha'_r \alpha'_l (1-\delta)+(1-\alpha'_r)(1-\alpha'_l)\delta}},\\~\nonumber\\
    \ket{\zeta_{3,3}}=\frac{\sqrt{\alpha'_r \alpha'_l(1-\delta)}\ket{00}+\sqrt{(1-\alpha'_r)(1-\alpha'_l)\delta}\ket{11}}{\sqrt{\alpha'_r \alpha'_l (1-\delta)+(1-\alpha'_r)(1-\alpha'_l)\delta}}.
\end{align}
Again, for finite resource consumption, we retrieve Eqs.$~$(\ref{c29}) and (\ref{c30}).

These four ({\it i.e.,} Eqs.$~$(\ref{c11} -- \ref{c12}) and Eqs.$~$(\ref{c29} -- \ref{c30})) are exactly the constraints obtained in Eqs.$~$(\ref{a18} -- \ref{a21}). So from Appendix \ref{AppendixB}, it follows that the minimal condition that has to be satisfied for finite resource consumption is 
\begin{align}
    \alpha'_r\alpha'_l\delta(1-\delta)(1-p)^2&\leq (1-\alpha'_l)(1-\alpha'_r)p^2\nonumber\\
    \Rightarrow~~~~~~~\frac{\alpha'_l\alpha'_r}{(1-\alpha'_l)(1-\alpha'_r)}&\leq \frac{p^2}{\delta(1-\delta)(1-p)^2},\label{c39}
\end{align}
which concludes the proof.
\end{proof}

\section{Distribution of Noise in Noisy Segments and Resource Consumption}  

In an \(n\)-node quantum repeater scenario, the position of a noisy segment exhibiting the narrowest (or broadest) range of \(\{p, \delta\}\), as determined by Theorem 2 for a given noisy state, explicitly depends on the distribution of \(\{\alpha_i\}\) in the free segments. The following observation outlines how the distribution of noise in the noisy segments impacts resource consumption.

\begin{observation} \label{obs3}
In an $n$-node scenario, assuming free segments share identical non-maximal pure entangled states the variance in resource consumption $R_v$ experiences a non-decreasing trend as the noisy segment migrates from the midpoint towards either end of the repeater scenario.
\end{observation}

\begin{proof}
For better visualisation we rewrite  Eq.$~$(\ref{c39}) in terms of concurrence as
\begin{equation}\label{c40}
   \frac{\Big(1+\sqrt{1-\mathcal{C}^{2}(\Psi_l)}\Big)\Big(1+\sqrt{1-\mathcal{C}^{2}(\Psi_r)}\Big)}{\Big(1-\sqrt{1-\mathcal{C}^{2}(\Psi_l)}\Big)\Big(1-\sqrt{1-\mathcal{C}^{2}(\Psi_r)}\Big)}\leq\frac{p^{2}}{\delta(1-\delta)(1-p)^2}.
\end{equation}

If all the consecutive $\{S_i\}^{n+1}_{i=1}$, except for $i=m$, share the same pure state $\ket{\psi_i}=\ket{\psi}=\sqrt{\alpha}\ket{00}+\sqrt{1-\alpha}\ket{11}$, \textit{i.e.,} $\alpha_i=\alpha$ for $i\in[1,n+1]\setminus \{m\}$, where $i$ is a natural number, with concurrence $\mathcal{C}=2\sqrt{\alpha(1-\alpha)}$ then Eq.$~$(\ref{c40}) can be written as
\begin{equation}\label{f2}
   \frac{\Big(1+\sqrt{1-\mathcal{C}^{2(m-1)}}\Big)\Big(1+\sqrt{1-\mathcal{C}^{2(n-m+1)}}\Big)}{\Big(1-\sqrt{1-\mathcal{C}^{2(m-1)}}\Big)\Big(1-\sqrt{1-\mathcal{C}^{2(n-m+1)}}\Big)}\leq\frac{p^{2}}{\delta(1-\delta)(1-p)^2}.
\end{equation}

\end{proof}

\begin{figure}[h!]
     \centering
         \centering
         \includegraphics[height=150px, width=200px]{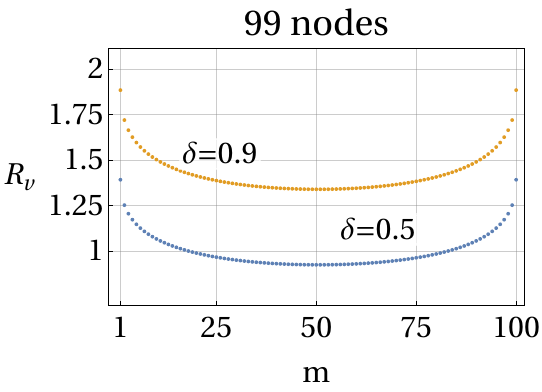}
         \caption{The saved resource \(R_v\) as a function of the noisy node position (\(m\)) for \(n = 99\) and the noisy state \(\rho(p, \delta)\) with parameters \(p\) and \(\delta\). When all free segments share identical states (i.e., \(\alpha_i = \alpha\)), the minimum value of \(R_v\) occurs at \(m = n/2\) for even \(n\), and at \(m \in \{(n-1)/2, (n+1)/2\}\) for odd \(n\). Moreover, as the noisy segment moves closer to either end of the chain, \(R_v\) consistently increases.}
         \label{Fig:3}
\end{figure}
\section{Two Noisy segments}\label{G}
  
       Let us consider a single node scenario $(A-N_1-B)$ with two segments $S_1$ and $S_2$. Segment $S_1$ shares a rank two mixed state $\rho_1=\rho(p,\delta)= p \mathrm{P}+(1-p)\zeta (\delta) $, where $\zeta (\delta):=\ket{\zeta (\delta)}\bra{\zeta (\delta)}$ with $\ket{\zeta (\delta)}=\sqrt{\delta}\ket{00}+\sqrt{1-\delta}\ket{11}$, and $\mathrm{P}=\ket{01}\bra{01}$. Let us assume that segment $S_2$ also has the same state shared, \textit{i.e.,} $\rho_2=\rho(p,\delta)=\rho_1$. Using Eq.$~$(4) of the main text we have
       \begin{equation}
           F^{\star}_{AB} \leq \min \{F^{\star}(\rho(p,\delta)), F_C\},\label{g1}
       \end{equation}
       where, $F_C=\frac{1+C^2(\rho(p,\delta))}{2}$.

       If $F_C<F^{\star}(\rho(p,\delta))$, then according to Eq.$~$(\ref{g1}), it is impossible to maintain the same fidelity between $A$ and $B$ as achieved in segment $S_1$ between $A$ and $N_1$. Let us now consider Bell measurement ($\beta=\frac{1}{2}$) at node $N_1$. Eq.$~$(\ref{g1}) determines the maximum fidelity of the final state between Alice and Bob. Using Eq.$~$(\ref{a1}) with $\delta=0.95$, we can determine the range of the noise parameter where it is impossible to maintain the same fidelity between Alice and Bob as in segment $S_1$ (or segment $S_2$) and it can be calculated from the previously mentioned inequality as
       \begin{align}
           0\leq p <0.25.\label{g2}
       \end{align}

       In the single node scenario where $\rho_1=\rho(p,\delta)$, but segment $S_2$ has a pure two-qubit entangled state $\ket{\psi_2}=\sqrt{\alpha}\ket{00}+\sqrt{1-\alpha}\ket{11}$, using Eq.$~$(\ref{a18}) with $\delta=0.95~\text{and}~\alpha=0.55$, we find that for 
       \begin{align}
           0.19\leq p <1,\label{g3}
       \end{align}
       we can establish $F^{\star}_{AB}=F^{\star}(\rho_1)$.

       Thus, from Eqs.$~$(\ref{g2}) and (\ref{g3}), we can see that for
       \begin{align}
           0.19\leq p <0.25,\label{g4}
       \end{align}
       it is impossible to establish the same fidelity between Alice and Bob as in the initial state if both nodes are noisy. This is in contrast to the case where one of the nodes is not noisy, allowing us to maintain the same fidelity.

\section{Generation of states in noisy segments}
As established by Lemma 1 of the main text, to achieve the same teleportation fidelity between the sender and receiver (using non-maximal resources) as that of the state in the noisy segment, this state cannot have any arbitrary noise. Identifying states that enable the same end-to-end teleportation fidelity in a quantum repeater scenario is generally challenging. However, in Theorem 1, we present an explicit protocol for a particular class of rank-2 noisy states that makes this feasible. This noise model is not merely a theoretical construct—it is well-grounded in experimental settings.  These kinds of states are prepared when one qubit of the non-maximally entangled states of the form $\sqrt{\omega}\ket{01}+\sqrt{1-\omega}\ket{10}$ passes through an Amplitude-Damping Channel (ADC) \cite{Nielsen2010} with Kraus operators $K_0= \ket{0}\bra{0}+\sqrt{1-p}\ket{1}\bra{1}$ and $K_1= \sqrt{p} \ket{0}\bra{1}$ with the channel parameter $p\in (0,1)$. The resultant rank-2 state is given as
\begin{equation}
\rho(p,\omega)= (1-p\omega) \ket{\Omega}\bra{\Omega}+ p(1-\omega) \ket{00}\bra{00},\label{h1}
\end{equation}
where 
\[
\ket{\Omega}= \left(  \frac{\sqrt{\omega}}{\sqrt{ (1-p\omega)}}\ket{01}+ \frac{\sqrt{(1-\omega)(1-p)}}{\sqrt{ (1-p\omega)}}\ket{10}\right).
\]
These are of the same form as the states in the noisy segment considered in Theorem 1. These results are due to the decay of an electronically excited state of atoms if qubits are built out of atomic levels, or loss of a photon whenever qubits are encoded in the photons \cite{riera2021entanglement}. Thus, while studying the noise robustness of our study later,  whenever we use this kind of noise model, we will refer to it as photon loss noise.
\section{Noise Robustness in Quantum Repeater Scenarios}\label{AppendixH}
The investigation so far considered the use of pure non-maximally entangled states in the {\bf free segments}. Here, we extend the analysis to explore how the average fidelity changes when these resources are replaced with noisy, mixed entangled states. The updated setup consists of two segments: a noisy segment characterized by a noisy state $\rho(p, \delta)$ and a free segment supplied with a noisy state $\sigma$. Together, these states facilitate the establishment of a quantum communication line between the end-to-end parties. To assess the impact of noise, we perform a robust analysis by considering two scenarios for $\sigma$: $(i)$ an asymptotic number of copies to determine its extreme potential, and $(ii)$ a single copy to address practical relevance. This study offers insights into both theoretical limits and realistic implementations. Importantly, the noisy segment is consistently analyzed using a single-copy scenario, while cases $(i)$ and $(ii)$ are applied exclusively to the resource in the free segment, ensuring a focused and meaningful exploration of its role in the communication setup. To recall the scenario: we consider a single-node setup where a node acts as a repeater between Alice and Bob. Without loss of generality, let the first segment between Alice and the node be a free segment characterized by the noisy state \(\sigma\), while the second segment, shared between the node and Bob, contains a single copy of a fixed noisy state \(\rho\).
\subsection{Asymptotic Analysis: Higher Noise in $\rho$ $\rightarrow$ Fewer Copies of $\chi$ are Sufficient}
In the free segment, multiple copies of the entangled resource $\chi$ can be used to improve the communication channel. A traditional approach involves distilling $\chi$ into maximally entangled states $\ket{\phi^+}$ to enable entanglement swapping and establish the same entangled state $\rho(p, \delta)$ between the end parties. However, achieving optimal teleportation fidelity may be more efficient by distilling a specific non-maximally entangled state $\ket{\psi}$ instead of a maximally entangled state. This distillation process can be expressed as: 
\begin{equation} \left(\chi\right)^{\otimes j} \rightarrow \left(\ket{\phi^+}\bra{\phi^+}\right)^{\otimes k} \leftrightarrow \left(\ket{\psi}\bra{\psi}\right)^{\otimes k^{\prime}}.
\end{equation}

Here, the first transformation, which distills maximally entangled states from $\chi$, is not necessarily reversible for general states $\chi$, even in the asymptotic limit \cite{Horodecki2009}. On the other hand, the second transformation, converting maximally entangled states to non-maximally entangled states, is reversible in the asymptotic limit \cite{BennettPure1996}. The rate at which maximally entangled states are distilled from $\chi$ is given by: \begin{equation} R_{\phi^+}(\chi) = \lim_{j \rightarrow \infty} \frac{k}{j}. \end{equation} To determine the distillation rate for non-maximally entangled states, we define the following quantity: \begin{equation} R_{\psi}(\chi) = \lim_{j \rightarrow \infty} \frac{k^{\prime}}{j} = \lim_{j \rightarrow \infty} \left( \frac{k}{j} \times \frac{k^{\prime}}{k} \right) = \frac{R_{\phi^+}(\chi)}{S(\rho_{\psi})}, \end{equation} where $\rho_{\psi} := \Tr_B\left(\ket{\psi}_{AB}\bra{\psi}\right)$ is the reduced state of one subsystem. For a non-maximally entangled state $\ket{\psi}$, the von Neumann entropy $S(\rho_{\psi}) < 1$, which implies that $R_{\psi}(\chi) > R_{\phi^+}(\chi)$. Therefore, in the asymptotic limit, preparing one copy of $\ket{\psi}$ requires $\frac{1}{R_{\psi}(\chi)}$ copies of $\chi$, which is fewer than the $\frac{1}{R_{\phi^+}(\chi)}$ copies needed to prepare a maximally entangled state. In general, finding the optimal distillation rate ($R_{\phi^+}$) for a mixed state is challenging due to the complexity of determining the optimal LOCC protocol over its asymptotic copies. However, for pure states, the protocol is known \cite{BennettPure1996} and the von Neumann entropy of the marginal state quantifies the rate. This simplifies our analysis, as for any given distillation rate $D_{\phi^+}(\chi)$ with some LOCC protocol $\Lambda^D_{LOCC}$, the distillation rate of a non-maximally entangled state is given by: 
\begin{equation}
D_{\psi}(\chi) = \frac{D_{\phi^+}(\chi)}{S(\rho_{\psi})}. \end{equation}
For our analysis, we consider $\chi=F\ketbra{\Psi^-}+\frac{(1-F)}{3}\Big(\ketbra{\Psi^+}+\ketbra{\Phi^+}+\ketbra{\Phi^-}\Big)$ as a Werner state and use the value of distillation $ D_{\Phi^+}(\chi)=1+F \log_2 (F)+(1-F) \log_2 \left(\frac{1-F}{3}\right)$ using the hashing protocol \cite{Bennett1996}. Suppose we follow the traditional protocol to establish optimal teleportation by distilling maximally entangled states in the free segment. In that case, the number of copies of $\chi$ required in the process is given by $\frac{1}{D_{\phi^+}(\chi)}$, which does not depend on the noisy state $\rho(p,\delta)=p\ket{01}\bra{01}+(1-p)\ket{\zeta (\delta)}\bra{\zeta (\delta)}$ in the noisy segment. However, our protocol suggests that a non-maximally entangled state $\ket{\psi}=\sqrt{\alpha}\ket{00}+\sqrt{1-\alpha}\ket{11}$ can sufficiently establish the same average teleportation fidelity with the noisy state $\rho(p,\delta)$. The upper bound on this value of $\alpha$ can be achieved from Appendix \ref{AppendixB}. For $\beta=\delta=\frac{1}{2}$ its exact expression can be calculated from Eq. (\ref{a25}) as
 \begin{equation}
            \alpha=\frac{4 p^2}{5 p^2-2 p+1}.
        \end{equation}
        Also von Neumann entropy of $\ket{\psi}$ is given by
        \begin{align}
            S(\rho_{\psi})=-\alpha \log_2 (\alpha)-(1-\alpha) \log_2 (1-\alpha).\label{I6}
        \end{align}
        Using these, the number of copies of $\chi$ required can be given as
        \begin{align}
         \frac{1}{D_{\psi}(\chi)}=  \frac{S(\rho_{\psi})}{D_{\phi^+}(\chi)}=\frac{4 p^2 \log_2 \left(\frac{4 p^2}{5 p^2-2 p+1}\right)+(p-1)^2 \log_2
   \left(\frac{(p-1)^2}{5 p^2-2 p+1}\right)}{(5 p^2-2 p+1) \left(-F \log_2 (1-F)+\log_2 \left(-\frac{2}{3} (F-1)\right)+F \log_2 (3 F)\right)}.
        \end{align}
        
        In Figure~\ref{f11}, we explore how this number decreases as the noise in the noisy segment increases for a fixed value of $F=0.8161$.

\begin{figure}[h!]
    \centering
       \includegraphics[height=150px, width=220px]{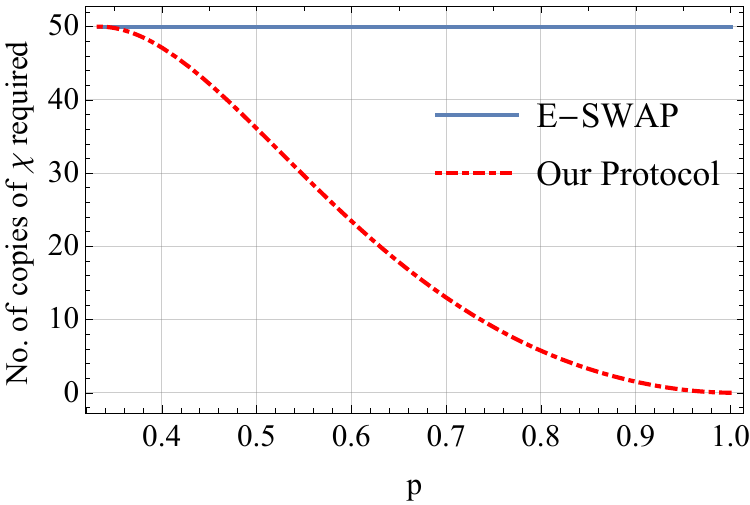}
         \caption{{\bf Noise vs. Copies :} This figure illustrates how the number of Werner states required as resources in the free segment decreases as the noise in the noisy segment increases. Unlike in traditional entanglement swapping protocols, where the number of required copies remains independent of noise, our protocol shows that higher noise in the noisy segment allows for a reduction in the number of copies needed in the free segment to achieve the same level of teleportation fidelity.}
        \label{f11}
\end{figure}

We can intuitively understand how noise in the free segments reduces the number of required copies. It is noteworthy that the conventional E-SWAP protocol does not have this flexibility and remains independent of noise. This observation motivates us to investigate how effectively our protocol conserves resources, as defined in Eq.~(\ref{d3}). Understanding this will provide further insight into the advantages of our protocol in quantum repeater scenarios.

Furthermore, suppose we have $n$ free segments in a linear network. As per discussed in Appendix \ref{AppendixD}, all free segments share the same non-maximally entangled state with a Schmidt coefficient $\alpha$. If we aim to distill this state in each free segment from an entangled resource $\chi$, which in this section is taken to be a Werner state, the number of copies of $\chi$ required in the $n$ node setup is given by
\begin{equation}
N=\frac{n}{D_{\psi}(\chi)} = \frac{nS(\rho_{\psi})}{D_{\phi^+}(\chi)}.\label{I8} \end{equation}
Furthermore, the Schmidt coefficient $\alpha$ of a non-maximally entangled state is related to its concurrence $\mathcal{C}(\alpha)$ by
$\alpha = \frac{1+\sqrt{1-\mathcal{C}^2}}{2}.$ Substituting this in Eq. (\ref{I6}) and then using Eq. (\ref{d3}) we get
\begin{align}
            S(\rho_{\psi})=-\Bigg(\frac{1+\sqrt{1-(1-\frac{R_v}{n})^2}}{2}\Bigg) \log_2 \Bigg(\frac{1+\sqrt{1-(1-\frac{R_v}{n})^2}}{2}\Bigg)-\Bigg(\frac{1-\sqrt{1-(1-\frac{R_v}{n})^2}}{2}\Bigg) \log_2 \Bigg(\frac{1-\sqrt{1-(1-\frac{R_v}{n})^2}}{2}\Bigg).\label{I9}
        \end{align}
 We use this in Eq. (\ref{I8}) and get

\begin{align}
  N=  \frac{n}{D_{\psi}(\chi)}= \frac{nS(\rho_{\psi})}{D_{\phi^+}(\chi)}=\frac{n\Bigg(-\Bigg(\frac{1+\sqrt{1-(1-\frac{R_v}{n})^2}}{2}\Bigg) \log_2 \Bigg(\frac{1+\sqrt{1-(1-\frac{R_v}{n})^2}}{2}\Bigg)-\Bigg(\frac{1-\sqrt{1-(1-\frac{R_v}{n})^2}}{2}\Bigg) \log_2 \Bigg(\frac{1-\sqrt{1-(1-\frac{R_v}{n})^2}}{2}\Bigg)\Bigg)}{\Big(1+F \log_2 (F)+(1-F) \log_2 \left(\frac{1-F}{3}\right)\Big)}
\end{align}
For a fixed fidelity value $F=0.8161$, we plot Figure \ref{f111}. The plot clearly illustrates the inverse proportionality between the number of required copies and the saved resource in our protocol. As we understand, the number of required copies per segment decreases due to noise, leading to a reduction in the total resource required compared to the ESWAP protocol. Consequently, an amount of resource equal to $(1 - C(\alpha))$ is saved at each free segment, and the total saved resource over $n$ free segments is given by  
\[
R_v = n(1 - C(\alpha)).
\]  
Figure \ref{f111} provides an intuitive visualization of the significance of $R_v$.

\begin{figure}[h!]
    \centering
       \includegraphics[height=150px, width=220px]{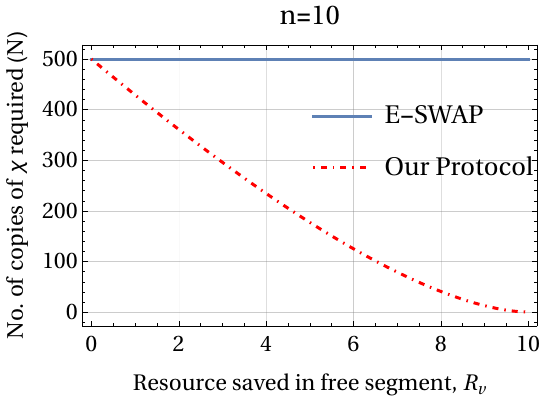}
         \caption{{\bf Resource Saved vs. Number of Copies :} This figure illustrates how the number of Werner states required as resources in the free segment decreases due to noise. Consequently, the amount of resource saved, \( R_v \), in terms of entanglement, increases in the free segment. The total number of segments considered here is 10. This plot provides a more clear understanding of the nature of the resource savings achieved.}
        \label{f111}
\end{figure}
\subsection{Single Shot Analysis}
Motivated by practical relevance, we perform a single-shot analysis considering specific noisy models for the free segment, such as white noise and photon loss noise. Additionally, we account for noise in the measurements at the nodes and examine how the fidelity changes when each type of noise is included separately, as well as when they are combined in the setup.
\subsection*{\textbf{Case 1:} White noise in the free segment}

We have demonstrated that non-maximally entangled states ($\ket{\psi}$) can effectively achieve optimal teleportation fidelity. Next, we will showcase the robustness of our findings in the presence of noise. First, we consider the scenario in which white noise is added to the non-maximally entangled state in the free segment. As we are looking at the change in fidelity against the noise injected in the free segment we choose the non-maximally entangled state to be
\begin{align}
            \chi_1=(1-q)\ketbra{\Phi'}+q\frac{\mathbb{I}}{4},
        \end{align}
        where $\ket{\Phi'}=\sqrt{\frac{3}{4}}\ket{00}+\sqrt{\frac{1}{4}}\ket{11}$.
This can be thought of as a scenario where with probability $(1-q)$ we have an entangled state, while with probability $q$ white noise is added. For $q=0$, this also boils down to the scenario in Appendix \ref{AppendixB}.

The state in the noisy segment is chosen to be
\begin{align}
            \rho(p)=(1-p)\ketbra{\Phi^{+}}+p\ketbra{01},
        \end{align}
  
         where $\ket{\Phi^{+}}=\frac{\ket{00}+\ket{11}}{\sqrt{2}}$.

 From Appendix \ref{AppendixB} this scenario corresponds to $\delta=\frac{1}{2}$ and $\alpha=\frac{3}{4}$. In addition, we only perform Bell measurement at the node between Alice and Bob; thus, we can choose $\beta=\frac{1}{2}$. Using these in Eqs. (\ref{a25}) and (\ref{a26}) the value of the optimal fidelity is
         \begin{align}
             F^{*}_{AB}=\frac{(p+1)^2}{8 p},\label{i4}
         \end{align}
         with the noise parameter p of the state in the noisy segment constrained to the range $\frac{\sqrt{3}}{2+\sqrt{3}}\leq p<1$.

In the case of a noisy state in the free segment, the protocol discussed earlier might not give us the optimal fidelity. However, if you restrict the protocol we followed for $q=0$ the optimal fidelity for the noisy free segment will be lower bounded by the following quantity
   
        \begin{align}
            F^{noisy}_{AB}=\frac{2 (p-1)^2 (2 p-1) (1-q)^3+2 (p-1) (p (p+5)-2) (1-q)^2+(p+1) (p (5 p+16)-5) (1-q)+(p+1)^2 (p+3)}{8 ((p-1) (1-q)+p+1) (5 p (1-q)+p+q)}
        \end{align}
  The range of $p$ and $q$ for which this expression is valid is given by
        \begin{align}
             \frac{1}{2}\leq (1-q)\leq 1~\&\&~\frac{\sqrt{3} (1-q)-q}{\sqrt{3} (1-q)-q+2}<p<\frac{-2 (1-q)^2-2 (1-q)+1}{2 (1-q)^2-4 (1-q)-1}
        \end{align}
       Both of these reduce to Eq. (\ref{i4}) when $q=0$. 
       
       We are interested in the percentage change in fidelity which can be mathematically written as
       \begin{align}
            F\%=\frac{F^{*}_{AB}-F^{noisy}_{AB}}{F^{*}_{AB}}\times100.
        \end{align}

      In Table \ref{table}, we have tabulated the changes in fidelity under various noise conditions. Additionally, we illustrate this variation in Figure \ref{f5} by plotting the fidelity against both the percentage increase in noise injected into the free segment and the noise parameter of the state in the noisy segment.
     
 \begin{table}
    \centering
     \begin{tabular}{ |p{1cm}||p{1cm}|p{1cm}|p{1cm}|p{1cm}|  }
 \hline
 \multicolumn{5}{|c|}{Percentage change in Fidelity} \\
 \hline
  \hline
 \multicolumn{1}{|@{}l||}{\backslashbox[1cm][l]{ \qquad p}{Noise \%}}& 2\% & 4\%& 6\%& 8\%\\
 \hline
 \hline
 \centering 0.70   &  0.27\%    &0.54\%&   0.80\% & 1.06\% \\
  \hline
 \centering0.75& 0.20\%    &0.40\%&   0.60\% & 0.80\% \\
  \hline
\centering 0.80 & 0.14\%    &0.29\%&   0.43\% & 0.58\% \\

  \hline
 \centering0.85 & 0.11\%    &0.20\%&   0.30\% & 0.39\% \\
  \hline
\end{tabular}
   \caption{{\bf Robustness of teleportation fidelity against noise in free and noisy segments.} When the free segment has a noisy state $\chi = (1-q) \ket{\psi}\bra{\psi} + q\frac{\mathbb{I}}{4}$, where $\ket{\psi}$ achieves the OFEF for the noisy segment $\rho(p) = (1-p)\ket{\phi^+}\bra{\phi^+} + p\ketbra{01}$, the end-to-end teleportation fidelity decreases due to the white noise factor $q$. This reduction, quantified by the percentage change $F\%=\frac{F^{*}_{AB}-F^{noisy}_{AB}}{F^{*}_{AB}}\times 100$, remains robust under increasing $q$. Interestingly, as the noise in the noisy segment ($p$) increases, $F\%$ decreases, indicating a mitigating effect of $p$ on fidelity loss.}
    \label{table}
\end{table}
We can also compare the scenario in which same amount of white noise is added to the maximally entangled state and a non-maximally entangled state in the free segment as earlier. For $q=0$ in the range given by Eq. (\ref{a26}) the difference between the fidelity of both these cases will be zero again courtesy Appendix \ref{AppendixB}. But in the case of noise injection with $q\neq 0$, there will be a finite difference between the fidelities.

Here, the state in the noisy segment is given by
         \begin{align}
            \rho(p)=(1-p)\ketbra{\Phi}+p\ketbra{01},
        \end{align}
        with  $\ket{\Phi}=\sqrt{\frac{3}{5}}\ket{00}+\sqrt{\frac{1}{5}}\ket{11}$ \textit{i.e} $\delta=\frac{3}{5}$.

 First, we consider the state in the free segment to be 
   \begin{align}
            \chi_{1}=(1-q)\ketbra{\Phi^{+}}+q\frac{\mathbb{I}}{4}.
        \end{align}
        
         The fidelity that we achieve between Alice and Bob will again not be optimal and is given as
       \begin{align}
            F^{ME}_{AB}=
            \frac{2 p^2 (2-q) (2 (1-q)+1)+p ((28-13 (1-q)) (1-q)+11)-9 q (1-q)+12}{10 (p (8 (1-q)+2)-3 (1-q)+3)}.\label{i25}
        \end{align}
The range of p and q for which this expression is valid is given by
        \begin{align}
            \frac{1}{2}\leq (1-q)\leq 1~\&\&~\frac{2 \sqrt{6} (1-q)+3 (1-q)-3}{2 \sqrt{6} (1-q)+8 (1-q)+2}<p<\frac{9 (1-q)^2+6 (1-q)-3}{4 (1-q)^2+6 (1-q)+2}
        \end{align}

        Next, we consider a case in which white noise is added to a non-maximally entangled state in the free segment, \textit{i.e.,}
        \begin{align}
            \chi_{2}=(1-q)\ketbra{\Phi'}+q\frac{\mathbb{I}}{4}.
        \end{align}
         The fidelity that we achieve between Alice and Bob is given in this case as
      \begin{align}
            F^{NME}_{AB}=&
            \frac{3 (p-1)^2 (11 p-6) (1-q)^3+5 (p-1) (p (5 p+24)-9) (1-q)^2}{10 (3 (p-1) (1-q)+2 p+3) (p (13 (1-q)+2)-3 (1-q)+3)}\nonumber\\
            &\qquad\qquad\qquad\qquad\qquad +\frac{(2 p+3) (p (14 p+57)-21) (1-q)+(p+4) (2 p+3)^2}{10 (3 (p-1) (1-q)+2 p+3) (p (13 (1-q)+2)-3 (1-q)+3)}.\label{i28}
        \end{align}

 The range of p and q for which this expression is valid is given as
       \begin{align}
           \frac{1}{2}\leq (1-q)\leq 1~\&\&~ \frac{3 \sqrt{2} (1-q)+3 (1-q)-3}{3 \sqrt{2} (1-q)+3 (1-q)+2}<p<\frac{-6 (1-q)^2-6 (1-q)+3}{4 (1-q)^2-11 (1-q)-2}
        \end{align}
 We again calculate the percentage change in the fidelity due to noise introduction as
        \begin{align}
            F\%=\frac{F^{ME}_{AB}-F^{NME}_{AB}}{F^{ME}_{AB}}\times100.\label{i22}
        \end{align}

 The range where both of these expressions are valid is given as
        \begin{align}
             &\left(\frac{1}{2}\leq (1-q)\leq \frac{3-3 \sqrt{2}+2 \sqrt{6}}{3+3 \sqrt{2}}~\&\&~ \frac{2 \sqrt{6} (1-q)+3 (1-q)-3}{2 \sqrt{6} (1-q)+8 (1-q)+2}<p<\frac{-6 (1-q)^2-6 (1-q)+3}{4 (1-q)^2-11 (1-q)-2}\right)\nonumber\\
            &||
   \left(\frac{3-3 \sqrt{2}+2 \sqrt{6}}{3+3 \sqrt{2}}<(1-q)\leq 1~\&\&~ \frac{3 \sqrt{2} (1-q)+3 (1-q)-3}{3 \sqrt{2} (1-q)+3 (1-q)+2}<p<\frac{-6 (1-q)^2-6 (1-q)+3}{4 (1-q)^2-11 (1-q)-2}\right)\label{i21}
        \end{align}

 Both the expressions in Eqs. (\ref{i25}) and (\ref{i28}) reduce to the optimal value
        \begin{align}
            F^*_{AB}=\frac{(2 p+3) (3 p+2)}{50 p},\label{h23}
        \end{align} 
        and the range becomes
        \begin{equation}
            \frac{3}{7} \left(5 \sqrt{2}-6\right)\leq p<1,\label{h24}
        \end{equation}
which are the same as the Eqs. (\ref{a25}) and (\ref{a26}) with $\delta=\frac{3}{5}, \alpha=\frac{3}{4}$ and $\beta=\frac{1}{2}$.

     We plot the quantity in Eq. (\ref{i22}) with both the percentage increase in the noise injected in the free segment and the noise parameter of the state in the noisy segment in Figure \ref{f6}.
\begin{figure}[h!]
         \begin{subfigure}{.5\textwidth}
              \centering
         \includegraphics[height=200px, width=250px]{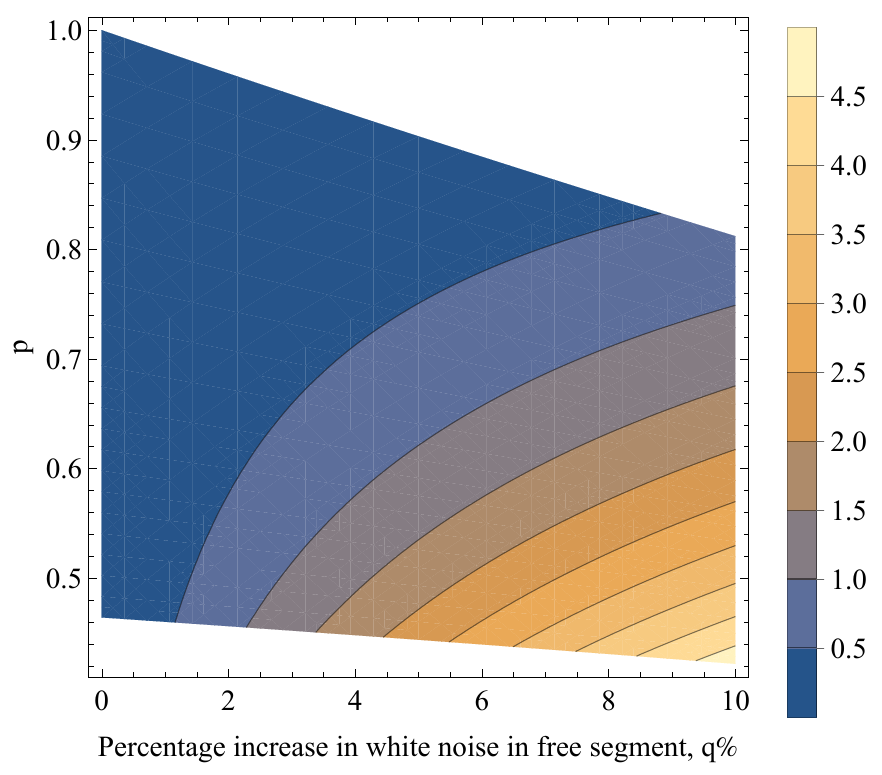}
         \caption{\centering}
         \label{f5}
         \end{subfigure}%
    \begin{subfigure}{.5\textwidth}
         \centering
         \includegraphics[height=200px, width=250px]{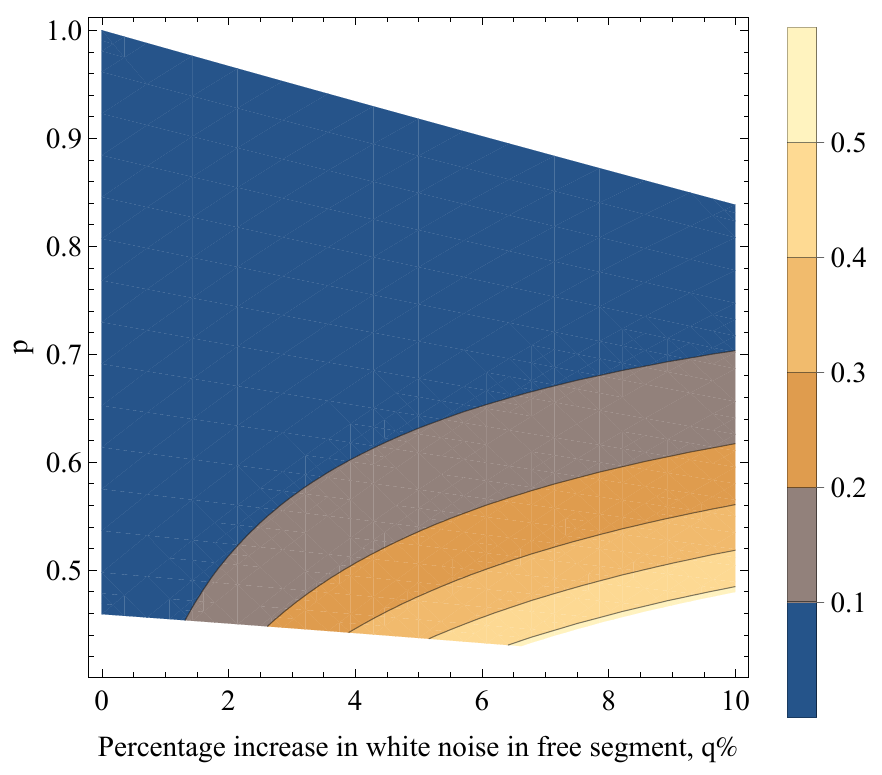}
         \caption{\centering}
         \label{f6}
    \end{subfigure}
    \caption{\textbf{(a)} Percentage change in fidelity, \( F\% = \frac{F^{*}_{AB} - F^{\text{noisy}}_{AB}}{F^{*}_{AB}} \times 100 \), as a function of the noise parameter \( p \) in the shared state within the noisy segment and the noise parameter \( q \) induced in the free segment. Notably, as \( p \) increases, the percentage change in fidelity decreases. Additionally, even for sufficiently high values of \( q \), the change in fidelity remains relatively low.\\ 
    \textbf{(b)} Percentage change in fidelity, \( F\% = \frac{F^{\text{ME}}_{AB} - F^{\text{NME}}_{AB}}{F^{\text{ME}}_{AB}} \times 100 \), when white noise (\( q \)) is added to the maximally entangled state compared to the non-maximally entangled state in the free segment, along with an increase in the noise parameter (\( p \)) of the state in the noisy segment. Similar phenomena are observed in the earlier plot. Both plots suggest that the noise in the noisy segment and the noise injected in the free segment are not arbitrary but are interrelated, as indicated by Eq. (\ref{i21}). Depending on the permissible tolerance for the fidelity change, a trade-off exists between these two noise parameters. }\label{BigF}
\end{figure}
\subsection*{\textbf{Case 2:} Photon-loss noise in the free segment}
 Next, we will showcase the robustness of our findings in the presence of a physically motivated noise model arising from the distribution of entanglement in a photonic setup that experiences undesirable photon losses. Consequently, the free segment is characterized as follows:
 \begin{align}
    \chi=(1-q)\ketbra{\Psi}+q\ketbra{00},\label{H8}
        \end{align}
        where $\ket{\Psi}=\sqrt{\frac{3}{4}}\ket{01}+\sqrt{\frac{1}{4}}\ket{10}$.
This can be thought of as a scenario where with probability $(1-q)$ we successfully distribute the entangled state, while with probability $q$ photon loss occurs. For $q=0$, this boils down to the scenario in Appendix \ref{AppendixB}.

The state in the noisy segment is of the similar form given in Eq. (\ref{h1}) and is expressed as
\begin{align}
            \rho(p)=(1-p)\ketbra{\Phi^{+}}+p\ketbra{01},
        \end{align}
  
         where $\ket{\Phi^{+}}=\frac{\ket{00}+\ket{11}}{\sqrt{2}}$.

         From Appendix \ref{AppendixB} this scenario also corresponds to $\delta=\frac{1}{2}$ and $\alpha=\frac{3}{4}$. We again perform Bell measurement at the node between Alice and Bob, thus we can choose $\beta=\frac{1}{2}$. Using these, the expression of optimal fidelity is the same as in Eq. (\ref{i4}).

For a noisy state in the free segment, the protocol that achieves optimality for \( q = 0 \), when applied to the noisy case, results in a fidelity given by:

        \begin{align}
            F^{q}_{AB}=\frac{p^3 ((19-9 (1-q)) (1-q)-8)+p^2 (1-q) (13 (1-q)-9)+p (8-3 (1-q) q)-(1-q)^2+(1-q)}{16 p (p (3 (1-q)-2)-2 (1-q)+2)}.
        \end{align}
  For the following range of $p$ and $q$ this above fidelity is valid and exceeds the classical bound:
        \begin{align}
            &\left(\frac{4}{4+\sqrt{3}}<(1-q)\leq \frac{6}{6+\sqrt{3}}~\&\&~\frac{\sqrt{3}}{6+\sqrt{3}}<p<-\frac{(1-q)}{7 (1-q)-8}\right)\nonumber\\
            &||\left(\frac{6}{6+\sqrt{3}}<(1-q)<1~\&\& ~\frac{\sqrt{3} (1-q)+4
   (1-q)-4}{\sqrt{3} (1-q)+6 (1-q)-4}<p<-\frac{(1-q)}{7 (1-q)-8}\right)\nonumber\\
   &|| \left(q=0~\&\&~ \frac{\sqrt{3}}{2+\sqrt{3}}\leq p<1\right)
        \end{align}
       Both of these reduce to Eq. (\ref{i4}) when $q=0$.
       
       To quantify the effect of noise, we focus on the percentage change in fidelity, which can be expressed mathematically as: 
       \begin{align}
             F\%=\frac{F^{*}_{AB}(\rho)-F^{q}_{AB}}{F^{*}_{AB}(\rho)}\times 100.
        \end{align}

        We plot this quantity with both the percentage increase in the noise injected in the free segment and the noise parameter of the state in the noisy segment in Figure$~$\ref{fig1}.\\



We also present the percentage change in the fully entangled fraction when the same noise level is applied to both maximally and non-maximally entangled states. Since they provide the same advantage in the noiseless case, our goal is to examine the robustness of this feature under noise. Specifically, for \( q = 0 \) within the range given by Eq. (\ref{a26}), the fidelity difference between these two cases is zero, as shown in Appendix \ref{AppendixB}. However, when noise is introduced (\( q \neq 0 \)), a finite difference between the fidelities emerges.

Here, the state in the noisy segment is given by
        \begin{align}
            \rho(p)=(1-p)\ketbra{\Phi}+p\ketbra{01},
        \end{align}
        with $\ket{\Phi}=\sqrt{\frac{3}{5}}\ket{00}+\sqrt{\frac{1}{5}}\ket{11}$ \textit{i.e.,} $\delta=\frac{3}{5}$.

 First, we consider the state in the free segment to be 
  \begin{align}
            \chi_1=(1-q)\ketbra{\Psi^{+}}+q\ketbra{00},
        \end{align}
         with $\ket{\Psi^{+}}=\frac{\ket{01}+\ket{10}}{\sqrt{2}}$.
         The fidelity that we achieve between Alice and Bob will again not be optimal and is given as
       \begin{align}
            F^{ME}_{AB}=
            \frac{-2 (p-1) (p (31 p+3)-9) (1-q)^2+(8 p-3) (p (19 p+12)-6) (1-q)-30 (p-1) p (2 p+3)}{50 p (p (11 (1-q)-6)-6 (1-q)+6)}.\label{i10}
        \end{align}
The range of p and q for which this expression is valid is given by
        \begin{align}
            \frac{4}{4+\sqrt{6}}<(1-q)\leq 1~\&\&~ \frac{\sqrt{6}}{5+\sqrt{6}}<p<-\frac{3 (1-q)}{7 (1-q)-10}
        \end{align}

        Next, we consider a case where photon-loss noise is added to a non-maximally entangled state in the free segment, \textit{i.e.,}
        \begin{align}
            \chi_{2}=(1-q)\ketbra{\Psi}+q\ketbra{00}.
        \end{align}
         The fidelity that we achieve between Alice and Bob will again not be optimal and is given in this case as
      \begin{align}
            F^{NME}_{AB}=
            \frac{-2 (p-1) (p (61 p-27)-9) (1-q)^2+(p (17 p (16 p-3)-114)+18) (1-q)-60 (p-1) p (2 p+3)}{50 p (17 p (1-q)-12 p-12 (1-q)+12)}.\label{i13}
        \end{align}

 The range of p and q for which this expression is valid is given as
       \begin{align}
            &\left (\frac {\sqrt {2}} {5 + \sqrt {2}} < 
    p < \frac {3 \sqrt {2}} {5 + 
        3 \sqrt {2}}~\&\&~\frac {20 p} {17 p + 3} < 
    (1-q) < \frac {12 p - 12} {3 \sqrt {2} p + 17 p - 3 \sqrt {2} - 
       12} \right)\nonumber\\
       &||\left (p = \frac {3
            \sqrt {2}} {5 + 
         3 \sqrt {2}}~\&\&~ \frac {60 \sqrt {2}} {15 + 60 \sqrt {2}} < 
     (1-q)\leq 1 \right)\nonumber\\
     &||\left (\frac {3 \sqrt {2}} {5 + 
      3 \sqrt {2}} < p < 1~\&\&~\frac {20 p} {17 p + 3} < 
   (1-q)\leq 1 \right)
        \end{align}
 We are interested in the quantity
        \begin{align}
            F\%=\frac{F^{ME}_{AB}-F^{NME}_{AB}}{F^{ME}_{AB}}\times 100,\label{i14}
        \end{align}
        which is the percentage change in the fidelity due to same noise introduction in a maximally entangled state vs a non-maximally entangled state in the free segment.

 But in order to calculate this quantity we have to find a range where both the expressions are valid.
   That range is given as
        \begin{align}
            &\left(\frac{4 \sqrt{6}}{3+4 \sqrt{6}}<(1-q)\leq -\frac{12}{-12-3 \sqrt{2}+\sqrt{6}}~\&\&~\frac{\sqrt{6}}{5+\sqrt{6}}<p<-\frac{3 (1-q)}{17 (1-q)-20}\right)\nonumber\\
            &||\left(-\frac{12}{-12-3
   \sqrt{2}+\sqrt{6}}<(1-q)<1~\&\&~ \frac{3 \sqrt{2} (1-q)+12 (1-q)-12}{3 \sqrt{2} (1-q)+17 (1-q)-12}<p<-\frac{3 (1-q)}{17 (1-q)-20}\right)\nonumber\\
   &||\left(q=0 ~\&\&~ \frac{3 \sqrt{2}}{5+3 \sqrt{2}}\leq
   p<1\right)\label{i36}
        \end{align}
        
       Both the expressions in Eqs. (\ref{i10}) and (\ref{i13}) reduce to the optimal value in Eqs. (\ref{h23}) and (\ref{h24}). We again plot the quantity in Eq. (\ref{i14}) both with the percentage increase in the noise injected in the free segment and the noise parameter of the state in the noisy segment in Figure \ref{fig2}.
  \begin{figure}[h!]
         \begin{subfigure}{.5\textwidth}
              \centering
         \includegraphics[height=200px, width=250px]{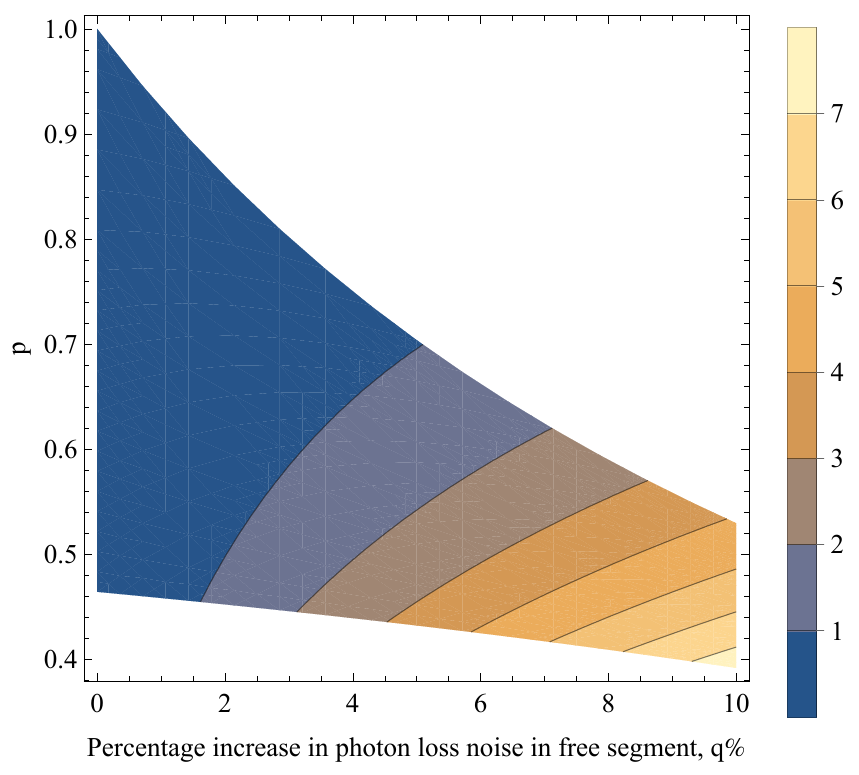}
         \caption{\centering}
         \label{fig1}
         \end{subfigure}%
    \begin{subfigure}{.5\textwidth}
         \centering
         \includegraphics[height=200px, width=250px]{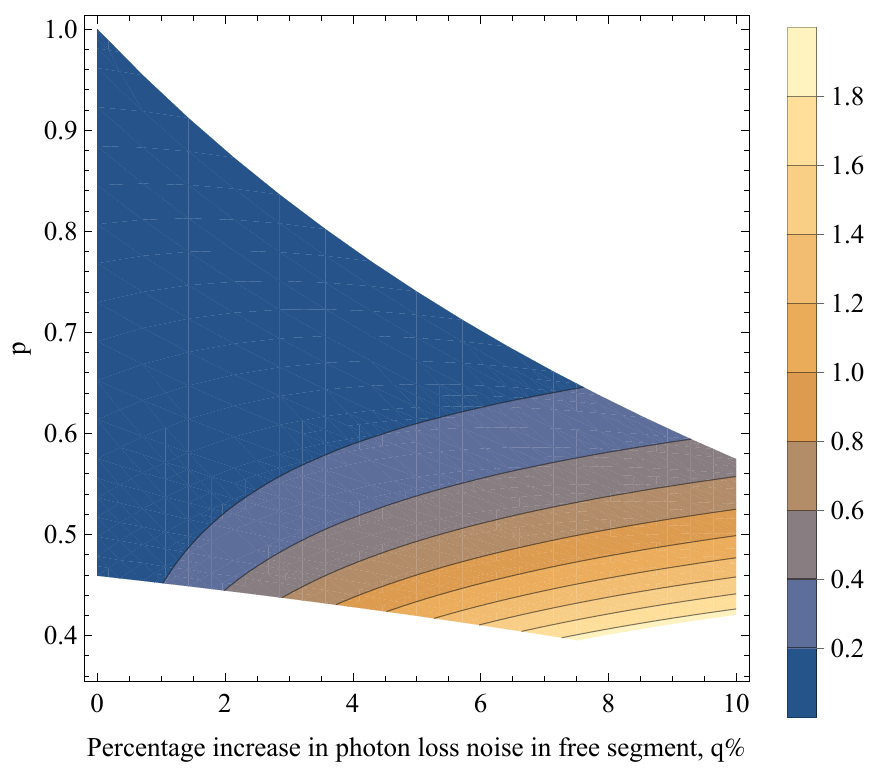}
         \caption{\centering}
         \label{fig2}
    \end{subfigure}
    \caption{\textbf{(a)} Percentage change in fidelity $ F\%=\frac{F^{*}_{AB}-F^{noisy}_{AB}}{F^{*}_{AB}}\times100$ with an increase in the noise parameter p of the state in the noisy segment and percentage increase of photon loss noise in the free segment. \textbf{(b)} Percentage change in fidelity $F\%=\frac{F^{ME}_{AB}-F^{NME}_{AB}}{F^{ME}_{AB}}\times 100$ when photon loss noise is added to the maximally entangled state as compared to the case when the same amount of photon loss noise is added to the non-maximally entangled state in the free segment along with an increase in the noise parameter p of the state in the noisy segment.  Again, both of these plots indicate that the amount of noise in the noisy segment and the amount of noise present in the free segment depend on each other as suggested by Eq. (\ref{i36}). Depending on the amount of change in the value of fidelity defined as permissible, there exists a trade-off between these two. Also comparing it with Figure \ref{BigF} we observe that fidelity is more robust under white noise injection as compared to the photon loss noise.}
\end{figure}

\subsection*{\textbf{Case 3}: Noisy Measurments}
We can consider a scenario where instead of a projective measurement in the computational basis we can have a POVM given as
\begin{align}
    &\ketbra{\Pi_0}=(1-\eta)\ketbra{0}+\eta\ketbra{1}\label{i32}\\
    &\ketbra{\Pi_1}=(1-\eta)\ketbra{1}+\eta\ketbra{0}\label{i33}
\end{align}
A measurement in the Bell basis can be thought of as the following sequence. First CNOT operation is applied on the desired two-qubit system with the first qubit as control and the second qubit as the target. Then, a $45^{\degree}$ rotation is locally applied on the first qubit. Consequently, the projections in Bell basis $\{\ketbra{\Phi^+},\ketbra{\Psi^+},\ketbra{\Phi^-},\ketbra{\Psi^-}\}$ corresponds to projections in computational basis $\{\ketbra{00},\ketbra{01},\ketbra{10},\ketbra{11}\}$. 
Thus according to Eqs. (\ref{i32}) and (\ref{i33}) these will also modify as 
\begin{align}
    &\ketbra{\Phi^+}\longrightarrow(1-\eta)^2\ketbra{\Phi^+}+\eta(1-\eta)\ketbra{\Psi^+}+\eta(1-\eta)\ketbra{\Phi^-}+\eta^2\ketbra{\Psi^-}\label{ii34}\\
    &\ketbra{\Psi^+}\longrightarrow(1-\eta)^2\ketbra{\Psi^+}+\eta(1-\eta)\ketbra{\Phi^+}+\eta(1-\eta)\ketbra{\Psi^-}+\eta^2\ketbra{\Phi^-}\label{ii35}\\
    &\ketbra{\Phi^-}\longrightarrow(1-\eta)^2\ketbra{\Phi^-}+\eta(1-\eta)\ketbra{\Phi^+}+\eta(1-\eta)\ketbra{\Psi^-}+\eta^2\ketbra{\Psi^+}\label{ii36}\\
    &\ketbra{\Psi^-}\longrightarrow(1-\eta)^2\ketbra{\Psi^-}+\eta(1-\eta)\ketbra{\Phi^+}+\eta(1-\eta)\ketbra{\Psi^-}+\eta^2\ketbra{\Phi^+}\label{ii37}
\end{align}
Considering the scenario where we have white noise in the free segment given by
\begin{align}
            \chi=(1-q)\ketbra{\Phi'}+q\frac{\mathbb{I}}{4},
        \end{align}
        where $\ket{\Phi'}=\sqrt{\frac{3}{5}}\ket{00}+\sqrt{\frac{1}{5}}\ket{11}$.

The state in the noisy segment is chosen to be
\begin{align}
            \rho(p)=(1-p)\ketbra{\Phi^{+}}+p\ketbra{01},
        \end{align}

        with the fixed value of the noise parameter $p=0.8$. From Eq. (\ref{a26}) the optimal fidelity for $\delta=\frac{1}{2}, \alpha=\frac{3}{5}$ and $\beta=\frac{1}{2}$ is 
        \begin{align}
            F^{*}_{AB}=\frac{81}{160},
        \end{align}

In the noisy case considered, the fidelity achieved will not be optimal and is given by 
\begin{align}
     &F^{noisy}_{AC}=
     \frac{-4 (2 (1-\eta)-1) (28 (1-\eta) (32 (1-\eta) (-\eta (1-\eta)+2)-59)+407) (1-q)^3}{160 ((28 (1-\eta)-11) (1-q)+15) ((56 (1-\eta)-37) (1-q)+45)}\nonumber\\
     &\qquad \qquad +\frac{(4 (1-\eta) (2135-32 (1-\eta) (-36 \eta (1-\eta)+79))-2165) (1-q)^2+2550 (2 (1-\eta)-1) (1-q)+3375}{160 ((28 (1-\eta)-11) (1-q)+15) ((56 (1-\eta)-37) (1-q)+45)}.
\end{align}
Whenever $\eta=0$ and $q=0$, this expression gives the optimal fidelity.
The range of validity in this case is given as
\begin{align}
    0.919<(1-\eta)\leq 1~\&\&  ~\frac{15}{-2 (1-\eta)+2 \sqrt{(1-\eta) ((1-\eta) (-96 \eta (1-\eta)+193)-130)+25}+1}<(1-q)\leq 1
\end{align}
 We again calculate the percentage change in fidelity as
 \begin{align}
            F\%=\frac{F^{*}_{AB}-F^{noisy}_{AC}}{F^{*}_{AB}}\times100.
        \end{align}

We contour plot $F\%$ with the percentage increase in noise in free segment $q\%$ and noise in the measurement $\eta\%$ in Figure \ref{f9}.\\

We can also consider the scenario where photon loss happens in the free segment
\begin{align}
            \chi=(1-q)\ketbra{\Phi}+q\ketbra{00},
        \end{align}
        where $\ket{\Phi}=\sqrt{\frac{3}{5}}\ket{01}+\sqrt{\frac{1}{5}}\ket{10}$.

        In this case,
        the fidelity achieved will not be optimal and is given by
\begin{align}
     &F^{noisy}_{AB}=\frac{1}{400} \Bigg(\frac{(8 \eta  (4 q-5)-16 q+45) (q (-3260 \eta +(4 \eta  (32 \eta  ((\eta -3) \eta +12)-365)+389) q+1930)+1425)}{((8 \eta -4) q-25) ((28 \eta
   -17) q-15)}\nonumber\\
   &\qquad\qquad+\frac{(8 \eta  (4 q-5)-16 q-5) (q (-4420 \eta +(4 \eta  (32 \eta  (3 (\eta -3) \eta +1)+105)-113) q+1310)+4275)}{((8 \eta -4) q+25) ((56 \eta -19)
   q-45)}\Bigg).
\end{align}
The range of validity of this expression is given as
\begin{align}
    0.919<(1-\eta)\leq 1\&\& \frac{5 (1-\eta) (7 (1-\eta)+1)}{(1-\eta) ((1-\eta) (-8 \eta (1-\eta)+51)-7)-2}<(1-q)\leq 1
\end{align}
We again contour plot $F\%$ with the percentage increase in noise in the free segment $q\%$ and noise in the measurement $\eta\%$ in Figure \ref{f10}.

\begin{figure}[h!]
         \begin{subfigure}{.5\textwidth}
             \centering
    \includegraphics[height=200px, width=250px]{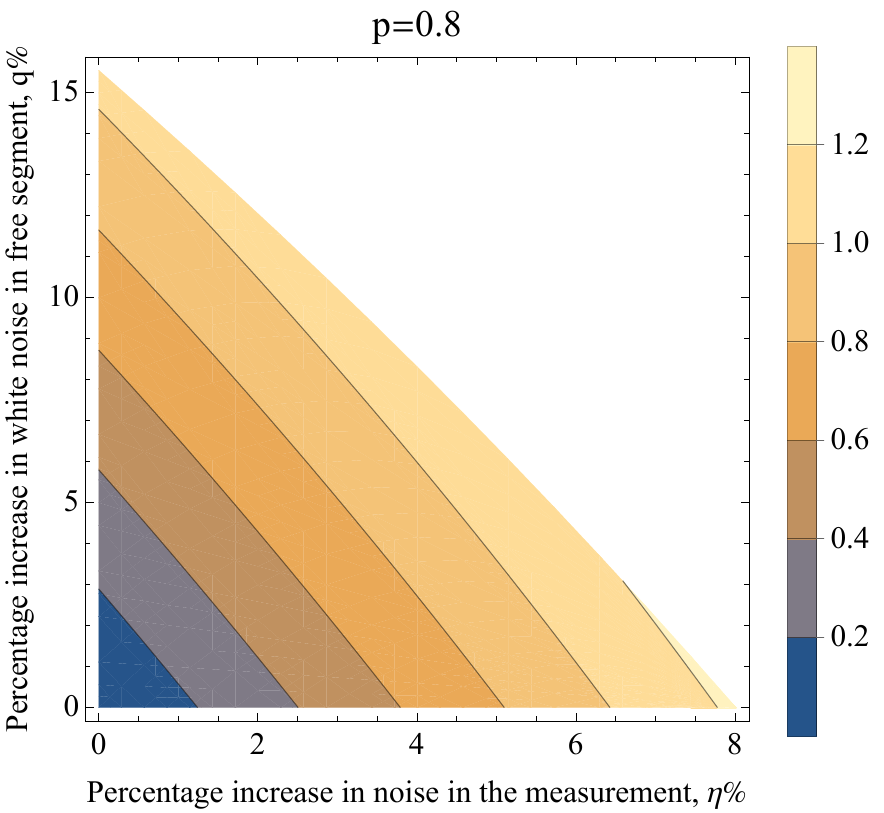}
    \caption{\centering}
    \label{f9}
         \end{subfigure}%
    \begin{subfigure}{.5\textwidth}
          \centering
    \includegraphics[height=200px, width=250px]{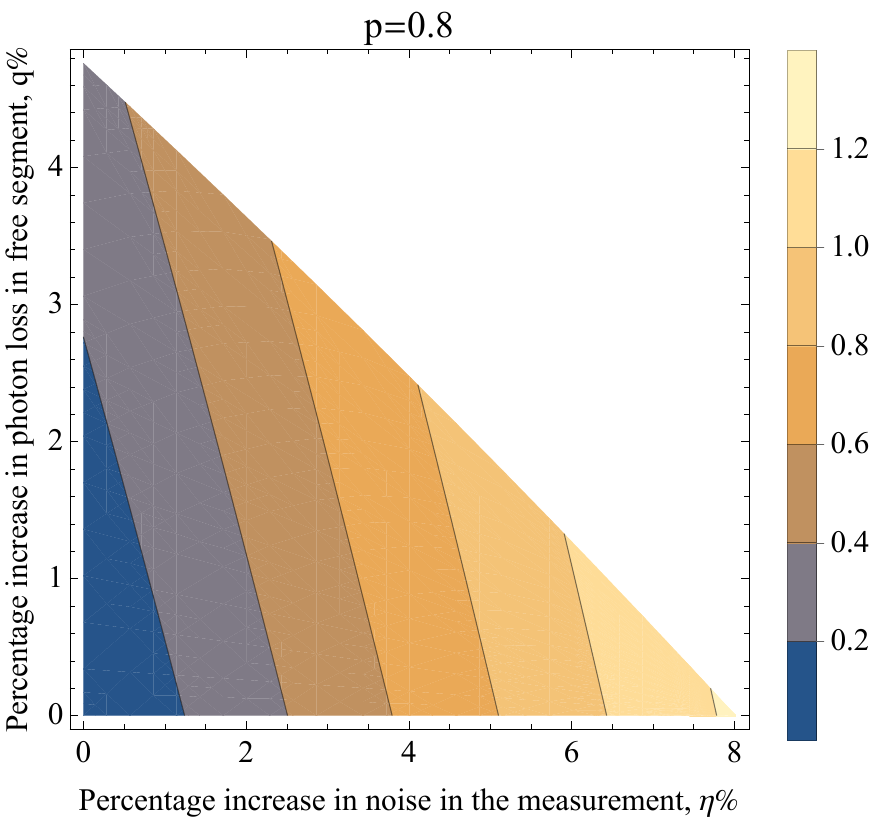}
    \caption{\centering}
    \label{f10}
    \end{subfigure}
    \caption{\textbf{(a)} Percentage change in fidelity $ F\%=\frac{F^{*}_{AB}-F^{noisy}_{AC}}{F^{*}_{AB}}\times100$ with an increase in the percentage of noise in the measurement and percentage of white noise added in the free segment. The noise parameter of the state in the noisy segment p is taken to be 0.8.  \textbf{(b)} Percentage change in fidelity  $ F\%=\frac{F^{*}_{AB}-F^{noisy}_{AC}}{F^{*}_{AB}}\times100$ with an increase in the percentage of noise in the measurement and percentage of colored noise added in the free segment. Again, the noise parameter p of the state in the noisy segment is 0.8. From both of these plots, we can infer that the noise in the measurement and the noise in the free segment constitute an interplay between them for a given amount of tolerance on the change in fidelity from the optimal value. This feature will be present for any allowed value of the noise parameter p in the noisy segment. }
\end{figure}

\section{Broader Context and Impact}\label{BCI}

In this work, we focus on optimizing teleportation fidelity in noisy environments, a key element for advancing long-distance quantum communication protocols. While our primary application remains long-distance communication, we extend our discussion to highlight the broader relevance of these results in quantum internet and quantum network, particularly in the contexts of quantum communication, quantum computing, and quantum cryptography.

\subsection{Teleportation-Based Quantum Computing}

Entanglement plays a central role not only in quantum communication but also in quantum computing, especially in measurement-based quantum computation (MBQC), where quantum gates are implemented through quantum measurements on an entangled resource state \cite{Gottesman1999,Leung2001,Nielsen2003,Robert2003,Leung2004,Jozsa2006}. A prominent variant of MBQC is teleportation-based quantum computation (TQC), which uses quantum teleportation for gate implementation \cite{Gottesman1999}. Recent advances in distributed quantum computing (DQC) have aimed to scale up quantum computations by enabling large-scale quantum processes across a network of quantum processing modules \cite{CALEFFI2024110672,Main2025}. Quantum teleportation provides a lossless method for communication across quantum modules, utilizing only bipartite entanglement shared between modules along with local operations and classical communication instead of direct transmission of quantum information.

While not directly explored in this work, it is hoped that the entanglement optimization strategies proposed here could provide an effective approach for reducing resource consumption in teleportation-based protocols, particularly in noisy environments. This may have implications for the scalability and efficiency of quantum computing architectures in the future.

\subsection{Universal Blind Quantum Computation}

Universal blind quantum computation (UBQC) allows a client to delegate quantum computation to a server while preserving both the privacy of the computation and the client's data  \cite{Broadbent2009}. Reliable quantum communication between the client and server is fundamental to UBQC, often realized through measurement-based quantum computing. Since our protocol optimizes quantum communication in noisy environments, it may potentially offer advantages in enhancing UBQC efficiency by reducing entanglement consumption. In practical implementations, this could contribute to more efficient and secure computation, although these benefits are to be explored in future work.

\subsection{ Applications in Quantum Cryptography}

Entanglement is a critical resource in quantum cryptography, particularly in entanglement-based quantum key distribution (QKD) protocols. The use of quantum repeaters to extend secure QKD over long distances has been well-established, but its practical implementation remains a challenge \cite{Yin2020}. Quantum teleportation, along with its variant quantum entanglement swapping, is pivotal for overcoming distance limitations in QKD, enabling secure transmission of quantum keys \cite{Hu2023}. It is anticipated that the optimization of entanglement consumption demonstrated in this work could lead to more efficient QKD protocols, especially in noisy environments, although further investigation into this potential is needed.

This approach could improve the scalability of QKD protocols, where minimizing entanglement consumption is crucial for practical implementations. The hope is that the insights gained here will inform future efforts in developing more resource-efficient quantum cryptography systems.

\subsection*{Summary}

The results of this study extend beyond long-distance quantum communication and suggest new avenues for future work in quantum communication, quantum computing, and quantum cryptography. Key potential advantages of the proposed protocol include: 

\begin{itemize}
    \item Achieving optimal teleportation fidelity in noisy environments.
    \item  Reducing entanglement consumption compared to conventional entanglement swapping methods.
    \item Have potential to enhance the efficiency of quantum communication across various platforms.
\end{itemize}
 
In summary, our work provides a framework for improving quantum communication performance with minimal entanglement usage. The findings point toward possible advancements in a variety of quantum information tasks, with future research exploring their broader applications. The proposed protocol may offer significant improvements in practical, noisy scenarios across multiple domains, ranging from distributed quantum computation to secure communication and quantum cryptography.

\end{document}